\documentclass[a4paper,UKenglish,cleveref, autoref, thm-restate]{lipics-v2019}

\usepackage{tikz}
\usetikzlibrary{calc}
\usepackage{csquotes}

\usepackage{mdframed}

\crefname{observation}{observation}{observations}
\Crefname{observation}{Observation}{Observations}
\newtheorem{observation}[theorem]{Observation}

\title{Completeness in the Polynomial Hierarchy for many natural Problems in Bilevel and Robust Optimization} 

\titlerunning{Completeness of many natural Problems in Bilevel and Robust Optimization} 

\author{Christoph Grüne}{Department of Computer Science, RWTH Aachen University, Germany}{gruene@algo.rwth-aachen.de}{https://orcid.org/0000-0002-7789-8870}{Funded by the German Research Foundation (DFG) – GRK 2236/2.}
\author{Lasse Wulf}{Section of Algorithms, Logic and Graphs, Technical University of Denmark, Kongens Lyngby, Denmark}{lawu@dtu.dk}{https://orcid.org/0000-0001-7139-4092}{Funded by the Austrian Science Fund (FWF):W1230 and by the Carlsberg Foundation CF21-0302 ``Graph Algorithms with Geometric Applications''.}

\authorrunning{C. Grüne and L. Wulf} 

\Copyright{Christoph Grüne and Lasse Wulf} 

\ccsdesc[500]{Theory of computation~Problems, reductions and completeness} 

\keywords{Computational Complexity, Bilevel Optimization, Robust Optimization, Stackelberg Games, Min-Max Regret, Most Vital Nodes, Most Vital Vertex, Most Vital Edges, Blocker Problems, Interdiction Problems, Two-Stage Problems, Polynomial Hierarchy, Sigma-2, Sigma-3} 

\category{} 

\relatedversion{} 

\supplement{}

\acknowledgements{We would like to thank James B. Orlin and the anonymous reviewers for their helpful feedback.}

\nolinenumbers 

\hideLIPIcs  

\EventEditors{John Q. Open and Joan R. Access}
\EventNoEds{2}
\EventLongTitle{42nd Conference on Very Important Topics (CVIT 2016)}
\EventShortTitle{CVIT 2016}
\EventAcronym{CVIT}
\EventYear{2016}
\EventDate{December 24--27, 2016}
\EventLocation{Little Whinging, United Kingdom}
\EventLogo{}
\SeriesVolume{42}
\ArticleNo{23}

\newcommand{\R}{\mathbb{R}}
\newcommand{\N}{\mathbb{N}}
\newcommand{\Z}{\mathbb{Z}}
\newcommand{\I}{\mathcal{I}}
\newcommand{\U}{\mathcal{U}}
\newcommand{\F}{\mathcal{F}}
\newcommand{\sol}{\mathcal{S}}
\newcommand{\powerset}[1]{2^{#1}}
\newcommand{\set}[1]{\{ #1 \}}
\newcommand{\fromto}[2]{\set{#1, \ldots, #2}}

\DeclareMathOperator{\poly}{poly}
\DeclareMathOperator{\reg}{reg}

\DeclareMathOperator*{\argmin}{arg\,min}
\newcommand{\leqSSP}{\leq_\text{SSP}}
\newcommand{\bin}{\set{0,1}}

\newcommand{\NP}{\textsl{NP}}

\begin{document}

\maketitle
\begin{abstract}
    In bilevel and robust optimization we are concerned with combinatorial min-max problems, for example from the areas of min-max regret robust optimization, network interdiction, most vital vertex problems, blocker problems, and two-stage adjustable robust optimization.
Even though these areas are well-researched for over two decades and one would naturally expect many (if not most) of the problems occurring in these areas to be complete for the classes $\Sigma^p_2$ or $\Sigma^p_3$ from the polynomial hierarchy, almost no hardness results in this regime are currently known.
However, such complexity insights are important, since they imply that no polynomial-sized integer program for these min-max problems exist, and hence conventional IP-based approaches fail.
We address this lack of knowledge by introducing at least 72 new $\Sigma^p_2$-complete and $\Sigma^p_3$-complete problems.
The majority of all earlier publications on $\Sigma^p_2$- and $\Sigma^p_3$-completeness in said areas are special cases of our meta-theorem. 
Precisely, we introduce a large list of problems for which the meta-theorem is applicable (including clique, vertex cover, knapsack, TSP, facility location and many more).
We show that for each of these problems, the corresponding min-max (i.e. interdiction/regret) variant is $\Sigma^p_2$- and the min-max-min (i.e. two-stage) variant is $\Sigma^p_3$-complete.    
\end{abstract}

\newpage

\section{Introduction}

In recent years, there has been enormous interest in the areas of Bilevel Optimization \cite{dempe2020bilevel, DBLP:journals/ejco/KleinertLLS21}, Robust Optimization \cite{DBLP:books/degruyter/Ben-TalGN09,robook}, Network Interdiction \cite{smith2013modern}, Stackelberg Games \cite{li2017review}, Attacker-Defender games \cite{HUNT2023}, and many other bilevel problems.
These research areas are vital in helping us to understand, which parts of a network are most vulnerable, prevent terrorist attacks, understand economic processes, and to understand and improve the robustness properties of many systems.

The common property of all these areas of research is that they study min-max optimization problems.
Equivalently, these problems can be characterized as an abstract two-turn game between two players.
The first player (from now on called Alice) starts the game and takes an action with the goal of minimizing some objective.
Afterwards, the second player (from now on called Bob) responds to Alice. Typically (but not always), Alice's and Bob's goals are opposite of each other.

Let us take as a prototypical example the maximum clique interdiction problem 
\cite{DBLP:journals/networks/PajouhBP14,DBLP:conf/iscopt/PaulusmaPR16, DBLP:journals/eor/FuriniLMS19, DBLP:journals/anor/Pajouh20}.
In this problem, we are given a graph $G = (V, E)$ and some budget $k \in \N$.
The goal of Bob is to find a clique of largest possible size in the graph.
However, before Bob's turn, Alice can delete up to $k$ vertices from the graph in order to impair Bob's objective. Hence the game is a min-max optimization problem described by
$$
    \min_{\substack{W \subseteq V\\ |W| \leq k}} \max \set{|C| : C \text{ is a clique in the graph } G \text{ with } W \cap C = \emptyset }.
$$
We remark that this problem follows a natural pattern, in which researchers often come up with new problems: 
First, some \emph{nominal problem} is taken (in this case, the maximum clique problem), and afterwards it is \emph{modified} into a more complicated min-max problem by adding an additional component (in this case, the possibility of Alice to interdict). The goal of this paper is to consider this very general pattern and shed some light on the understanding of its computational complexity. Roughly speaking, we show that whenever the nominal problem is already an NP-complete problem, then under some mild assumptions the complexity of the min-max variant will be significantly larger than the complexity of the nominal variant.
As a first example of this behavior, consider the following: It was proven by Rutenburg \cite{DBLP:journals/amai/Rutenburg93} that the maximum clique interdiction problem is complete for the complexity class $\Sigma^p_2$.

\paragraph*{The natural complexity class for bilevel problems}
In the year 1976, Stockmeyer \cite{DBLP:journals/tcs/Stockmeyer76} introduced the  \emph{polynomial hierarchy}, containing the complexity classes $\Sigma^p_k$ for all $k \in \N$.
While the complexity class $\Sigma^p_1 = \NP$ is very well known even to outsiders of theoretical computer science, the complexity class $\Sigma^p_2$ seems to be less known to non-specialists. 
However, for min-max optimization problems, it turns out that $\Sigma^p_2$ is the natural class to describe them. A formal definition of the classes $\Sigma^p_k$ for all $k \in \N$ is given in \cref{sec:prelim}.
Roughly speaking, the class $\Sigma^p_2$ contains all the problems of the form: 
Does there EXIST some object $x$, such that FOR ALL objects $y$ some easy-to-check property $P(x,y)$ holds? For example, the (decision version of the) maximum clique interdiction problem belongs to $\Sigma^p_2$. This is because it has objective value $t$ if and only if there EXISTS some set $W \subseteq V$ of size at most $k$, such that ALL cliques of size $t+1$ or more intersect $W$.
Hence the class $\Sigma^p_2$ naturally corresponds to decision variants of min-max optimization problems.
Analogously, there exists the class $\Sigma^p_3$, corresponding to min-max-min optimization problems, and so on.

Researchers are interested in the question which problems are complete for the class $\Sigma^p_2$, since this has several interesting consequences.
Similarly to the widely believed conjecture $P \neq \NP$, it is also widely believed that $\NP \neq \Sigma^p_2$. 
If this conjecture is true, it means that a $\Sigma^p_2$-complete or $\Sigma^p_3$-complete problem can not be described by a mixed integer program of polynomial size (because this would imply a polynomial-time reduction to MIP, which is NP-complete). 
This means that current existing integer program solvers, which have in some cases been very successful in tackling NP-complete problems, are not very successful at solving the harder case of $\Sigma^p_2$-problems. 
This theoretical difference can also be observed in practice.
For example, Woeginger \cite{DBLP:journals/4or/Woeginger21} points out as an example a bilinear program, which comes from social choice theory. Solving this program was posed to the Operations Research community as a challenge by Kurz and Napel \cite{DBLP:journals/ol/KurzN16}.
Despite the problem having very small input, the challenge is still widely open and seems out of reach of current methods. 
Woeginger argues that the mathematical optimization community in the past has proven greatly successful at tackling NP-complete problems through the  means of sophisticated IP-solvers. 
Today, the community is at a threshold, where new and more powerful tools and techniques need to be developed in order to tackle $\Sigma^p_2$-hard problems.
Similar observations were the motivation for a recent DAGSTUHL seminar on the topic of optimization at the second level.
In its report \cite{DBLP:journals/dagstuhl-reports/BrotcorneBHH22}, it is noted that \enquote{methodologies that have been developed for NP-complete problems over the last 50 years do not directly apply to robust and/or bilevel optimization problems}, and furthermore that \enquote{we will need to develop new techniques, new tricks,
new insights, new algorithms, and new theorems to get a grip on this area}. The goal of this paper is to address this problem by providing a novel and powerful tool, which can be used to shed light on the complexity landscape of bilevel (and more generally, multi-level) optimization. 

\subsection{Literature overview}
Due to space restrictions, it is impossible to consider every sub-area of bilevel optimization in a single paper.
In this paper, we restrict our attention to the three areas of \emph{network interdiction}, \emph{min-max regret robust optimization} and \emph{two-stage adjustable robust optimization}.
We choose these areas, since they represent three highly popular sub-areas of bilevel optimization, and since we can showcase a diverse set of techniques in the three different cases, while staying concise at the same time. 
In fact, since there are hundreds of papers in each of these popular areas, it is not possible to give a complete overview.
We split this literature survey into an overview of the general area, as well as an overview of work specifically on $\Sigma^p_2$- and $\Sigma^p_3$-completeness.

\paragraph*{General area}
In the area of \emph{network interdiction}, one is concerned with the question, which parts of a network are most vulnerable to attack or failure.
Such questions can be formulated as a min-max problem, and can be imagined as a game between Alice and Bob, where Alice interdicts a limited set of elements from the network in order to disturb Bob's objective.
Other names for slight variants of the network interdiction problem are the most vital nodes/most vital edges problem and the vertex/edge blocker problem. 
Network interdiction has been considered for a vast amount of standard optimization problems.
Among these are problems in P such as
shortest path  \cite{malik1989k,bar1998complexity,DBLP:journals/mst/KhachiyanBBEGRZ08},
matching \cite{DBLP:journals/dam/Zenklusen10a},
minimum spanning tree \cite{DBLP:journals/ipl/LinC93},
or maximum flow \cite{WOOD19931}.
These publications show the NP-hardness of the corresponding interdiction problem variant (except \cite{malik1989k} showing that the single most vital arc problem lies in P).
Furthermore, algorithms for interdiction problems that are NP-complete were developed, for example for
vertex covers \cite{DBLP:conf/iwoca/BazganTT10, DBLP:journals/dam/BazganTT11},
independent sets \cite{DBLP:conf/iwoca/BazganTT10, DBLP:journals/dam/BazganTT11,DBLP:journals/gc/BazganBPR15, DBLP:conf/tamc/PaulusmaPR17,DBLP:journals/dam/HoangLW23},
colorings \cite{DBLP:journals/gc/BazganBPR15, DBLP:conf/iscopt/PaulusmaPR16, DBLP:conf/tamc/PaulusmaPR17},
cliques \cite{DBLP:journals/networks/PajouhBP14,DBLP:conf/iscopt/PaulusmaPR16, DBLP:journals/eor/FuriniLMS19, DBLP:journals/anor/Pajouh20},
dominating sets \cite{DBLP:journals/eor/PajouhWBP15}, or
1- and $p$-center \cite{DBLP:conf/cocoa/BazganTV10, DBLP:journals/jco/BazganTV13},
1- and $p$-median \cite{DBLP:conf/cocoa/BazganTV10, DBLP:journals/jco/BazganTV13}.
Typically, the complexity of these problems is not analyzed beyond NP-hardness.
A general survey is provided by Smith, Prince and Geunes \cite{smith2013modern}.

In the area of \emph{min-max regret robust optimization}, one is faced with an uncertain cost scenario.
The decision-maker seeks to minimize the maximum deviation between incurred cost and optimal cost.
The min-max regret criterion is popular in robust optimization and has been considered for many standard optimization problems.
Among the problems, where the nominal problem is contained in P and are analyzed from a complexity viewpoint, are for example
shortest path \cite{DBLP:journals/dam/AverbakhL04},
minimum spanning tree \cite{DBLP:journals/dam/AverbakhL04},
min cut \cite{DBLP:journals/disopt/AissiBV08},
min $s$-$t$ cut \cite{DBLP:journals/disopt/AissiBV08},
1-center and 1-median with uncertainty in the node or edge weights \cite{DBLP:journals/informs/AverbakhB00,DBLP:journals/dam/Averbakh03},
and scheduling variants \cite{lebedev2006complexity,DBLP:journals/ol/Conde14}.
Among these, some interval min-max regret problems are shown to be in P, such as min cut, 1-center and 1-median.
The rest of the problems is shown to be NP-hard.
In the realm of interval min-max regret problems, for which the nominal problem is NP-complete, there are publications considering the complexity of the knapsack problem \cite{DBLP:journals/disopt/DeinekoW10}, the set cover problem \cite{DBLP:journals/coap/CocoSN22}, and the TSP and Steiner tree problem \cite{DBLP:journals/talg/GaneshMP23}.
Further but more general publications on this topic are \cite{kouvelis2013robust,kasperski2016robust}.
A general survey is provided by Aissi, Bazgan and Vanderpooten \cite{DBLP:journals/eor/AissiBV09}.

In the area of \emph{two-stage adjustable robust optimization}, one is faced with an uncertain cost scenario.
The decision-maker has to make a two-step decision, where in the first step a partial decision has to be made without full knowledge of the scenario (here-and-now), and in a second step a partial corrective decision can be made with full knowledge of the scenario (wait-and-see).
Two-stage adjustable optimization has been considered for a large amount of standard optimization problems, for example in \cite{DBLP:conf/or/KasperskiZ15,kasperski2016robust,DBLP:journals/dam/KasperskiZ17,DBLP:journals/eor/ChasseinGKZ18,DBLP:journals/eor/GoerigkLW22}.
A general survey is provided by Yan{\i}ko{\u{g}}lu, Gorissen, and den Hertog \cite{DBLP:journals/eor/YanikogluGH19}.

\paragraph*{Closely related work}
Despite a tremendous interest in bilevel optimization, and despite the complexity classes $\Sigma^p_k$ being the natural complexity classes for this type of decision problems, we are aware of only a handful of publications on the matter of $\Sigma^p_k$-completeness with respect to these areas.
Rutenburg \cite{DBLP:journals/amai/Rutenburg93} shows $\Sigma^p_2$-completeness of the maximum clique interdiction problem. 
Deineko and Woeginger \cite{DBLP:journals/disopt/DeinekoW10} show $\Sigma^p_2$-completeness of the min-max regret interval knapsack problem.
Caprara, Carvalho, Lodi and Woeginger \cite{DBLP:journals/siamjo/CapraraCLW14} show $\Sigma^p_2$-completeness of an interdiction variant of knapsack that was originally introduced by DeNegre \cite{10.5555/2231641}.
Fröhlich and Ruzika \cite{DBLP:journals/tcs/FrohlichR21} show the $\Sigma^p_2$-completeness of two versions of a location-interdiction problem.
Nabli, Carvalho and Hosteins \cite{DBLP:journals/jcss/NabliCH22} study the multilevel critical node problem and prove its $\Sigma^p_3$-completeness.
Coco, Santos and Noronha \cite{DBLP:journals/coap/CocoSN22} show $\Sigma^p_2$-completeness of the  min-max regret maximum benefit set covering problem.
Goerigk, Lendl and Wulf \cite{DBLP:journals/dam/GoerigkLW24} show $\Sigma^p_3$-completeness of two-stage adjustable variants of TSP, independent set, and vertex cover.
Grüne \cite{DBLP:conf/latin/Grune24} introduces a reduction framework for so-called universe gadget reductions to show $\Sigma^p_3$-completeness of recoverable robust versions of typical optimization problems.
Tomasaz, Carvalho, Cordone and Hosteins \cite{DBLP:journals/corr/abs-2406-01756} show $\Sigma^p_2$-completeness of an interdiction-knapsack problem and $\Sigma^p_3$-completeness of an fortification-interdiction-knapsack problem.
The compendium by Umans and Schaefer \cite{schaefer2008completeness} contains many $\Sigma^p_2$-complete problems, but few of them are related to bilevel optimization.
A seminal paper by Jeroslow \cite{DBLP:journals/mp/Jeroslow85} shows $\Sigma^p_k$-completeness for various classes of general multi-level programs.

The research that is most closely connected to this paper is the Ph.D. thesis of Johannes \cite{johannes2011new}.
She proves $\Sigma^p_2$-hardness for a variety of problems including “adversarial problems”, “partial inverse optimization problems” and several other types of classes.
She also discusses a transitive transformation for NP-completeness results that she calls “value preserving” \cite[Theorem 2.1.1]{johannes2011new}.
She shows that a wide range of adversarial problems and partial preprocessing problems can be proven to be $\Sigma^p_2$-hard by relying on their corresponding NP-completeness proofs, providing the NP-completeness transformation was “value preserving.”
In our paper, we establish a similar meta-theorem for other classes of problems, including  interdiction, min-max regret, and two-stage adjustable optimization.
     
Recently, we have learned that Johannes and Orlin [2020] have written an unpublished manuscript that extended the research in Johannes’s thesis.
They refer to certain NP-completeness transformations as \enquote{ratcheting transformations} because they can be used to prove hardness for extensions of these problems to problems higher in the polynomial time hierarchy.
There are additional overlaps of their work with our meta-theorem.

To the best of our knowledge, this completes the list of known $\Sigma^p_k$-completeness results for bilevel optimization.
We find it remarkable that so few results exist in this area, despite being an active research area for well over two decades.
In contrast, NP-completeness proofs are known for a huge number of problems.
One possible reason for this, according to Woeginger \cite{DBLP:journals/4or/Woeginger21}, is that \enquote{at the current moment, establishing $\Sigma^p_2$-completeness is usually tedious and mostly done via lengthy reductions that go all the way back to 2-Quantified Satisfiability}.

\subsection{Our contribution}

We make a large step towards understanding the complexity of bilevel combinatorial optimization, by introducing a general and powerful meta-theorem to prove $\Sigma^p_k$-completeness.
Our meta-theorem can be applied to such problems which are already NP-complete and have an additional property (explained below).
Using the meta-theorem, we can \enquote{upgrade} an existing NP-completeness proof to a $\Sigma^p_k$-completeness proof with very little work.
In other words, this means that for a lot of NP-complete problems, the $\Sigma^p_2$-completeness of its min-max variant follows essentially \enquote{for free} from its NP-completeness.
An analogous statement holds for min-max-min problems and $\Sigma^p_3$-completeness.
We apply our meta-theorem to the areas of bilevel optimization, in particular network interdiction, min-max regret robust optimization and two-stage adjustable optimization.
This way, we obtain at least 72 natural $\Sigma^p_2$- or $\Sigma^p_3$-complete problems relevant to these areas.
We remark that earlier papers showed $\Sigma^p_2$- or $\Sigma^p_3$-completeness only for one problem at a time, usually in a tedious fashion.
In contrast, our meta-theorem contains essentially all known $\Sigma^p_2$- or $\Sigma^p_3$-completeness results in the areas of network interdiction, min-max regret, and two-stage adjustable robust optimization as a special case (namely,
min-max-regret knapsack \cite{DBLP:journals/disopt/DeinekoW10},
min-max regret maximum benefit set cover \cite{DBLP:journals/coap/CocoSN22},
interdiction maximum clique \cite{DBLP:journals/amai/Rutenburg93},
interdiction knapsack \cite{DBLP:journals/siamjo/CapraraCLW14,DBLP:journals/corr/abs-2406-01756}, as well as
two-stage adjustable TSP, vertex cover, independent set \cite{DBLP:journals/dam/GoerigkLW24}).
While each of these earlier works proves $\Sigma^p_2$-hardness of only one specific problem, our meta-theorem proves $\Sigma^p_2$-completeness simultaneously for a large class of problems, including all problems mentioned in \cref{sec:ssp-reductions}, as well as possible further problems which are added in the future.

We call this class of problems for which the meta-theorem is applicable the class SSP-NP-complete (SSP-NPc), for reasons which are explained in \cref{sec:framework}.
We show that at least the following 25 classical problems are contained in the class SSP-NPc:
\begin{quote}
        Satisfiability,
        3-Satis\-fiability,
        Vertex Cover, Dominating Set, Set Cover, Hitting Set, Feedback Vertex Set, Feedback Arc Set, 
        Uncapacitated Facility Location, 
        p-Center, p-Median,
        Independent Set, Clique,
        Subset Sum, Knapsack, Partition, Scheduling,
        Directed/Undirected Hamiltonian Path, 
        Directed/Undirected Hamiltonian Cycle, 
        Traveling Salesman Problem,
        Two Directed Vertex Disjoint Path, 
        $k$-Vertex Directed Disjoint Path,
        Steiner Tree.
\end{quote}

A formal description of all these problems is provided in \Cref{sec:ssp-reductions}.
Furthermore, one can add new problems to the class SSP-NPc with very little work by finding a so-called \emph{SSP reduction} starting from any problem which is already contained in the class.
Since we could easily show for many classic problems that they are contained in SSP-NPc, we suspect that many more problems can be added in the future.
In fact, we observed the pattern that proving some problem to be contained in SSP-NPc is usually significantly easier than formulating a complex and technically challenging $\Sigma^p_k$-completeness proof.
This offers future researchers a convenient way to prove $\Sigma^p_k$-completeness of relevant problems. 

Since the class SSP-NPc contains so many well-known problems, our work shows that, in a certain sense, the \enquote{standard} or \enquote{normal} behavior of a NP-complete problem is to become $\Sigma^p_2$-complete when modified to be a min-max problem.
While this behavior is very intuitive, we are the first to be able to prove this intuition to be true for a very broad range of problems in robust and bilevel optimization.
Now we are ready to describe our main results.

\textbf{Network interdiction/blocker problems/most vital nodes or edges.}
In \cref{sec:interdiction}, we are concerned with the \emph{minimum cost blocker problem}. 
In the minimum cost blocker problem one is given a nominal problem $\Pi$ and additionally a cost function.
The question is for Alice to find a \emph{blocker} of minimum cost.
A blocker is a set which intersects every solution of the nominal problem.
(For example, every Hamiltonian cycle, every minimum vertex cover, every maximum clique, etc.)
Our main result is that for every problem $\Pi \in \text{SSP-NPc}$, the corresponding minimum cost blocker problem is $\Sigma^p_2$-complete.

\textbf{Min-max regret robust optimization with interval uncertainty.}
In \cref{sec:regret}, we are concerned with the \emph{min-max regret} robust optimization problem with interval uncertainty.
Faced with an uncertain cost function, the goal is to minimize the maximum deviation between incurred cost and optimal cost.
The uncertain cost is modelled by assuming that each cost coefficient stems from a pre-specified interval.
Our main result in this section is that for every problem $\Pi \in \text{SSP-NPc}$, the corresponding min-max regret robust optimization problem with interval uncertainty is $\Sigma^p_2$-complete. 

\textbf{Two-stage adjustable robust optimization}.
In \cref{sec:two-stage}, we are concerned with \emph{two-stage adjustable robust optimization} with discrete budgeted uncertainty (also called discrete $\Gamma$-uncertainty).
Roughly speaking, the class of two-stage problems models problems that are divisible into two stages:
The decision on the first stage has to be made without full information on the real instance (here-and-now), and the decision in the second stage is made after the uncertainty is revealed (wait-and-see).
Two-stage adjustable problems can be formulated via min-max-min expressions.
Therefore, the natural class for them is the class $\Sigma^p_3$.
Consequently, in \cref{sec:two-stage}, we demonstrate how our main idea can also be used to show completeness for the third stage (and potentially higher stages) of the polynomial hierarchy. Our main result in this section is that for every problem $\Pi \in \text{SSP-NPc}$, the corresponding two-stage adjustable problem with discrete budgeted uncertainty is $\Sigma^p_3$-complete.
\subsection{Technical Overview}
\label{sec:technicalOverview}
We give a short overview of the techniques and ideas used to obtain our main theorems.
For that purpose it becomes necessary to formally describe to what kind of optimization problems the meta-theorem applies.
In this paper, we consider \emph{linear optimization problems} (LOP).
Inspired by problems appearing in the literature, we define an LOP to be a problem expressed as a tuple $(\I, \U, \F, d, t)$.
Here, $\I \subseteq \bin^*$ is the set of input instances of the problem encoded as words in binary.
Associated to each input instance $I \in \I$, we assume that there is a \emph{universe} $\U(I)$, a linear \emph{cost function} $d^{(I)} : \U(I) \to \Z$ over the universe, the \emph{feasible sets} $\F(I) \subseteq \powerset{\U(I)}$ containing all feasible subsets of the universe, and a \emph{threshold} $t^{(I)}$.

Our ideas are best explained using an example.
Consider (the decision version of) the vertex cover problem.
For an instance $I$, we are given an undirected graph $G = (V,E)$, and some threshold $t^{(I)} \in \Z_{\geq 0}$.
The question is if there is a vertex cover of size at most $t^{(I)}$.
We can rephrase this question as $d^{(I)}(F) \leq t^{(I)}$, where $F$ is some vertex cover and $d^{(I)} = \bf{1}$ is the unit cost function.
Interpreting the vertex cover problem as an LOP in the above sense means the following:
The input instance is given as a tuple $I = (G, t)$ (encoded in binary).
The universe associated to some input instance $I$ is given by $\U(I) = V$, and the feasible sets are given by $\F(I) = \set{F \subseteq V : F \text{ is a vertex cover}}$. 
The cost function and threshold are given by $d^{(I)}$ and $t^{(I)}$. 

The decision question associated to the problem is to decide if there is a feasible set $F \in \F(I)$ such that its cost is below the threshold, that is $d^{(I)}(F) \leq t^{(I)}$.
(Note that this models a minimization problem, but maximization problems can be modeled as well by using negative coefficients.)
In general, we define the \emph{solution set} $\sol(I)$ as the set of all solutions of the instance $I$, that is
\begin{equation}
    \sol(I) := \set{F \in \F(I) : d^{(I)}(F) \leq t^{(I)}}. \label{eq:def-sol-LOP}
\end{equation}

Roughly stated, our main idea is now to show that many well-known NP-completeness proofs from some problem $\Pi_1$ to some other problem $\Pi_2$ have a special property, which we call the \emph{SSP property}.
On an intuitive level, this property states that the universe of $\Pi_1$ can be injectively embedded into the universe of $\Pi_2$ in such a way that the following two properties hold: 
(P1). Every solution of $\Pi_1$ corresponds to a partial solutions of $\Pi_2$, 
and (P2). every solution of $\Pi_2$ when restricted to the image of the embedding corresponds to a solution of $\Pi_1$.

We show that a surprisingly large amount of NP-completeness reductions which are known from the literature actually have the SSP property (\cref{sec:ssp-reductions}). Even the historically first NP-completeness reduction, i.e. the reduction used by Cook and Levin to show that \textsc{Satisfiability} is NP-complete has the SSP property (\cref{thm:cook-levin}).
We then proceed to show that every NP-completeness reduction with the SSP property can be upgraded to a $\Sigma^p_2$-completeness proof (or $\Sigma^p_3$-completeness proof, respectively) between the min-max variants (min-max-min variants, respectively) of the problems $\Pi_1$ and $\Pi_2$, hence proving the desired result of $\Sigma^p_2$-completeness ($\Sigma^p_3$-completeness, respectively).

The next step for us is to consider a slight generalisation of the concept of an LOP.
This has two reasons:
First, it turns out that for most of our arguments, we do not make explicit use of $\F, d$ and $t$.
We only make use of $\I$, $\U$, and $\sol$.
This means that we can abstract from these unnecessary details to have a cleaner argument.
The second reason is that there exist many optimization problems, which can be expressed as an LOP only in an awkward, non-natural way.
For example, consider the Hamiltonian cycle problem.
Even though it is possible to express the Hamiltonian cycle problem as an LOP, using trivial values for $d$ and $t$, this definition seems a bit unnatural.
For this reason, we consider the concept of so-called \emph{subset search problems (SSP)}.
An SSP is a problem $\Pi = (\I, \U, \sol)$, where $\I$ is the set of input instances, $\U(I)$ is the universe associated to each instance, and $\sol(I)$ is the set of solutions associated to each instance.
A formal definition is provided in \cref{def:SSP}.

Every LOP can be interpreted as an SSP in a straight-forward way, using \cref{eq:def-sol-LOP}.
Hence the vertex cover problem is an example of an SSP problem.
Another example of an SSP problem is the problem \textsc{3-Satisfiability}.
The input is some formula $\varphi$ with clauses $c_1, \dots, c_m$. The universe is the set $L = \fromto{\ell_1}{\ell_n} \cup \fromto{\overline \ell_1}{\overline \ell_n}$ of all literals.
The solution set of $\varphi$ is the set of all the subsets of the literals which encode a satisfying solution, i.e. 
$$
\sol(\varphi) = \set{L' \subseteq L : |L' \cap \set{\ell_i, \overline \ell_i}| = 1 \ \forall i \in \set{1,\dots,n}, L' \cap c_j \neq \emptyset \ \forall j \in \set{1, \dots, m}}.
$$

We are now ready to explain our main idea of SSP reductions. 
Our ideas are best explained with an example.
Consider the SSP problem \textsc{3Sat} with universe $\U = \fromto{\ell_1}{\ell_n} \cup \fromto{\overline \ell_1}{\overline \ell_n}$ (the literals) and the SSP problem \textsc{Vertex Cover} with universe $\U' = V$ (the vertices). 
We recall the classical NP-hardness reduction from \textsc{3Sat} to \textsc{Vertex Cover}, depicted in \Cref{fig:reduction:3sat-vertex-cover} from the book of Garey and Johnson \cite{DBLP:books/fm/GareyJ79}.
Given a \textsc{3Sat} instance consisting out of literals $\fromto{\ell_1}{\ell_n} \cup \fromto{\overline \ell_1}{\overline \ell_n}$ and clauses $C$, the reduction constructs a graph $G = (V,E)$ the following way:
The graph contains vertices $W := \fromto{v_{\ell_1}}{v_{\ell_n}} \cup \fromto{v_{\overline \ell_1}}{v_{\overline \ell_n}}$ such that each vertex $v_{\ell_i}$ is connected to vertex $v_{\overline \ell_i}$ with an edge.
Furthermore, for each clause, we add a new triangle to the graph, such that the three vertices of the triangle are connected to the corresponding vertices of the literals appearing in the clause.
The following is easily verified:
Every vertex cover of $G$ has size at least $|L|/2 + 2|C|$ and $G$ has a vertex cover of size $|L|/2 + 2|C|$ if and only if  the \textsc{3Sat} instance is a Yes-instance.

\tikzstyle{vertex}=[draw,circle,fill=black, minimum size=4pt,inner sep=0pt]
\tikzstyle{edge} = [draw,-]
\begin{figure}[thpb]
\centering
\resizebox{0.67\textwidth}{!}{
\begin{tikzpicture}[scale=1,auto]

\node[vertex] (x1) at (0,0) {}; \node[above] at (x1) {$v_{\ell_1}$};
\node[vertex] (notx1) at (2,0) {}; \node[above] at (notx1) {$v_{\overline \ell_1}$};
\draw[edge] (x1) to (notx1);

\node[vertex] (x2) at (4,0) {}; \node[above] at (x2) {$v_{\ell_2}$};
\node[vertex] (notx2) at (6,0) {}; \node[above] at (notx2) {$v_{\overline \ell_2}$};
\draw[edge] (x2) to (notx2);

\node[vertex] (x3) at (8,0) {}; \node[above] at (x3) {$v_{\ell_3}$};
\node[vertex] (notx3) at (10,0) {}; \node[above] at (notx3) {$v_{\overline \ell_3}$};
\draw[edge] (x3) to (notx3);

\node[vertex] (c1) at (4,-2.25) {}; \node[below] at (c1) {$v^{c_1}_{\overline \ell_1}$};
\node[vertex] (c2) at (5,-1.25) {}; \node[above left] at (c2) {$v^{c_1}_{\overline \ell_2}$};
\node[vertex] (c3) at (6,-2.25) {}; \node[below] at (c3) {$v^{c_1}_{\ell_3}$};
\node[] at (5,-1.92) {$c_1$};
\draw[edge] (c1) to (c2) to (c3) to (c1);
\draw[edge] (notx1) to (c1);
\draw[edge] (notx2) to (c2);
\draw[edge] (x3) to (c3);

\node at ($(x1)+(-1,0)$) {$W$};
\draw[dashed,rounded corners] ($(x1)+(-.5,+.7)$) rectangle ($(notx3) + (.5,-.4)$);

\end{tikzpicture}
}
\caption{Classic reduction of \textsc{3Sat} to \textsc{Vertex Cover} for $\varphi = (\overline \ell_1 \lor \overline \ell_2 \lor \ell_3)$.}
\label{fig:reduction:3sat-vertex-cover1}
\end{figure}

The above reduction is of course well-known.
However, we want to bring attention to the fact that this reduction has the SSP property.
This property is that the reduction maps the set of all solutions of the \textsc{3Sat} instance to the set of all solutions of the \textsc{Vertex Cover} instance in a one-to-one fashion.
More precisely, consider the set $W$. 
Every small vertex cover (of size $|L|/2 + 2|C|$) restricted to the set $W$ directly encodes a possible solution of the \textsc{3Sat} instance.
Conversely, for every single solution $\alpha$ of the \textsc{3Sat} instance, we can find a small vertex cover $S'$ (of size $|L|/2+2|C|$) such that $S' \cap W$ encodes $\alpha$.

We describe this one-to-one correspondence more formally.
Let $\sol(\varphi) \subseteq \powerset{\U}$ denote the solutions of \textsc{3Sat} (i.e. the set of all subsets of the literals which encode a satisfying assignment). Let $k := |L|/2 + 2|C|$ and  $\sol'(G, k) \subseteq \powerset{\U'}$ denote  the solutions of \textsc{Vertex Cover} (i.e. the set of all vertex covers of size at most $k$). 
We consider the injective function $f : \U \to \U'$ with $f(\ell_i) = v_{\ell_i}$ and $f(\overline \ell_i) = v_{\overline \ell_i}$. 
This function $f$ can be interpreted as a function which embeds the universe $\U$ into the universe $\U'$. It describes which literals in $\U$ correspond to which vertices in $\U'$. 
Note that the vertex subset $W = f(\U)$ is the image of $\U$.
Then the following holds:
For every satisfying assignment $S \in  \sol(\varphi)$, there exists at least one vertex cover $S' \in \sol'(G, k)$ such that $S' \cap W = f(S)$. This is property (P1).
Conversely, we also have that for every vertex cover $S' \in \sol'(G, k)$, the set $f^{-1}(S' \cap W)$ is contained in $\sol(\varphi)$. This is property (P2). It can be seen that 
$\text{(P1)} \land \text{(P2)}$ is equivalent to the set-based equation
\begin{equation}
    \set{f(S) : S \in \sol(\varphi) } = \set{S' \cap f(\U) : S' \in  \sol'(G, k)}. \label{eq:SSP}
\end{equation}

The above equation defines the SSP property. It turns out to be the key ingredient which is required for our meta-theorem.
We call a reduction with the SSP property an SSP reduction (compare \cref{def:ssp-reduction}).
We write $\Pi_1 \leqSSP \Pi_2$ for the fact that there exists an SSP reduction from $\Pi_1$ to $\Pi_2$.
We introduce the class SSP-NPc as an analogon to the class of NP-complete problems, but using polynomial-time SSP reductions instead of normal polynomial-time reductions.
\cref{sec:ssp-reductions} contains a list of problems in SSP-NPc.
In order to add a new problem $\Pi$ to the list, it suffices to prove $\Pi' \leqSSP \Pi$ for an arbitrary problem $\Pi' \in  \text{SSP-NPc}$.

This completes the description of the idea of SSP-reductions.
We note that another well-known variants of reductions exists in the literature, so-called \emph{parsimonious} reductions. On the first glance, our reductions seem similar to parsimonious reductions.
However, these two concepts are not the same, because parsimonious reductions map solutions to solutions bijectively, while our reductions $f : \U \to \U'$ map elements to elements injectively.

Finally, we explain how the idea of SSP reductions is used to obtain a meta-theorem.
The main idea is that the existence of an SSP reduction tells us that the two involved problems have a very similar solution structure.
Therefore, it suffices to prove $\Sigma^p_2$-completeness for the min-max variant of only one single problem in SSP-NPc (say \textsc{Sat}, for example), and then invest a little extra work to show that this $\Sigma^p_2$-completeness actually carries over to the min-max variant of all other problems in SSP-NPc.
For example, in \cref{sec:interdiction} on interdiction problems, we apply the following proof strategy:

First, we consider the interdiction variant only for the single problem $\Pi = \textsc{Sat}$.
We show using traditional techniques that Interdiction-\textsc{Sat} is $\Sigma^p_2$-complete (by a reduction from $\exists \forall$-\textsc{3Dnf-Sat}).
Next, we consider an arbitrary SSP-NP-complete problem $\Pi'$.
Since $\Pi'$ is SSP-NP-complete, there is a reduction $\textsc{Sat} \leqSSP \Pi'$.
This SSP reduction implies that given a \textsc{Sat} instance $I$, we can find a $\Pi'$ instance $I'$, such that $I$ can be imagined as a \enquote{sub-instance} of $I'$.
Specifically, there is an injective function $f$ mapping the universe of $\U(I)$ into the universe $\U'(I')$.
Furthermore, the topology of solutions is maintained (by \cref{eq:SSP}, or equivalently properties (P1) and (P2)).
Hence the $\textsc{Sat}$ instance $I$ can be imagined as a $\Pi'$ sub-instance of $I'$.
We show that this relation extends in such a way that the corresponding \textsc{Interdiction-Sat} instance can be imagined as a sub-instance of \textsc{Interdiction-$\Pi'$}. 
We can modify the costs of interdiction such that all newly elements that are part of $I'$, but not part of $I$ receive infinite costs.
On the other hand, all universe elements that are part of $I$ receive the same costs in $I'$.
Therefore all solutions of $I'$ are blocked by some blocker if and only if the sub-instance of \textsc{Sat} is blocked.
Since \textsc{Interdiction-Sat} $\Sigma^p_2$-hard, it follows that \textsc{Interdiction-$\Pi'$} is $\Sigma^p_2$-hard.

The results about min-max regret robust optimization in \cref{sec:regret} and two-stage adjustable robust optimization in \cref{sec:two-stage} follow essentially the same strategy.
We remark that in all three sections, we actually show $\Sigma^p_2$-completeness ($\Sigma^p_3$-completeness) of a more restricted problem than the original problem for all $\Pi \in \text{SSP-NPc}$.
These versions may be of independent interest, since they show that the considered problems are already hard even if the input parameters are more restricted.
We call these restricted versions \textsc{Combinatorial Interdiction}-$\Pi$ (\Cref{def:comb-interdiction}), \textsc{Restricted Interval Min-Max Regret}-$\Pi$ (\Cref{def:restricted-min-max-regret-problem}), and \textsc{Combinatorial Two-Stage Adjustable}-$\Pi$ (\Cref{def:comb-two-stage}).
\section{Preliminaries}
\label{sec:prelim}

A \emph{language} is a set $L\subseteq \bin^*$.
A language $L$ is contained in $\Sigma^p_k$ iff there exists some polynomial-time computable function $V$ (verifier), and $m_1,m_2,\ldots, m_k \leq \poly(|w|)$ such that for all $w \in \set{0,1}^*$
$$
    w \in L \ \Leftrightarrow \ \exists y_1 \in \set{0,1}^{m_1} \ \forall y_2 \in \set{0,1}^{m_2} \ldots \ Q y_k \in \set{0,1}^{m_k}: V(w,y_1,y_2,\ldots,y_k) = 1,
$$
where $Q = \exists$, if $k$ is odd, and $Q = \forall$, if $k$ is even.

An introduction to the polynomial hierarchy and the classes $\Sigma^p_k$ can be found in the book by Papadimitriou \cite{DBLP:books/daglib/0072413} or in the article by Jeroslow \cite{DBLP:journals/mp/Jeroslow85}.
An introduction specifically in the context of bilevel optimization can be found in the article of Woeginger \cite{DBLP:journals/4or/Woeginger21}.

A \emph{many-one-reduction} or \emph{Karp-reduction} from a language $L$ to a language $L'$ is a map $f : \bin^* \to \bin^*$ such that $w \in L$ iff $f(w) \in L'$ for all $w \in \bin^*$. 
A language $L$ is $\Sigma^p_k$-hard, if every $L' \in \Sigma^p_k$ can be reduced to $L$ with a polynomial-time many-one reduction. If $L$ is both $\Sigma^p_k$-hard and contained in $\Sigma^p_k$, it is $\Sigma^p_k$-complete.

A \emph{boolean variable} $x$ is a variable which takes one of the values 0 or 1. Let $X = \fromto{x_1}{x_n}$ be a set of variables. The corresponding \emph{literal set} is given by $L = \fromto{x_1}{x_n} \cup \fromto{\overline x_1}{\overline x_n}$. A \emph{clause} is a disjunction of literals. 
A boolean formula is in \emph{conjunctive normal form} (CNF) if it is a conjunction of clauses. It is in \emph{disjunctive normal form} (DNF), if it is a disjunction of conjunctions of literals.
In this paper, we use the notation where a clause is represented by a subset of the literals, and a CNF formula $\varphi$ is represented by a set of clauses. 
For example, the set $\set{\set{x_1, \overline x_2},\set{x_2, x_3}}$ corresponds to the formula $(x_1 \lor \overline x_2)\land (x_2 \lor x_3)$.
We write $\varphi(X)$ to indicate that formula $\varphi$ depends only on $X$.
Sometimes we are interested in cases where the variables are partitioned, we indicate this case by writing $\varphi(X,Y,\ldots)$.
An \emph{assignment} of variables is a map $\alpha : X \to \bin$.
The evaluation of the formula $\varphi$ under assignment $\alpha$ is denoted by $\varphi(\alpha) \in \bin$.
We also denote a partial assignment $\alpha$ on $X$ of a formula with partitioned variables $\varphi(X, Y)$ with $\varphi(\alpha, Y)$.
We further denote $\alpha \vDash \varphi$ if $\alpha$ satisfies formula $\varphi$.

For some cost function $c : U \to \R$, and some subset $U' \subseteq U$, we define the cost of the subset $U'$ as $c(U') := \sum_{u \in U'} c(u)$. For a map $f : A \to B$ and some subset $A' \subseteq A$, we define the image of the subset $A'$ as $f(A') = \set{f(a) : a \in A'}$. Even though these two definitions are ambiguous, in this paper it will always be clear from context, which definition applies.
\section{Framework}
\label{sec:framework}

The goal of this section is to introduce the class of \emph{SSP-NP-complete} problems, i.e. the class of all problems for which our meta-theorem is applicable.
As explained in \cref{sec:technicalOverview}, we first consider linear optimization problems (LOP problems) and then an abstraction of LOP problems, which we call SSP problems.
We then introduce the concept of an SSP reduction, and finally define the class SSP-NPc.

\subsection{Linear Optimization Problems}
It is important to remark that our meta-theorem cannot cover every single discrete optimization problem.
This is for two reasons: First, the set of all discrete optimization problems is incredibly diverse.
There seems to be no universally agreed upon definition of the term ``discrete optimization problem''.
Secondly, we need to assume a minimal amount of structure in order to meaningfully describe min-max variants of some problem. 
For this reason, we make the following assumption:
We assume that the problem in question has some \emph{universe} $\U$.
We also assume that there are some feasible solutions associated to the problem, that every feasible solution can be encoded purely as a subset of the universe, and it can be checked efficiently whether some proposed subset is a feasible solution.
Finally, we assume that there is some linear cost function on the universe, and the goal of the problem is to find a feasible solution of small cost.
Note that all of these assumptions are typical for discrete optimization problems.
In order to talk about the computational complexity of an LOP problem, we need to treat it as a decision problem. 
Therefore, we assume that the input contains some threshold, and the question is whether there is some feasible set whose cost is below the threshold.
Formally, this leads to the following definition.

\begin{definition}[Linear Optimization Problem]
\label{def:LOSPP}
    A linear optimization problem (or in short LOP problem)  $\Pi$ is a tuple $(\I, \U, \F, d, t)$, such that
    \begin{itemize}
        \item $\I \subseteq \{0,1\}^*$ is a language. We call $\I$ the set of instances of $\Pi$.
        \item To each instance $I \in \I$, there is some
        \begin{itemize}
            \item set $\U(I)$ which we call the universe associated to the instance $I$.
            \item set $\F(I) \subseteq \powerset{\U(I)}$ that we call the feasible solution set associated to the instance $I$. 
            \item function $d^{(I)}: \U(I) \rightarrow \Z$ mapping each universe element $e$ to its costs $d^{(I)}(e)$.
            \item threshold $t^{(I)} \in \Z$. 
        \end{itemize}
    \end{itemize}
    For $I \in \I$, we define the solution set $\sol(I) := \set{S \in \F(I) : d^{(I)}(S) \leq t^{(I)}}$ as the set of feasible solutions below the cost threshold. 
    The instance $I$ is a Yes-instance, if and only if $\sol(I) \neq \emptyset$.
    We assume (for LOP problems in NP) that it can be checked in polynomial time in $|I|$ whether some proposed set $F \subseteq \U(I)$ is feasible.
\end{definition}

The following are two examples of LOP problems:

\begin{description}
    \item[]\textsc{Traveling Salesman Problem}\hfill\\
    \textbf{Instances:} Complete graph $G = (V, E)$, weight function $w: E \xrightarrow{} \Z_{\geq 0}$, number $k \in \N$.\\
    \textbf{Universe:} Edge set $E =: \U$.\\
    \textbf{Feasible solution set:} The set of all TSP tours $T \subseteq E$.\\
    \textbf{Solution set:} The set of feasible $T$ with $w(T) \leq k$.

    \item[]\textsc{Vertex Cover}\hfill\\
    \textbf{Instances:} Graph $G = (V, E)$, number $k \in \N$.\\
    \textbf{Universe:} Vertex set $V =: \U$.\\
    \textbf{Feasible solution set:} The set of all vertex covers of $G$.\\
    \textbf{Solution set:} The set of all vertex covers of $G$ of size at most $k$.
\end{description}

Recall that a TSP tour is defined as a simple cycle traversing every vertex. Note that for the vertex cover problem, the function $d^{(I)}$ is the unit cost function. The threshold for both of these problems is the number $k$.
The two problems above are minimization problems, but we can model maximization problems in this framework by using negative cost functions.
(For example, for the knapsack problem, $d^{(I)}$ is equal to the negative profits.) 

\subsection{Introducing SSPs as an abstraction of LOPs}
\label{sec:ssp-reduction-intro}
As explained in \cref{sec:technicalOverview}, for most of the arguments in the following paragraphs, we do not really care about the set $\F(I)$ of feasible solutions, the cost function $d^{(I)}$, or the threshold $t^{(I)}$.
Rather, these are distracting details which we would like to get rid of.
For this reason, we introduce the concept of a \emph{subset search problem} (SSP).
An SSP problem is simply a tuple $\Pi = (\I, \U, \sol)$, where $\I$ is the set of instances, $\U(I)$ is the universe, and $\sol(I)$ is the solution set associated to each instance.
The name \enquote{subset search problem} stems from the fact, that we assume that each solution is encoded purely as a subset of the universe, and that the goal of the problem is to search for and find a solution.
The SSP concept has the advantage that it captures also such problems, which do not really fit into the LOP scheme.
For example, consider the Hamiltonian cycle problem.
It is in a certain sense unnatural to describe it using feasible sets, a cost function, and a cost threshold.
However, we can define it perfectly well as an SSP problem:
The instance is given by $I = (V, E)$ for some graph, the universe is $\U(I) = E$, and the solutions set is $\sol(I) = \set{T \subseteq E : T \text{ is a Hamiltonian cycle}}$.
Formally, we define an SSP the following way: 

\begin{definition}[Subset Search Problem (SSP)]
\label{def:SSP}
A subset search problem (or short SSP problem) $\Pi$ is a tuple $(\I, \U, \sol)$, such that
\begin{itemize}
    \item $\I \subseteq \set{0,1}^*$ is a language. We call $\I$ the set of instances of $\Pi$. 
    \item To each instance $I \in \I$, there is some set $\U(I)$ which we call the universe associated to the instance $I$. 
    \item To each instance $I \in \I$, there is some (potentially empty) set $\sol(I)\subseteq \powerset{\U(I)}$ which we call the solution set associated to the instance $I$.
\end{itemize}
\end{definition}

As a remark, every LOP problem becomes an SSP problem with the definition $\sol(I) := \set{S \in \F(I) : d^{(I)}(S) \leq t^{(I)}}$.
We call this the \emph{SSP problem derived from an LOP problem}.
An example of a natural SSP problem is the satisfiability problem:

\begin{description}
    \item[]\textsc{Satisfiability}\hfill\\
    \textbf{Instances:} Literal set $L = \fromto{\ell_1}{\ell_n} \cup \fromto{\overline \ell_1}{\overline \ell_n}$, clause set $C = \fromto{c_1}{c_m}$ such that $c_j \subseteq L$ for all $j \in \fromto{1}{m}.$\\
    \textbf{Universe:} $L =: \U$.\\
    \textbf{Solution set:} The set of all subsets $L' \subseteq \U$ of the literals such that for all $i \in \fromto{1}{n}$ we have $|L' \cap \set{\ell_i, \overline \ell_i}| = 1$, and such that $|L' \cap c_j| \geq 1$ for all clauses $c_j \in C$.
\end{description}

\begin{definition}
Let $\Pi = (\I, \U, \sol)$ be an SSP problem.
An instance $I \in \I$ is called Yes-instance, if $\sol(I) \neq \emptyset$.
The decision problem associated to the SSP problem $\Pi$ is the language $\set{I \in \I : \sol(I) \neq \emptyset}$ of all Yes-instances. 
\end{definition}

\subsection{A new type of reduction}
In this subsection, we introduce the concept of the \emph{SSP property}.
As explained in \cref{sec:technicalOverview}, the SSP property states that there is an injective embedding of one SSP problem into another, such that the solution sets of the two SSP problems correspond one-to-one to each other in a strict fashion.
An \emph{SSP reduction} is a usual many-one reduction which additionally has the SSP property.
By formalizing the intuition gained in \cref{sec:technicalOverview}, we obtain the following definition.
Note that in this definition the function $g$ corresponds to the standard many-one reduction, while the functions $(f_I)_{I \in \I}$ are the injective embedding functions corresponding to the SSP property (analogous to \cref{eq:SSP}).
Note that since $\U(I)$ can be different for every instance $I$, we have that $(f_I)_{I \in \I}$ is a family of functions, and not a single function. 

\begin{definition}[SSP Reduction]
\label{def:ssp-reduction}
    Let $\Pi = (\I,\U,\sol)$ and $\Pi' = (\I',\U',\sol')$ be two SSP problems. We say that there is an SSP reduction from $\Pi$ to $\Pi'$, and write $\Pi \leqSSP \Pi'$, if
    \begin{itemize}
        \item There exists a function $g : \{0,1\}^* \to \{0,1\}^*$ computable in polynomial time in the input size $|I|$, such that $I$ is a Yes-instance iff $g(I)$ is a Yes-instance (i.e. $\sol(I) \neq \emptyset$ iff $\sol'(g(I)) \neq \emptyset$).
        \item There exist functions $(f_I)_{I \in \I}$ computable in polynomial time in $|I|$ such that for all instances $I \in \I$, we have that $f_I : \U(I) \to \U'(g(I))$ is an injective function mapping from the universe of the instance $I$ to the universe of the instance $g(I)$ such that 
        $$
            \set{f_I(S) : S \in \sol(x) } = \set{S' \cap f_I(\U(I)) : S' \in  \sol'(g(I))}.
        $$
    \end{itemize}
\end{definition}

An example of an SSP reduction from \textsc{3-Satisfiability} to \textsc{Vertex Cover} was shown in \cref{sec:technicalOverview}.
Many more examples of SSP reductions are shown in \cref{sec:ssp-reductions}.
A schematic description how the mapping $f_I$ of an SSP reduction between SSP problems $\Pi$ and $\Pi'$ applies for a specific instance $I$ of $\Pi$ is depicted in \Cref{fig:SSP-reduction:universe-embedding}.
Next, we show that SSP reductions are transitive, which enables us to easily show reductions between a multitude of problems.

\begin{figure}[!ht]
        \centering
        \resizebox{0.42\textwidth}{!}{
            \begin{tikzpicture}
                \draw[] ($(-5,0)$) rectangle ($(-2,-2)$);
        		\node[] () at (-4.5,-1.7) {$\U(I)$};

                \draw[dashed] ($(2,0)$) rectangle ($(5,-2)$);
                \draw[] ($(0,0.5)$) rectangle ($(5.5,-3)$);
        		\node[] () at (2.8,-1.7) {$f_I(\U(I))$};
        		\node[] () at (0.8,-2.7) {$\U'(g(I))$};

                \draw[->, thick] (-2,-0.7) -- (2,-0.7);
                \node[above] () at (-1,-0.7) {$f_I$};
            \end{tikzpicture}
        }
        \caption{
            The relation between the universes by applying an SSP reduction between the problem $\Pi = (\I, \U, \sol)$ and $\Pi' = (\I', \U', \sol')$ for a given instance $I \in \I$.
            Let $I \in \I$ be that instance of $\Pi$, then the SSP reduction $(g, (f_I)_{I \in \I})$ maps the universe $\U(I)$ into the universe $\U'(g(I))$ of problem $\Pi'$ such that $f_I(\U(I)) \subseteq \U'(g(I))$.
            Note that $g(I)$ is the instance defined by the usual reduction mapping $g$.
            The function $f_I$ maintains a one-to-one correspondence between the elements of $\U(I)$ and $f_I(\U(I))$.
        }
        \label{fig:SSP-reduction:universe-embedding}
    \end{figure}

\begin{lemma}
\label{lem:SSP-transitive}
    SSP reductions are transitive, i.e. for SSP problems $\Pi_1, \Pi_2, \Pi_3$ with $\Pi_1 \leqSSP \Pi_2$ and $\Pi_2 \leqSSP \Pi_3$, it holds that $\Pi_1 \leqSSP \Pi_3$.
\end{lemma}
\begin{proof}
Consider for $i = 1,2,3$ the three SSP problems $\Pi_i = (\I_i, \U_i, \sol_i)$.
There is an SSP reduction $(g_1,f_1)$ from $\Pi_1$ to $\Pi_2$ and an SSP reduction $(g_2, f_2)$ from $\Pi_2$ to $\Pi_3$.
We describe an SSP reduction from $\Pi_1$ to $\Pi_3$.
We require a tuple $(g, (f_I)_{I \in \I})$.
For the first function $g$, we set $g := g_2 \circ g_1$.
This suffices since $\sol_1(I) \neq \emptyset \Leftrightarrow \sol_2(g_1(I)) \neq \emptyset \Leftrightarrow \sol_3((g_2 \circ g_1)(I)) \neq \emptyset$ and $g$ is poly-time computable.
Let $I_1 := I$ be the initial instance of $\Pi_1$, $I_2 := g_1(I_1)$ be the instance of $\Pi_2$ and $I_3 := g_2(I_2)$ be the instance of $\Pi_3$.

For the second function $f_I$, we define for each instance $I \in \I_1$ the map $f_I := (f_2)_{I_2} \circ (f_1)_{I}$.
Observe that for each instance $I \in \I_1$, the function $f_I$ is injective and maps to $\U_3(I_3)$ and is poly-time computable.
It remains to show that $f$ has the desired SSP property.
In order to reduce the notation, we omit the subscript in $f$ (i.e. $f = f_2 \circ f_1$).
We also write $\U_1, \U_2, \U_3$ instead of $\U_1(I_1), \U_2(I_2), \U_3(I_3)$ and $\sol_1,\sol_2, \sol_3$ instead of  $\sol_1(I_1),\sol_2(I_2), \sol_3(I_3)$.
Since $\Pi_1 \leqSSP \Pi_2$ and $\Pi_2 \leqSSP \Pi_3$ and since for injective functions it holds that $f(A \cap B) = f(A) \cap f(B)$, we have
\begin{align*}
    \set{f(S_1) : S_1 \in \sol_1} = &\ \set{f_2(f_1(S_1)) : S_1 \in \sol_1}\\
    = &\ \set{f_2(Y) : Y \in \set{f_1(S_1) : S_1 \in \sol_1}}\\
    = &\ \set{f_2(Y) : Y \in \set{S_2 \cap f_1(\U_1) : S_2 \in \sol_2}}\\
    = &\ \set{f_2(S_2) \cap f_2(f_1(\U_2)) : S_2 \in \sol_2}\\
    = &\ \set{S_3 \cap f_2(\U_2) \cap f_2(f_1(\U_1)) : S_3 \in \sol_3}\\
    = &\ \set{S_3 \cap f(\U_1) : S_3 \in \sol_3}.
\end{align*}
\end{proof}

Further, we define the class SSP-NP, which is the analogue of NP restricted to SSP problems.

\begin{definition}[SSP-NP]
    The class SSP-NP consists out of all the SSP problems which are polynomial-time verifiable.
    Formally, an SSP problem $\Pi = (\I, \U, \sol)$ belongs to SSP-NP, if $|\U(I)| = \poly(|I|)$ and if there is an algorithm receiving tuples of an instance $I \in \I$ and a subset $S \subseteq \U(I)$ as input and decides in time polynomial in $|I|$, whether $S \in \sol(I)$.
\end{definition}

For the remainder of the paper, in a slight abuse of notation, let us say that $\text{SSP-NP} \subseteq \text{NP}$.
Note that this is not formally completely correct, since the class NP is a set of languages, while the class SSP-NP is a set of SSP problems.
However, we can say that some SSP problem $\Pi$ is in NP, if the corresponding decision problem $\{ I \in \mathcal{I} \mid \sol(I) \neq \emptyset \}$ is in NP.

Because of the analogous definition of the class SSP-NP to NP and the natural SSP adaptation of \textsc{Satisfiability}, we are able to adapt the theorem of Cook and Levin \cite{DBLP:conf/stoc/Cook71,DBLP:journals/annals/Trakhtenbrot84} to the class SSP-NP and show that \textsc{Satisfiability} is the canonical SSP-NP-complete problem.

\begin{theorem}[Cook-Levin Theorem Adapted to SSPs]
\label{thm:cook-levin}
\textsc{Satisfiability} is SSP-NP-complete with respect to polynomial-time SSP reductions, i.e. for every SSP problem $\Pi$ contained in SSP-NP, we have  $\Pi \leqSSP \textsc{Satisfiability}$.
\end{theorem}
\begin{proof}
    We consider the original proof by Cook and show that it is actually a polynomial-time SSP reduction.
    Let $\Pi = (\I, \U, \sol)$ be an arbitrary problem in SSP-NP with universe $\U(I)$ and solution set $\sol(I)$ associated to each instance $I \in \I$ of $\Pi$.
    We have to show that $\Pi \leqSSP \textsc{Satisfiability}$. 
    Recall that $\textsc{Satisfiability} = (\I', \U', \sol')$, where $\I'$ is the set of \textsc{Sat}-instances, and for each formula $\varphi \in \I'$, the set $\U'(\varphi)$ is its literal set, and $\sol(\varphi)$ is the set of literal sets corresponding to satisfying assignments.

    Since $\Pi$ is in SSP-NP, there exists a deterministic Turing machine $M$, such that given as input some tuple $(I, S)$ with $I \in \I$ and $S \subseteq \U(I)$ (encoded in binary), the Turing machine $M$ decides in polynomially many steps (say at most $|I|^k$ for some $k$),  whether $S \in \sol(I)$. 
    Here we use the equivalent definition of the class NP in terms of verifiers and in terms of nondeterministic Turing machines \cite{DBLP:books/daglib/0023084}. 
    Now, the proof of Cook implies that there exists a CNF-formula $\varphi(Y,Z)$ with the following properties:
    \begin{itemize}
        \item   The formula has size polynomial in $|I|$ and can be constructed in polynomial time from $I$.
        \item   The variables are split into two parts $Y, Z$.
                Here, the variables $Z$ encode in binary the input $S \subseteq \U(I)$ of the Turing machine $M$.
                The number of these variables is $|Z| = |\U(I)|$.
                The set $Y$ contains all other variables.
        \item   The partial formula \enquote{$\varphi(Y,S)$} is satisfiable if and only if $M$ accepts $(I,S)$.
                More formally, for all assignments $\alpha : Z \to \bin$, we let $S_\alpha$ be the corresponding subset of $\U(I)$ (defined by letting $\U(I) = \set{u_1,\dots,u_m}$ and $Z = \set{z_1,\dots,z_m}$ and considering the binary encoding $u_i \in S_\alpha$ iff $\alpha(z_i) = 1$ for $i = 1,\dots,m$).
                Furthermore, we let $\varphi(Y, \alpha)$ be the formula where the $Z$-variables are assigned by $\alpha$, and the $Y$-variables are still free.
                Then we have for all $\alpha : Z \to \bin:$
                $$
                    \varphi(Y, \alpha) \text{ is satisfiable } \Leftrightarrow \ M \text{ accepts } (I,S_\alpha) \text{ after at most $|I|^k$ steps}.
                $$
    \end{itemize}
    
    We now claim that this reduction by Cook immediately yields a polynomial-time SSP reduction $(g, (f_I)_{I \in \I})$.
    Formally, we let $g(I) := \varphi(Y,Z)$.
    Note that by the properties of Cook's reduction, $I$ is a Yes-instance of $\Pi$ if and only if $\exists Y,Z \ \varphi(Y,Z)$ is satisfiable.
    Hence this is a correct many-to-one reduction.
    For the SSP property, we define $f_I(u_i) := z_i$ for all $i \in \fromto{1}{|\U(I)|}$.
    Informally speaking, the universe element $u_i$ is mapped to the positive literal $z_i \in \U(\varphi)$ which encodes in binary in the input to $M$, whether the element $u_i$ is included in $S$.
    It now follows from the above equivalence that this is an SSP reduction:
    If $S \in \sol(I)$, then $M$ accepts $(I, S)$ after $|I|^k$ steps, and for the corresponding assignment $f_I(S)$ it holds that it can be completed to a satisfying assignment of $\varphi$.
    On the other hand, every satisfying assignment $S' \in \sol(\varphi)$ restricted to the positive literal set $Z = f_I(\U(I))$ encodes a set $S = f_I^{-1}(S' \cap Z)$ such that $S \in \sol(I)$.
    This proves the SSP property and hence $\Pi \leqSSP \textsc{Satisfiability}$.
\end{proof}

With \cref{thm:cook-levin} in mind, we define the class of SSP-NP-complete problems (SSP-NPc) as the set of all SSP-NP problems that are complete for the class SSP-NP with respect to SSP reductions.

\begin{definition}
    The class of SSP-NP-complete problems is called SSP-NPc and consists of all $\Pi \in$ SSP-NP such that $\textsc{Satisfiability} \leqSSP \Pi$. 
\end{definition}
\section{Interdiction Problems}
\label{sec:interdiction}

In this section, we consider the closely related topics of \emph{minimum cost blocker problems}, \emph{most vital vertex/edge problems} and \emph{interdiction problems}.
All of these problems are slight variants of each other.
Formally, we consider the following problem.

\begin{definition}[Interdiction Problem]
\label{def:interdiction}
    Let an SSP problem $\Pi = (\I, \U, \sol)$ be given. The interdiction problem associated to $\Pi$ is denoted by $\textsc{Interdiction-}\Pi$ and defined as follows: The input is an instance $I \in \I$ together with a cost function $c : \U(I) \to \Z$ and a threshold $t \in \Z$.
    The question is whether
    $$
        \exists B \subseteq \U(I) \ \text{with} \ c(B) \leq t : \forall S \in \sol(I):  B \cap S \neq \emptyset.
    $$
\end{definition}

The main result of this section is that $\textsc{Interdiction-}\Pi$ is $\Sigma^p_2$-complete for all SSP-NP-complete problems $\Pi$ (\cref{thm:main-result-interdiction}). We also show in \cref{sec:combinatorial-interdiction} the $\Sigma^p_2$-completeness of a more restricted version of interdiction, which could be of independent interest.

We make a few remarks regarding \cref{def:interdiction}:
First, note that the interdiction variant of all LOP problems can be defined over the SSP problem derived from it as described in \Cref{sec:framework}.
The set $\sol(I)$ then contains all the sets the attacker wants to block.
For example in the vertex cover problem, $\sol(I)$ contains all vertex covers of instance $I$ of size smaller or equal to $t$.
(W.l.o.g. due to the properties of the reductions studied in this paper, we can for all problems from \cref{sec:ssp-reductions} assume that $t$ is chosen to be the optimal threshold.)
Second, note that in our problem \textsc{Interdiction-$\Pi$}, the underlying instance is not changed.
In particular, we do not delete elements of the universe (e.g. vertices of the graph).
For some problems, there might be a subtle difference between deleting elements and forbidding elements to be in the solution.
(The vertex cover problem is one such example:
It makes a big difference of deleting a vertex $v$ and its incident edges, or forbidding that $v$ is contained in the solution, but still having the requirement that all edges incident to $v$ get covered by the vertex cover. In the first case, the vertex cover interdiction problem stays in NP, hence we can not hope to obtain a general $\Sigma^p_2$-completeness result.
In this paper we only consider the second case.)

As the third remark, we note that in the literature, usually the minimum cost blocker problem and the most vital nodes/edges problem are slightly differently defined:
The minimum cost blocker problem asks for a minimum cost blocker which decreases the objective value by a set amount. On the other hand, the most vital nodes asks for the maximum value by which the objective can be decreased, when given a certain cost budget for the blocker.  
The interdiction problem, which we formulated here as a decision problem enables us to capture the $\Sigma^p_2$-completeness of both these variants.
It follows by standard arguments that both of the above problems become $\Sigma^p_2$-complete in our setting.

The last remark is that $\textsc{Interdiction-}\Pi$ can be understood as a game between Alice ($\exists$-player, trying to find a blocker) and Bob ($\forall$-player, trying to find a solution).
Note that this could be considered different from other typical robust optimization problems, where the $\exists$-player tries to find a solution.
In the remainder of this section, we locate the complexity of $\textsc{Interdiction-}\Pi$ exactly.
The easy part is to show containment in $\Sigma^p_2$.

\begin{lemma}\label{thm:interdictionInSigma2}
    If $\Pi = (\I, \U, \sol)$ is a problem in SSP-\NP, then $\textsc{Interdiction-}\Pi$ is in $\Sigma^p_2$.
\end{lemma}
\begin{proof}
    We provide a polynomial time algorithm $V$ such that for $m_1, m_2 \leq poly(|I|)$:
    $$
        I \in L \ \Leftrightarrow \ \exists y_1 \in \{0,1\}^{m_1} \ \forall y_2 \in \{0,1\}^{m_2} : V(I, y_1, y_2) = 1.
    $$
    
    With the $\exists$-quantified $y_1$, we encode the blocker $B \subseteq \U(I)$.
    The encoding size of $y_1$ is polynomially bounded in the input size of $\Pi$ because $|\U(I)| \leq poly(|I|)$.
    Next, we encode the solution $S \in \sol(I)$ to the nominal problem $\Pi$ using the $\forall$-quantified $y_2$ within polynomial space.
    This is doable because the problem $\Pi$ is in NP.
    At last, the verifier $V$ has to verify the correctness of the given tuple $(B, S)$ provided by the $\exists$-quantified $y_1$ and $\forall$-quantified $y_2$.
    Checking whether $c(B) \leq t$ and $B \cap S \neq \emptyset$ is trivial and checking whether $S \in \sol(I)$ is clearly in polynomial time because $\Pi$ is in SSP-NP.
    It follows that $\textsc{Interdiction-}\Pi$ is in $\Sigma^p_2$.
\end{proof}

\subsection{Combinatorial Interdiction Problems}
\label{sec:combinatorial-interdiction}
In the literature, often the cost version from above is analyzed and used to model real-world problems.
However, for us it proves to be helpful to introduce a more restricted variant of  \textsc{Interdiction-$\Pi$}, which we call the \emph{combinatorial version} of \textsc{Interdiction-$\Pi$}.
This combinatorial version is slightly more specific than the cost version and allows for a more precise reduction in the SSP framework.
We show the $\Sigma^p_2$-hardness of the combinatorial version of \textsc{Interdiction-$\Pi$} for all SSP-NP-complete problems.
In the end, we adapt this result to the cost version. This shows the $\Sigma^p_2$-hardness of all interdiction problems, for which the nominal problem is SSP-NP-complete.

\begin{definition}[Combinatorial Interdiction Problem]
\label{def:comb-interdiction}
    Let an SSP problem $\Pi = (\I, \U, \sol)$ be given.
    The combinatorial interdiction problem associated to $\Pi$ is denoted by \textsc{Comb.} \textsc{Interdiction-}$\Pi$ and defined as follows:
    The input is an instance $I \in \I$ together with a set of blockable elements $B \subseteq \U(I)$ and a threshold $t \in \Z$.
    The question is whether
    $$
        \exists B' \subseteq B \ \text{with} \ |B'| \leq t : \forall S \in \sol(I): B' \cap S \neq \emptyset.
    $$
\end{definition}

The difference between the combinatorial version and the cost version is that we ignore the costs of the elements and introduce a set of possibly blockable elements.
We use the canonical \textsc{Satisfiability} problem as the starting point for our meta reduction.
Therefore, we apply \cref{def:comb-interdiction} to $\Pi = \textsc{Satisfiability}$ yielding the following:

\begin{definition}[\textsc{Combinatorial Interdiction-Satisfiability}]
We denote the combinatorial interdiction version of \textsc{Sat} by \textsc{Comb. Interdiction-Sat}.
The input is a CNF with clauses $C$ and literals $L$, a set blockable literals $B \subseteq L$ and a threshold $t$.
The question is whether there is a set $B' \subseteq B$ with $|B'| \leq t$ such that for all $S \in \sol(I)$, we have $B' \cap S \neq \emptyset$.
(In other words, there is no satisfying assignment whose literals are completely disjoint from $B'$.)
\end{definition}

\subsection{A Meta-Reduction for Combinatorial Interdiction Problems}
For the beginning of our meta-reduction, we prove that the canonical \textsc{Satisfiability} problem \textsc{Combinatorial Interdiction-Satisfiability} is $\Sigma^p_2$-complete.

\begin{lemma}
\label{lem:interdiction-sat-sigma-2-complete}
    \textsc{Combinatorial Interdiction-Satisfiability} is $\Sigma^p_2$-complete.
\end{lemma}
\begin{proof}
    Analogously to \Cref{thm:interdictionInSigma2}, \textsc{Comb. Interdiction-Sat} is in $\Sigma^p_2$.
    As the basis of our hardness proof, we use the problem \textsc{$\exists\forall$DNF-Sat}, which is $\Sigma^p_2$-hard as shown by Stockmeyer \cite{DBLP:journals/tcs/Stockmeyer76}.
    In this problem, we are given a SAT formula $\varphi(X,Y)$ in disjunctive normal form (DNF), such that its variables are partitioned into two parts $X,Y$. 
    The question is whether there is a variable assignment for $X$ such that for all variable assignments for $Y$ we have $\varphi(X,Y) = 1$.
    We reduce \textsc{$\exists\forall$DNF-Sat} to \textsc{Comb. Interdiction-Sat}.
    Let $\exists X \forall Y \varphi(X, Y)$ be the \textsc{$\exists\forall$DNF-Sat} instance.
    We transform this instance into an equivalent instance $(\psi, B, t)$ of \textsc{Comb. Interdiction-Sat}.
    More precisely, this means that $\exists X \forall Y \varphi(X, Y)$ if and only if there is a set $B' \subseteq B$ with $|B'| \leq t$ such that there is no solution $S \in \sol(\exists X' \psi(X'))$ with $B \cap S = \emptyset$.

    We use an idea of Goerigk, Lendl and Wulf \cite[Theorem 1]{DBLP:journals/dam/GoerigkLW24}. We quickly sketch the main idea: Let $n := |X|$, i.e. $X = \{x_1, \dots, x_n \}$. \textsc{Comb. Interdiction-Sat} can be understood as a game between Alice and Bob, where Alice selects a blocker $B' \subseteq B$ and Bob tries to select a satisfying assignment avoiding $B'$. How can we model the formula $\exists X \forall Y \varphi(X, Y)$ with this game? The idea is that Alice's choice of a blocker $B'$ should correspond to the $\exists X$ stage and Bob's choice of assignment should correspond to the  $\forall Y$ stage. How can we encode an assignment of the $X$-variables in terms of a blocker $B'$? The idea is to introduce new variables $X^t = \{x^t_1, \dots, x^t_n\}$ and $X^f = \{x^f_1, \dots, x^f_n\}$. We say that Alice plays honestly, if $|B' \cap \set{x^t_i, x^f_i}| = 1$ for all $i = 1, \dots, n$. In the other case, i.e. $B'$ contains both $x^t_i, x^f_i$ for some $i$, we say that Alice cheats. We will add some \enquote{cheat-detection} gadgets to the formula, which make sure that if Alice cheats, then Bob can trivially win the game. We make sure that  the cheat-detection gadget can be used if and only if Alice cheats. If Alice plays honestly, note that $x^t_i \in B'$ enables Bob to choose $x^f_i$ as part of his solution. This corresponds to the assignment $\alpha(x) = 0$. On the other hand, $x^f_i \in B'$ corresponds to $\alpha(x_i) = 1$. We are now ready to give the formal reduction. 

    \begin{description}
        \item[Definition of the instance.]
        Given an instance $\varphi$ of $\exists \forall$DNF-SAT, the instance  $(\psi, B, t)$ of combinatorial interdiction SAT is defined as follows: We start by considering the auxiliary formula $\varphi'(X, Y) := \neg \varphi(X, Y)$. Note that $\varphi'$ is in CNF by De Morgan's law. We furthermore have
        $$
            \exists X \forall Y \varphi(X, Y) \ \leftrightarrow \ \exists X \neg \exists Y \varphi'(X, Y). 
        $$
        A new formula $\varphi''$ is created from $\varphi'$ in terms of a substitution process: We introduce $2n$ new variables $X^t = \{x^t_1, \dots, x^t_n\}$ and $X^f = \{x^f_1, \dots, x^f_n\}$. For all $i = 1,\dots, n$, we substitute each occurrence of some literal $x_i$ by the positive literal $x_i^t$. We substitute each occurrence of some literal $\overline x_i$ by the positive literal $x_i^f$. All other literals are kept the same. A new formula $\varphi'''$ is created from $\varphi''$ by introducing a new variable $s$ and appending the positive literal $s$ to every clause, that is
        $$
            \varphi''' \equiv \varphi'' \lor s.
        $$
        Finally, we introduce $n$ new variables $\set{s_1, \dots s_n}$. We let $Z = \set{s} \cup \set{s_1, \dots, s_n}$ and define the formula $\psi$ by
        $$
            \psi(X^t, X^f, Y, Z) = \varphi'''(X^t, X^f, Y, s) \land \left( \bigwedge\limits_{i=1}^n (x^t_i \lor \overline s_i)\land (x^f_i \lor \overline s_i) \right) \land (\overline s \lor s_1 \lor s_2 \lor \dots \lor s_n).
        $$
        We remark that the newly added elements between $\varphi''$ and $\psi$ form the cheat-detection gadget.
        Finally, we define the set of blockable literals by $B := X^t \cup X^f$ and the number of blockable literals by $t := n$. This completes the description of the instance $(\psi, B, t)$.
        \item[Correctness]
        Note that Alice can only block the positive literals of $X^t \cup X^f$ because by definition of the comb.\ interdiction problem we have $B' \subseteq B = X^t \cup X^f$.
        Furthermore, we claim that in an optimal game, Alice has to play honestly.
        To prove this, consider the case where both literals $x^t_i$ and $x^f_i$ are blocked or less than $t$ literals in total are blocked by Alice.
        In both cases there is a $j \neq i$ such that Alice blocks neither $x^t_j$ nor $x^f_j$ (due to $|B'| \leq t$ and the pigeonhole principle).
        Hence Bob is able to take both literals $x^t_j$ and $x^f_j$ into his solution.
        This enables Bob to also take both $s_j$ and $s$ into the solution without violating the constraints of $\psi$.
        For all other $p \neq j$, Bob can take the literals $\overline s_p$ into his solution.
        In this case $\psi$ is trivially satisfied.
        We conclude that Alice has to play honestly, i.e. $|B' \cap \set{x^t_i, x^f_i}| = 1$ for all $i = 1,\dots, n$.
        
        Hence for a fixed honest choice of Alice, we obtain a fixed chosen assignment $\alpha(X) \to \bin^{|X|}$ which Bob is forced to take. 
        (More precisely, Bob is forced to take $\overline s, \overline s_1,\dots, \overline s_n$ into his solution. 
        Then the cheat-detection clauses are trivially verified. 
        The remaining formula does not contain any negative literal $\overline x_i^f, \overline x^t_i$, 
        so we can w.l.o.g.\ assume that under optimal play Bob takes the one from the two positive literals $x^t_i, x^f_i$ that is not blocked by Alice.)
        Following the above restriction, Alices' goal is described by the formula  $\neg \exists Y, Z \ \psi(X^t, X^f, Y, Z)$.
        This in the end is equivalent to the formula $\neg \exists Y \varphi'(\alpha, Y)$, since for honest behavior of Alice, the only way to satisfy the formula is to set $s, s_1, \dots, s_n$ to false.
        We conclude that if $\varphi$ is a yes-instance of $\exists \forall$DNF-SAT, then Alice can play honestly and win the game, hence we have a Yes-instance of \textsc{Comb. Interdiction-Sat}.

        On the other hand, assume that the described tuple $(\psi, B, t)$ is a Yes-instance of \textsc{Comb. Interdiction-Sat}.
        By the previous argument, since Alice was able to win, she must have played honestly.
        Then, only one of $x^t_i$ and $x^f_i$ can be in the blocker $B'$ for all $i \in \fromto{1}{|X|}$.
        Thus, there is a blocker $B' \subseteq B$ fixing the assignment on $X^t \cup X^f$ such that there is no solution to $\exists Y, Z \ \psi(X^t, X^f, Y, Z)$.
        This is equivalent to fixing the assignment on $X^t \cup X^f$, such that $\neg \exists Y, Z \ \psi(X^t, X^f, Y, Z)$.
        By transforming $\psi(X^t, X^f, Y, Z)$ back to $\neg \varphi(X, Y)$, we get that there is an assignment to $X$ such that $\forall Y \varphi(X, Y)$, which is equivalent to $\exists X \forall Y \varphi(X, Y)$, which is a Yes-instance for \textsc{$\exists\forall$DNF-Sat}.
        \item[Polynomial Time] All transformations are doable in polynomial time because only a polynomial number of additional variables as well as clauses are added to the formula.
    \end{description}
\end{proof}

With the $\Sigma^p_2$-hardness of \textsc{Comb. Interdiction-Sat} established, we provide a meta-reduction to all \textsc{Comb. Interdiction-$\Pi$}, if there is an SSP reduction between the nominal \textsc{Satisfiability} and the nominal $\Pi$.
In other words, we prove the $\Sigma^p_2$-hardness of \textsc{Comb. Interdiction-$\Pi$}.

\begin{theorem}
\label{thm:interdictionSigma2Hard}
For all SSP-NP-complete problems $\Pi$, the combinatorial interdiction variant \textsc{Comb. Interdiction-$\Pi$} is $\Sigma^p_2$-complete.
\end{theorem}
\begin{proof}
    Analogously to \Cref{thm:interdictionInSigma2}, \textsc{Comb. Interdiction-$\Pi$} is in $\Sigma^p_2$.
    For the hardness, we use the observation that if a problem $\Pi = (\I, \U, \sol)$ is SSP-NP-complete, there is an SSP reduction $(g, (f_I)_{I\in \I})$ from \textsc{Sat} to $\Pi$.
    We extend the SSP reduction $(g, (f_I)_{I\in \I})$ to a polynomial-time reduction $g'$ from \textsc{Comb. Interdiction-Sat} to \textsc{Comb. Interdiction-$\Pi$} as depicted in \Cref{fig:interdiction-meta-reduction}.
    The main idea is that the underlying reduction from \textsc{Sat} to $\Pi$ remains the same function $g$ and the additional set of blockable elements is redefined such that the blocking sets (which are the solutions to the problems) have a one-to-one correspondence.

    \begin{figure}[!ht]
        \centering
        \scalebox{1}{
            \begin{tikzpicture}
                \node[] (Sat) at (0,0) {\textsc{Sat}};
                \node[] (Pi) at (6,0) {$\Pi$};
                \node[] (ISat) at (-0.5,-2) {\textsc{Comb. Interdiction-Sat}};
                \node[] (IPi) at (6.5,-2) {\textsc{Comb. Interdiction-$\Pi$}};
                \draw[->] (Sat) to node[above] {$(g,(f_I)_{I \in \I})$} (Pi);
                \draw[->] (ISat) to node[above] {$g'$} (IPi);
                \draw[->] (0,-0.5) to node[left] {$I' = (I,B,t)$} (0,-1.5);
                \draw[->] (6,-0.5) to (6,-1.5);
            \end{tikzpicture}
        }
        \caption{The fact that $\textsc{Sat}$ is SSP reducible to $\Pi$ induces a reduction from $\textsc{Comb. Interdiction-Sat}$ to $\textsc{Comb. Interdiction-}\Pi$.}
        \label{fig:interdiction-meta-reduction}
    \end{figure}
    
    Let $I' = (I,B,t)$ be the instance of \textsc{Comb. Interdiction-Sat}, where $I \in \I$ is the corresponding \textsc{Sat} instance, $B$ is the set of blockable elements and $t$ is the threshold.
    Then, the reduction $g'$ is defined by $g'(I') = (g(I), B_{new}, t)$, where $B_\text{new}$ is the new blockable set in \textsc{Comb. Interdiction-$\Pi$} defined as
    $$
        B_\text{new} = f_I(B) \subseteq \U(g(I)).
    $$
    I.e. the injectively mapped universe elements from $\U(I)$ remain blockable in \textsc{Comb. Interdiction-$\Pi$} if and only if they are blockable in \textsc{Comb. Interdiction-Sat}.
    Furthermore, observe that by this definition of $B_\text{new}$, we have that all newly introduced universe elements, i.e. all elements in $\U(g(I)) \setminus f_I(\U(I))$ are not blockable.
    Now, recall that the SSP property for $(g,(f_I)_{I \in \I})$ states that the solutions of \textsc{Sat} and $\Pi$ correspond one-to-one to each other on the set $f_I(\U(I))$.
    Namely, if $\sol(I)$ denotes the solutions of the \textsc{Sat} instance and $ \sol(g(I))$ denotes the solutions of the corresponding $\Pi$ instance, then
    $$
        \set{f_I(S) : S \in \sol(I) } = \set{S' \cap f_I(\U(I)) : S' \in  \sol(g(I))}.
    $$
    As a consequence of this, and the fact that $B_\text{new} \subseteq f_I(\U(I))$, we have that each blocking set $B \subseteq \U(I)$ in \textsc{Comb. Interdiction-Sat} is in a direct one-to-one correspondence to some blocking set in $\textsc{Comb. Interdiction-}\Pi$. In particular, there exists a blocker for the instance $(I, B, t)$ if and only if there exists a blocker for the instance $(g(I), B_\text{new}, t)$. Finally, note that the whole reduction $g'$ can be computed in polynomial time, in particular since $g$ and $f_I$ (and therefore $B_\text{new}$) can be computed in polynomial time.
    In summary, \textsc{Comb. Interdiction-Sat} reduces to \textsc{Comb. Interdiction-$\Pi$} and \textsc{Comb. Interdiction-$\Pi$} is $\Sigma^p_2$-hard.
\end{proof}

\subsection{Adapting the Meta-Reduction to the Cost Version}

While the previous subsection showed $\Sigma^p_2$-completeness of the combinatorial interdiction problem \textsc{Comb. Interdiction-$\Pi$}, the goal of this subsection is to show the same for \textsc{Interdiction-$\Pi$}, i.e. the version with element costs. This is done via an easy reduction, which re-adapts the combinatorial interdiction problem to the cost version.
Note that the following theorem also holds for all LOP problems. Specifically, if one is given an LOP problem, one can consider the SSP problem derived from it as described in \Cref{sec:framework}.

\begin{theorem}
\label{thm:main-result-interdiction}
    For all SSP-NP-complete problems $\Pi$, the interdiction variant \textsc{Interdiction-$\Pi$} is $\Sigma^p_2$-complete.
\end{theorem}
\begin{proof}
    Due to \Cref{thm:interdictionInSigma2}, \textsc{Interdiction-$\Pi$} is in $\Sigma^p_2$.
    We further reduce \textsc{Comb. Interdiction-$\Pi$} to \textsc{Interdiction-$\Pi$}. Assume an instance $(I, B, t)$ of \textsc{Comb. Interdiction-$\Pi$} is given. 
    We use the cost function $c: \U(I) \rightarrow \Z$ in \textsc{Interdiction-$\Pi$} to distinguish the elements in the blockable set $B$ from those that are not blockable.
    For this, we set $c(b) = 1$ for all $b \in B$ (blockable) and $c(u) = t+1$ for all $u \in \U(I) \setminus B$ (not blockable).
    It is clear that every blocker $B'$ with $c(B') \leq t$ uses no elements from $\U(I) \setminus B$ and at most $t$ elements from $B$.
    The reduction is obviously polynomial-time computable.
    Consequently, the reduction is correct and \textsc{Interdiction-$\Pi$} is $\Sigma^p_2$-complete as well.
\end{proof}

\subsection{The Meta-Reduction is also an SSP reduction}
\label{sec:interdiction-is-also-SSP}
In \cref{sec:combinatorial-interdiction}, we showed that \textsc{Comb. Interdiction-Sat} reduces to \textsc{Comb. Interdiction-$\Pi$}. The goal of this subsection is to show the slightly stronger statement that 
this reduction is again an SSP reduction itself. This fact can be succinctly stated as \enquote{$\textsc{Satisfiability} \leqSSP \Pi$ implies $\textsc{Comb. Interdiction-Sat} \leqSSP \textsc{Comb. Interdiction-}\Pi$}. We believe that this is an elegant result which deserves to stand on its own right.
In order state this result it becomes necessary to explain in which way $\textsc{Comb. Interdiction-Sat}$ and \textsc{Comb. Interdiction}-$\Pi$ are interpreted as SSP problems. This is done the following way.

\begin{observation}\label{obs:interdictionIsSSP}
    The combinatorial interdiction variant $\textsc{Comb. Interdiction-}\Pi$ of an SSP problem $\Pi$ is an SSP problem.
\end{observation}
\begin{proof}
    Let $\Pi = (\I, \U, \sol)$ be an SSP problem and we denote \textsc{Comb. Interdiction-$\Pi$} $= (\I',\U',\sol')$.
    Note that by definition, we have $\U = \U'$.
    Now, let $I \in \I$ be an instance of $\Pi$.
    Then, we can define the corresponding instances of $\textsc{Comb. Interdiction-}\Pi$ by setting
    $$
        \I' = \{(I, B, t) \mid I \in \I, \ B \subseteq \U(I), \ t \in \N\}.
    $$
    Furthermore, we set the solutions for the combinatorial variant to
    $$
        \sol'(I') = \{B' \subseteq B \mid |B'| \leq t \ \text{and} \ \forall S \in \sol(I): B' \cap S \neq \emptyset\}
    $$
    Thus, \textsc{Comb. Interdiction-$\Pi$} $= (\I',\U',\sol')$ is an SSP problem.
\end{proof}

According to \Cref{obs:interdictionIsSSP}, the equivalent definition of \textsc{Comb. Interdiction-Satis\-fia\-bility} as SSP problem is the following.

\begin{definition}[\textsc{Combinatorial Interdiction-Satisfiability} as SSP Problem]
The interdiction version of \textsc{Sat} is a tuple $(\I', \U', \sol')$ with input $I = (L, C, B, t) \in \I'$ of literals and clauses, universe $\U'(I) = L$, solution set $\sol$, blockable set $B$ and threshold $t$.
The solution set is given by  $\sol'(I) = \set{B' \subseteq B$ : $|B'| \leq t \text{ and for all } S \in \sol(I) \text{ we have } B' \cap S \neq \emptyset}$.
\end{definition}

Now, we prove that the meta-reduction from \Cref{thm:interdictionSigma2Hard} is also an SSP reduction by using the reduction from \Cref{thm:interdictionSigma2Hard} and defining the corresponding function $(f'_{I'})_{I' \in \I'}$ to complete the SSP reduction.

\begin{corollary}\label{thm:combInterdictionSSPreduction:SSP}
For all SSP-NP-complete problems $\Pi$, \textsc{Comb. Interdiction-Sat} $\leqSSP$ \textsc{Comb. Interdiction-$\Pi$}.
\end{corollary}
\begin{proof}
    As implied by \Cref{obs:interdictionIsSSP}, \textsc{Comb. Interdiction-Sat} and \textsc{Comb. Inter\-dic\-tion-$\Pi$} are also an SSP problems.
    We extend the SSP reduction $(g, (f_I)_{I\in \I})$ to an SSP reduction $(g', (f_{I'})_{I'\in \I'})$ from \textsc{Comb. Interdiction-Sat} to \textsc{Comb. Interdiction-$\Pi$} as depicted in \Cref{fig:interdiction-ssp-reduction}, where $g'$ is the reduction from the proof of \Cref{thm:interdictionSigma2Hard}.
    
    \begin{figure}[!ht]
        \centering
        \scalebox{1}{
            \begin{tikzpicture}
                \node[] (Sat) at (0,0) {\textsc{Sat}};
                \node[] (Pi) at (6,0) {$\Pi$};
                \node[] (ISat) at (-0.5,-2) {\textsc{Comb. Interdiction-Sat}};
                \node[] (IPi) at (6.5,-2) {\textsc{Comb. Interdiction-$\Pi$}};
                \draw[->] (Sat) to node[above] {$(g,(f_I)_{I \in \I})$} (Pi);
                \draw[->] (ISat) to node[above] {$(g',(f'_{I'})_{I' \in \I'})$} (IPi);
                \draw[->] (0,-0.5) to node[left] {$I' = (I,B,t)$} (0,-1.5);
                \draw[->] (6,-0.5) to (6,-1.5);
            \end{tikzpicture}
        }
        \caption{\textsc{Satisfiability} $\leqSSP$ $\Pi$ implies \textsc{Comb. Interdiction-Sat} $\leqSSP$ \textsc{Comb. Interdiction-$\Pi$}.}
        \label{fig:interdiction-ssp-reduction}
    \end{figure}

    We recall the reduction from the proof of \Cref{thm:interdictionSigma2Hard}.
    As we have already proven the correctness of $g'$, we focus solely on the function $(f'_{I'})_{I' \in \I'}$.
    We define $f_I = f'_{I'}$ such that $B_\text{new}$ is the new blockable set in \textsc{Comb. Interdiction-$\Pi$} defined as
    $$
        B_\text{new} = f'_{I'}(B) = \{u \in \U(g(I)) \mid u \in f'_{I'}(B)\} =  \{u \in \U(g(I)) \mid u \in f_{I}(B)\}.
    $$
    Because the universe for \textsc{Sat} and \textsc{Comb. Interdiction-Sat} is the same, $f'_{I'}$ is well-defined.
    Consequently, we have $f_I(B) = f'_{I'}(B) \subseteq \U(g(I))$.
    As was stated in the proof of \Cref{thm:interdictionSigma2Hard}, the blockers correspond one-to-one to each other on the set $f_I(B)$.
    Hence, this reduction is indeed solution preserving.
\end{proof}

\section{Min-Max Regret Problems with Interval Uncertainty}
\label{sec:regret}
In this section, we consider min-max regret robust optimization problems with interval uncertainty.
Our main result in this chapter is that for every single LOP problem with the property that the SSP problem derived from it is in SSP-NPc (compare \cref{sec:framework}) by a so-called \emph{tight} reduction -- we define this term in \Cref{subsec:adapting-meta-reduction-regret} --, the corresponding min-max regret problem is $\Sigma^p_2$-complete.
This means that the problem is described by $\Pi = (\I, \U, \F, d, t)$, where for every instance $I \in \I$, $\U(I)$ denotes its universe, $\F(I)$ denotes its feasible solutions, $d^{(I)} : \U(I) \to \Z$ denotes its cost function, and $t^{(I)}$ denotes its threshold.
Like in \cref{def:LOSPP}, we define $\sol(I) = \set{S \in \F(I) : d^{(I)}(S) \leq t^{(I)}}$.
As an example, if $\Pi = \textsc{vertex cover}$, then $\F(I)$ contains all vertex covers, but $\sol(I)$ only contains those vertex covers of size at most the threshold $t^{(I)}$.

In order to define the min-max regret version of $\Pi$, we use the following definitions, which are standard in the area of min-max regret robust optimization \cite{DBLP:journals/disopt/DeinekoW10}:
For some cost function $c :  \U(I) \to \Z$, we define $S^\star_c \in \argmin_{S' \in \F(I)} c(S')$ to be an optimal feasible solution of $\Pi$ with respect to cost function $c$.
The \emph{regret} of some feasible solution $S \in \F(I)$ is defined as $\reg(S,c) := c(S) - c(S^\star_c)$. 
Given numbers $\underline c_e, \overline c_e \in \Z$ for every universe element $e \in \U(I)$ such that $\underline c_e \leq \overline c_e$, the \emph{interval uncertainty set} defined by the coefficient sequence $(\underline c_e, \overline c_e)_{e \in \U(I)}$ is defined as
$$
    C := \set{c \mid c : \U(I) \to \R \text{ is a function s.t. } c(e) \in [\underline c_e, \overline c_e] \ \forall e \in \U(I)}.
$$
(We remark that we are using the letter $C$ for the uncertainty set instead of the more standard letter $\U$, since we are using $\U$ already to denote the universe of the nominal problem.)

\begin{definition}[Min-Max Regret Problem with Interval Uncertainty]
\label{def:interval-regret}
    Let an LOP problem $\Pi = (\I, \U, \F, d, t)$ be given.
    The min-max regret problem with interval uncertainty associated to $\Pi$ is denoted by \textsc{Interval Min-Max Regret-$\Pi$} and defined as follows:
    The input is an instance $I \in \I$ together with integers $\underline c_e \leq \overline c_e$ for all $e \in \U(I)$ and a threshold $t_R \in \Z$.
    The question is whether 
    $$
        \min_{S \in \F(I)} \ \max_{c \in C} \ \reg(S,c) \leq t_R.
    $$
\end{definition}

We remark that this definition only applies to LOP problems, i.e.\ not to every problem in \cref{sec:ssp-reductions}.
This is simply due to the fact that it makes use of the concept of the feasible solutions $\F(I)$, and not every SSP problem has a set $\F(I)$.
However, in the next subsection, we introduce a slight variant of the min-max regret problem which is defined for every SSP problem and also show $\Sigma^p_2$-completeness for that variant.

Before we start, let us define the \emph{maximum regret} $r_\text{max}(S) := \max_{c \in C} \reg(S,c)$.
For a fixed feasible set $S \in \F(I)$, it is a well-known fact \cite{DBLP:journals/disopt/DeinekoW10} that the maximum regret is caused by the so-called \emph{canonical cost scenario} $c_S : \U(I) \to \Z$, defined by
$$
    c_S(e) :=   \begin{cases} \overline c_e & \text{ if $e \in S$} \\
                    \underline c_e & \text{ if $e \not\in S$}. 
                \end{cases}
$$
This implies that $r_\text{max}(S) = \reg(S, c_S) = c_S(S) - c_S(S^\star_{c_S}) = c_S(S) - \min_{S' \in \F(I)} c_S(S')$.

With this observation on the canonical cost scenario, we can view interval min-max regret problems as a two-player game of an $\exists$-player that wants to find a feasible solution against an adversary (the $\forall$-player).
The $\exists$-player is in control of the \emph{min} and thus is able to choose the solution $S$.
On the other hand, the adversary is in control of the \emph{max} and is thus able to choose the cost function $c$ from the set $C$ as well as the solution $S'$.
Consequently, we can reformulate the question to 
$$
    \exists S \in \F(I) : \forall S' \in \F(I) : c_S(S) - c_S(S') \leq t_R.
$$
With this perspective, it is easy to show the containment in the class $\Sigma^p_2$. 

\begin{lemma}
\label{lem:regret-containment-sigma-2}
    If $\Pi = (\I, \F, \U, d, t)$ is an LOP problem such that the derived SSP problem is in SSP-\NP, then \textsc{Interval Min-Max Regret-$\Pi$} is contained in the class $\Sigma^p_2$.
\end{lemma}
\begin{proof}
    We provide a polynomial time algorithm $V$ such that for $m_1, m_2 \leq poly(|I |)$:
    $$
        I \in L \ \Leftrightarrow \ \exists y_1 \in \{0,1\}^{m_1} \ \forall y_2 \in \{0,1\}^{m_2} : V(I, y_1, y_2) = 1.
    $$
    Observe that for \textsc{Interval Min-Max Regret-$\Pi$} we can characterize the yes-instances as follows:
    $$
        I \in L \ \Leftrightarrow \ \exists S \in \F(I) \ \forall S' \in \F(I) : c_S(S) - c_S(S') \leq t_R.
    $$
    Therefore with the $\exists$-quantified $y_1$, we encode the solution $S \subseteq \U(I)$ with $S \in \F(I)$.
    Because $|\U(I)| \leq poly(|I|)$, the encoding size of $y_1$ is polynomially bounded in the input size of $\Pi$.
    Next, we encode all solutions $S' \in \F(I)$ to the nominal problem $\Pi$ using the $\forall$-quantified $y_2$ within polynomial space as well.
    At last, the verifier $V$ has to verify the correctness of the given solution provided by the $\exists$-quantified $y_1$ and $\forall$-quantified $y_2$.
    Since we have assumed that we can efficiently check (for $\Pi$ in NP) whether proposed solutions $F \subseteq \U(I)$ are indeed feasible solutions, it can be checked in polynomial time whether $S \textit{ and } S' \in \F(I)$.
    Furthermore, the canonical cost scenario $c_S$ can be efficiently computed from $S$.
    Consequently, it can be checked in polynomial time whether $c_S(S) - c_S(S') \leq t_R$.
    It follows that $\textsc{Interdiction-}\Pi$ is in $\Sigma^p_2$.
\end{proof}

\subsection{Simplifying the Structure: Restricted Interval Min-Max Regret Problems}
In order to show the meta-theorem, we define a restricted variant of interval min-max regret problems.
In contrast to the non-restricted variant, which can only be defined for LOP problems (since the definition relies on the concept of $\F(I)$), the restricted version can be defined more generally for every SSP.
We call these problems \textsc{Restricted Interval Min-Max Regret-$\Pi$} for a nominal SSP problem $\Pi$.

\begin{definition}[Restricted Interval Min-Max Regret Problem]
\label{def:restricted-min-max-regret-problem}
    Let an SSP problem $\Pi = (\I, \U, \sol)$ be given.
    The restricted regret problem associated to $\Pi$ is denoted by \textsc{Restricted Interval Min-Max Regret-$\Pi$} and defined as follows:
    The input is an instance $I \in \I$ together with integers $\underline h_e, \overline h_e$ for all $e \in \U(I)$, such that $\underline h_e \leq \overline h_e$ and both $\underline h_e, \overline h_e \in \bin$, and a threshold $q \in \Z$.
    The question is whether 
    $$
        \min_{S \in \sol(I)} \ \max_{h \in H} \ \left( h(S) - \min_{S' \in \sol(I) } h(S') \right) \leq q.
    $$
    Here $H$ denotes the interval uncertainty set defined by the coefficients $\underline h_e, \overline h_e$. We assume that the input is a Yes-instance ($\sol(I) \neq \emptyset$), otherwise the (restricted) min-max regret is undefined.
\end{definition}

Observe that the restricted interval min-max regret is defined over the solutions $\sol(I)$ instead of the feasible solutions $\F(I)$ as for the standard interval min-max regret problems based on an LOP problem. 
Furthermore, the uncertainty set is restricted by $\underline h_e, \overline h_e$ to be from the set $\bin$.
(We remark that while it is not standard to calculate the regret relative to $\sol(I)$ instead of $\F(I)$, this has been considered recently by Coco, Santos and Noronha for the case of min-max regret for the set cover problem \cite{DBLP:journals/coap/CocoSN22}.) 

We begin with the canonical \textsc{Sat} variant \textsc{Restricted Interval Min-Max Regret-Sat} and show that it is $\Sigma^p_2$-complete.
Then, using this \textsc{Sat} problem as a basis for a reduction, we explain how this can be used to show that \textsc{Restricted Interval Min-Max Regret-$\Pi$} is $\Sigma^p_2$-hard.
In a similar spirit to previous sections of this paper, we consider $L$ as the universe of the problem and we define the set of solutions 
$$
    \sol(\varphi) := \set{L' \subseteq L : |L' \cap \set{x_i, \overline x_i}| = 1 \ \forall i = 1,\dots,n; \ |L' \cap c_j| \geq 1 \ \forall j = 1,\dots,m}
$$
as the subset of all literals encoding a satisfying assignment. 
Applying \cref{def:restricted-min-max-regret-problem} to the satisfiability problem yields the following.

\begin{definition}[\textsc{Restricted Interval Min-Max Regret-Sat}]
\label{def:min-max-regret-sat}
    We denote the restricted interval min-max regret version of \textsc{Sat} by \textsc{Restricted Interval Min-Max Regret-Sat}. The input is a CNF with clauses C and literals L, together with a threshold parameter $q \in \Z$ and integers $\underline h_\ell, \overline h_\ell$ for all $\ell \in L$, such that $\underline h_\ell \leq \overline h_\ell$ and both $\underline h_\ell, \overline h_\ell \in \bin$ for all $\ell \in L$.
    Let $H$ be the interval uncertainty set defined by the $\underline h_\ell, \overline h_\ell$.
    The question is whether
    $$
        \min_{S \in \sol(\varphi)} \ \max_{h \in H} \ \left( h(S) - \min_{S' \in \sol(\varphi) } h(S') \right) \leq q.
    $$
    We assume that the input formula $\varphi$ is satisfiable ($\sol(\varphi) \neq \emptyset$), otherwise the (restricted) min-max regret is undefined.
\end{definition}

As the next step, we show the $\Sigma^p_2$-completeness of \textsc{Restricted Interval Min-Max Regret-Sat}.

\begin{lemma}
    \textsc{Restricted Interval Min-Max Regret-Sat} is $\Sigma^p_2$-complete.
\end{lemma}
\begin{proof}
The containment in the class $\Sigma^p_2$ is analogous to \cref{lem:regret-containment-sigma-2}.
It remains to show hardness.
For the hardness, we reduce from \textsc{$\exists\forall$DNF-Sat}, which is $\Sigma^p_2$-hard as shown by Stockmeyer \cite{DBLP:journals/tcs/Stockmeyer76}.

\begin{description}
\item[Definition of the instance]
Assume we are given an instance $\exists X \forall Y \psi(X, Y)$ of \textsc{$\exists\forall$DNF-Sat}, where $\psi$ is in DNF.
We show how to construct an equivalent instance of \textsc{Restricted Interval Min-Max Regret-Sat} by describing $\varphi, \underline h_\ell,\overline h_\ell, t$ as in \cref{def:min-max-regret-sat} (where $\varphi$ is a CNF-formula and $\ell$ runs over the literals of $\varphi$).

We start by letting $\psi' := \neg \psi$. Note that because $\psi$ is in DNF, $\psi'$ is in CNF by applying De Morgan's rule. The variables and literals are the same between $\psi$ and $\psi'$. Observe that
$$
    \exists X \forall Y \psi(X,Y) \ \Leftrightarrow \ \exists X \neg \exists Y \psi'(X,Y).
$$
We define a new CNF formula $\varphi$, by introducing a new variable $z$ and appending it to every clause of $\psi'$, that is $\varphi \equiv \psi' \lor z$.
Note that $\varphi$ is satisfiable ($\sol(\varphi) \neq \emptyset$), since any assignment $\alpha$ with $\alpha(z) = 1$ is satisfying. 
Let $L_X := X \cup \overline X$ be the literal set belonging to the variables $X$ and $L_Y := Y \cup \overline Y$ be the literal set belonging to the variables $Y$. The formula $\varphi$ has literal set $L := L_X \cup L_Y \cup \set{z, \overline z}$.
We define the uncertainty set $H$, by making a case distinction.
For a literal $\ell \in L_X \cup \set{z, \overline z}$, we let $\underline h_\ell = 0$ and $\overline h_\ell = 1$.
For a literal $\ell \in L_Y$, we let $\underline h_\ell = \overline h_\ell = 0$. 
Finally we let $q := |X| = |L_X| / 2$.
This completes the description of the \textsc{Restricted Interval Min-Max Regret-Sat} instance $(\varphi, H, q)$.

\item[Correctness]
We claim that the regret for $(\varphi, H, q)$ is bounded from above by $q$ if and only if $\exists X \neg \exists Y \psi'(X,Y)$ is satisfiable.
The main idea behind this claim is that the min-max regret problem can be understood as a game between Alice and Bob: 
Alice selects an initial solution $S \in \sol(\varphi)$, Bob selects the cost function $h \in H$ together with some other solution $S' \in \sol(\varphi)$. 
Bob's goal is to cause a large regret for Alice (at least $q+1$).
By construction, such a large regret is only possible, if $S'$ and $S$ take exactly opposite literals on the set $L_X \cup \set{z, \overline z}$.
But this means Alice can \enquote{block} Bob, if she chooses her assignment on $X$ in a smart way.

We are now ready to formally prove the claim.
First, assume $\exists \alpha_1 \neg \exists \alpha_2 \psi'(\alpha_1,\alpha_2)$, where $\alpha_1 : X \to \bin$ and $\alpha_2 : Y \to \bin$ are truth assignments.
We let $A_1$ be the literals corresponding to $\alpha_1$ (i.e. if $\alpha(x) = 1$ then $x \in A_1$, else $\overline x \in A_1$).
Then Alice chooses a solution $S$ the following way: From $\set{z,\overline z}$, she chooses the literal $z$. 
From $L_X$, she chooses exactly the literals $\overline A_1$ opposite to $A_1$. From $L_Y$, she chooses an arbitrary literal for each variable $y \in Y$.
Observe that the subset $S \subseteq L$ described this way encodes a satisfying assignment of $\varphi$, since $z \in S$.
Furthermore, observe that the canonical cost scenario causing the maximum regret to Alice is given by 
$$
    h_S(\ell) =     \begin{cases}
                        1 & \text{ if } \ell \in (L_X \cup \set{z, \overline z}) \cap S \\
                        0 & \text{ if } \ell \in (L_X \cup \set{z, \overline z}) \setminus S \\
                        0 & \text{ if } \ell \in L_Y.
                    \end{cases}
$$
The properties of $h_S$ imply the following: The only way for Bob to cause a regret of $q+1 = |X| + 1$ is if  $S' \supseteq A_1 \cup \set{\overline z}$ (i.e. if he makes the exact opposite choice of Alice on the set $L_X \cup \set{z, \overline z}$).
But this is impossible: since $ \neg \exists \alpha_2 \psi'(\alpha_1,\alpha_2)$, Bob can never complete such a choice to a satisfying assignment $S' \in \sol(\varphi)$.
This means that the maximum regret is at most $q$.

On the other hand, assume that the regret is at most $q$.
In this case, Alice must have chosen $z$, (that is,\ $z \in S$), because otherwise
$$
    \min_{S \in \sol(\varphi)} \ \max_{h \in H} \ \left( h(S) - \min_{S' \in \sol(\varphi) } h(S') \right) = q+1.
$$
This is because if Alice did not choose $z$, then Bob can choose $z$ and exactly the opposite of Alice's choices on the set $L_X$ and an arbitrary assignment of $Y$.
(This choice by Bob satisfies the formula, since every assignment with $z = 1$ trivially satisfies the formula.)
Furthermore, it causes a regret of $q + 1$ since on the set $L_X \cup \set{z, \overline z}$ it is exactly the opposite of Alice's choice.
Hence we conclude that Alice must have chosen $z$.
But then, by a similar argument to the above, we have that $\exists \alpha_1 \neg \exists \alpha_2 \psi'(\alpha_1,\alpha_2)$.

We have successfully proven the claim.
In conclusion, we have that the regret is bounded by $q$ if and only if $\exists X \forall Y \psi(X,Y)$ is satisfiable.

\item[Polynomial Time]
Moreover, the reduction is polynomial-time computable because all transformations of the formula $\psi$ as well as the computation of the numbers $\underline h_\ell, \overline h_\ell$ are polynomial-time computable.
\end{description}

In conclusion, we have shown that \textsc{Restricted Interval Min-Max Regret-Sat} is $\Sigma^p_2$-complete.
\end{proof}

\subsection{A Meta-Reduction from Restricted Interval Min-Max Regret Problems}
\label{sec:restricted-regret-hardness}
Having established the hardness of our base problem \textsc{Restricted Interval Min-Max Regret-Sat}, we are now ready to prove our meta-theorem, i.e. we prove that for every SSP problem $\Pi$, which is also SSP-\NP-complete, its corresponding restricted interval min-max regret version is $\Sigma^p_2$-complete.

\begin{theorem}
    \label{thm:restricted-min-max-regret-reduction}
    For all SSP-NP-complete problems $\Pi$, the restricted interval min-max regret variant \textsc{Restricted Interval Min-Max Regret-$\Pi$} is $\Sigma^p_2$-complete.
\end{theorem}
\begin{proof}
The containment in the class $\Sigma^p_2$ is analogous to \cref{lem:regret-containment-sigma-2}.
Let an abstract, SSP-NP-complete problem $\Pi = (\I, \U, \sol)$ be given.
SSP-NP-completeness implies that $\textsc{Satisfiability} \leqSSP \Pi$.
(Where \textsc{Satisfiability} is defined like in \cref{sec:framework}). 
By the definition of an SSP reduction, there exists a poly-time computable function $g$ which maps CNF-formulas $\varphi$ to equivalent instances $g(\varphi) \in \I$ of $\Pi$.
Let $I := g(\varphi)$.
Furthermore there exist poly-time computable functions $f_\varphi : \U(\varphi) \to \U(I)$ with the SSP property. 
Here, $\U(\varphi)$ is the literal set of $\varphi$, and $\U(I)$ is the universe of the instance $I = g(\varphi)$ of $\Pi$. In order to reduce the notation, we write $f$ instead of $f_\varphi$.
The SSP property states that
$$
    \set{f(S) : S \in \sol(\varphi)} = \set{S' \cap f(\U(\varphi)) : S' \in \sol(I)}
$$
where $\sol(\varphi)$ is the set of all subsets of literals which satisfy $\varphi$.

The above is a reduction from \textsc{Satisfiability} to $\Pi$.
We now show that we can \enquote{upgrade} the reduction to obtain a reduction from \textsc{Restricted Interval Min-Max Regret-SAT} to \textsc{Restricted Interval Min-Max Regret-$\Pi$}. This upgraded reduction is described the following way:

Assume we are given an instance $(\varphi, (\underline h_\ell, \overline h_\ell)_{\ell\in L}, q)$ of \textsc{Restricted Interval Min-Max Regret-SAT} with the properties as described in \cref{def:min-max-regret-sat}.
We define an instance $(I, t_R, (\underline h_e, \overline h_e)_{e \in \U(I)})$ of \textsc{Restricted Interval Min-Max Regret-$\Pi$} in the following way:
\begin{itemize}
    \item   The instance is given by $I = g(\varphi)$.
            (Note this can be computed in polynomial time).
    \item   Let $n := |\U(I)|$.
            The numbers $\underline h_e, \overline h_e$ for $e \in \U(I)$ are defined by a case distinction.
            If $e \in f(\U(\varphi))$, then we let $\ell := f^{-1}(e)$ and 
            $$ 
                \underline h_e := \underline h_\ell, \ \text{ and } \ \overline h_e := \overline h_\ell.
            $$
            In the other case, $e \not \in f(\U(\varphi))$, we let $\underline h_e = \overline h_e := 0$.
            We define the resulting interval uncertainty set by $\tilde H$.
    \item   The threshold $t_R$ is given by $t_R := q$. 
\end{itemize}
This completes our description of the instance of \textsc{Restricted Interval Min-Max Regret-$\Pi$}.
Note that $\varphi$ is a satisfiable formula by \cref{def:min-max-regret-sat}, hence $g(\varphi) = I \in \I$ is a Yes-instance of $\Pi$ and $\sol(I) \neq \emptyset$.
Furthermore, the transformation is computable in polynomial time because $g$ and $f(\U)$ are polynomial-time computable.

At last, we prove that for the two instances of \textsc{Restricted Interval Min-max Regret Sat} and \textsc{Restricted Interval Min-Max Regret-$\Pi$} as defined above, it holds that their min-max regrets are equal.
Let $W :=  f(\U(\varphi))$ and consider the following chain of equalities:

\begin{align*}
    \text{min-max-regret}&(I, (\underline h_e, \overline h_e)_{e \in \U(I)})\\
    = &\min_{S \in \sol(I)} \ \max_{\Tilde h \in \Tilde H} \ \left( \Tilde h(S) - \min_{S' \in \sol(I) } \Tilde h(S') \right) \\
    = &\min_{S \in \sol(I)} \ \max_{\Tilde h \in \Tilde H} \ \left( \Tilde h(S \cap W) - \min_{S' \in \sol(I) } \Tilde h(S' \cap W) \right) \\
    = &\min_{S \in \sol(I)} \ \max_{h \in H} \ \left( h(f^{-1}(S \cap W)) - \min_{S' \in \sol(I) } h(f^{-1}(S' \cap W)) \right) \\
    = &\min_{T \in \sol(\varphi)} \ \max_{h \in H} \ \left( h(T) - \min_{T' \in \sol(\varphi) } h(T') \right)\\
    = &\ \text{min-max-regret}(\varphi, (\underline h_\ell, \overline h_\ell)_{\ell \in L})
\end{align*}

The first equality is by definition.
The second equality follows from the fact that $\tilde h$ is only non-zero on the set $W$.
The third equality follows since the uncertainty set $\tilde H$ behaves on $W$ analogous to $H$ on $f^{-1}(W) = L$. 
The fourth equality is due to the SSP property.
To see this, apply $f^{-1}$ to both sides of the equation $\set{f(S) : S \in \sol(\varphi)} = \set{S' \cap f(\U(\varphi)) : S' \in \sol(I)}$.
The last equality is again by definition.

We remark that the crucial step in the proof of the above lemma is the usage of the SSP property.
In summary, we have described a polynomial time reduction from the $\Sigma^p_2$-complete problem \textsc{Restricted Interval Min-max Regret Sat} to \textsc{Restricted Interval Min-Max Regret-$\Pi$}, such that the objective value of both problems stays exactly the same.
This proves that \textsc{Restricted Interval Min-Max Regret-$\Pi$} is $\Sigma^p_2$-hard. Together with the containment, as argued above, the problem is $\Sigma^p_2$-complete.
\end{proof}

\subsection{Adapting the Meta-Reduction to the Interval Min-Max Regret Version}\label{subsec:adapting-meta-reduction-regret}
For the completion of our main result, we have to show that the result on the restricted interval min-max regret problem is also applicable to the non-restricted interval min-max regret problem.
We do this by re-adapting the cost function of the restricted version back to the standard version based on the LOP problem.
Furthermore, we have to prove that the regret is not affected by this adaptation.
The underlying instance, however, stays the same within the reduction.

In order not to affect the regret, we use an additional property on the SSP reduction from the SSP problem $\Pi_1$ to the (SSP version derived from) LOP problem $\Pi_2$.
Specifically, the reduction needs to include only optimal solutions in the solution set.
Then, we can split the costs of each element into a dominant part, the original costs $d$ from the LOP problem, and a subordinate part, the costs based on the restricted interval min-max regret problem $h$.
Thus, all possible solutions that are considered for the regret are minimizers for $d$, i.e., optimal solutions for problem $\Pi_2$.
Accordingly, we are able to subtract the costs for $d$ and obtain the regret solely based on the costs $h$.
We call these reductions \emph{tight}.

\begin{definition}
    Let $\Pi_1$ be an SSP problem and $\Pi_2 = (\I, \U, \F, d, t)$ be an LOP problem.
    Consider an SSP reduction $(g, (f_I)_{I \in \I})$ from $\Pi_1$ to (the SSP problem derived from) $\Pi_2$.
    The reduction is called tight if for all yes-instances $I_1$ of $\Pi_1$, the corresponding instance $I_2 = g(I_1)$ of $\Pi_2$ with the associated parameter $t := t^{(I_2)}$ and associated cost function $d := d^{(I_2)}$, the following holds:
    $$
    	\{F \in \F(I_2) : d(F) \leq t\} \neq \emptyset \text{ and } \{F \in \F(I_2) : d(F) \leq t - 1\} = \emptyset
    $$
\end{definition}

All SSP reductions (to SSP problems derived from LOP problems) that can be found in \Cref{sec:ssp-reductions} fulfill this definition and are thus tight.

\begin{theorem}
    \label{thm:min-max-regret-reduction}
    For all LOP problems $\Pi$ with the property that the SSP problem derived from them is contained in SSP-NPc by a tight SSP reduction, the interval min-max regret variant \textsc{Interval Min-Max Regret-$\Pi$} is $\Sigma^p_2$-complete.
\end{theorem}
\begin{proof}
The containment in the class $\Sigma^p_2$ follows from \cref{lem:regret-containment-sigma-2}.
Let an abstract, \NP-complete LOP problem $\Pi = (\I, \F, \U, d, t)$ be given.
The corresponding SSP problem is defined by $\Pi' = (\I, \U, \sol)$ with
$$
    \sol = \{f \in \F(I) \mid d^{(I)}(f) \leq t^{(I)} \}.
$$

Assume we are given an instance $(I, q, (\underline h_e, \overline h_e)_{e \in \U(I)})$ of \textsc{Restricted Interval Min-max Regret-$\Pi'$} with the properties as described in \cref{def:restricted-min-max-regret-problem}.
We define an instance $(I, t_R, (\underline c_e, \overline c_e)_{e \in \U(I)})$ of \textsc{Interval Min-Max Regret-$\Pi$} in the following way:
\begin{itemize}
    \item The underlying instance $I$, as well as the corresponding universe $\U(I)$ remain the same.
    \item   Let $n := |\U(I)|$.
            We define the numbers $\underline c_e, \overline c_e$ for all $e \in \U(I)$ by
            $$ 
                \underline c_e := 2(n + 1)d^{(I)}(e) + \underline h_e, \ \text{ and } \overline c_e := 2(n + 1)d^{(I)}(e) + \overline h_e.
            $$
            We define the resulting interval uncertainty set by $C$.
    \item   The threshold $t_R$ is given by $t_R := q$.
\end{itemize}
This completes our description of the instance of \textsc{Interval Min-Max Regret-$\Pi$}.
(Note that $t_R$ and $t^{(I)}$ are different thresholds).
The transformation is computable in polynomial time, because the underlying instance is unchanged and only numbers encoded in polynomial size ($\Pi$ is in NP) are added to $\underline h_e$ and $\overline h_e$ resulting in $\underline c_e$ and $\overline c_e$.

Again, we prove that for the two instances of \textsc{Restricted Interval Min-max Regret-$\Pi$} and \textsc{Interval Min-Max Regret-$\Pi$} as defined above, it holds that their min-max regrets are equal.
For this, consider the following chain of equalities:

\begin{align*}
    \text{min-max-regret}(I, (\underline c_e, \overline c_e)_e) = &\min_{S \in \F(I)} \ \max_{c \in C} \ \left( c(S) - \min_{S' \in \F(I) } c(S') \right)\\
    = &\min_{S \in \sol(I')} \ \max_{c \in C} \ \left( c(S) - \min_{S' \in \sol(I) } c(S') \right) \\
    = &\min_{S \in \sol(I')} \ \max_{h \in H} \ \left( h(S) - \min_{S' \in \sol(I) } h(S') \right) \\
    = &\ \text{min-max-regret}(I', (\underline h_\ell, \overline h_\ell)_\ell)
\end{align*}

The first equality is by definition.
For the second equality, note that no matter which $c \in C$ is chosen, we always have $c(S') \in [2(n+1)d^{(I)}(S'), 2(n+1)d^{(I)}(S') + n]$ by definition of the coefficients $\underline c_e,\overline c_e$.
Hence any $S'$ minimizing $c(S')$ also minimizes $d^{(I)}(S')$.
Hence by the definition of $\sol(I)$ and by the fact that $\sol(I) \neq \emptyset$, we have that $S' \in \sol(I)$.
By a similar argument, if $S$ optimizes the first minimum, but $S \not\in \sol(I)$, then $ c(S) - \min_{S' \in \F(I) } c(S') \geq 2(n+1) - n = n+1$.
But this is a contradiction, since this last expression is always at most $n$ if both $S,S' \in \sol(I)$ and the SSP reduction is tight.

The third equality follows from the fact that the function $c$ can be decomposed into the part which is contributed by $d^{(I)}$ and the part contributed by $h$.
Since both $S,S'$ are minimizers of $d^{(I)}$, we have that $d^{(I)}(S) = d^{(I)}(S')$ and this part cancels out.
The last equality is again by definition.
\end{proof}

\subsection{The Meta-Reduction is an SSP reduction}
In \cref{sec:restricted-regret-hardness}, we showed that \textsc{Restricted Interval Min-Max Regret SAT} reduces to \textsc{Restricted Interval Min-Max Regret-$\Pi$}.
The goal of this subsection (analogous to \cref{sec:interdiction-is-also-SSP}) is to show the slightly stronger statement that 
this reduction is again an SSP reduction itself.
In order to state this result, it becomes necessary to explain in which way \textsc{Restricted Interval Min-Max Regret-$\Pi$} is interpreted as SSP problem.
This is done the following way.

\begin{observation}\label{lem:regretIsSSP}
    The restricted interval min-max regret variant \textsc{Restricted Interval Min-Max Regret-$\Pi$} of an SSP problem $\Pi$ is an SSP problem.
\end{observation}
\begin{proof}
    Let $\Pi = (\I, \U, \sol)$ be an SSP problem. We denote \textsc{Restricted Interval Min-Max Regret-$\Pi$} $=: (\I',\U',\sol')$.
    First, we set $\U = \U'$.
    Second, we can define the instance set $\I'$ of the restricted min-max regret variant of $\Pi$ by setting
    $$
        \I' = \{(I, H, q) \mid I \in \I, \ H \ \text{is an interval uncertainty set}, \ q \in \Z\}.
    $$
    Furthermore, for an instance $I' = (I, H, q) \in \I'$, we define its solution set as
    $$
        \sol'(I') = \left\{S \in \sol(I) \mid \max_{c \in H} \ \left( c(S) - \min_{S' \in \sol(I) } c(S') \right) \leq q \right\}.
    $$
    Thus, \textsc{Restricted Interval Min-Max Regret-$\Pi$} $= (\I',\U',\sol')$ is an SSP problem.
\end{proof}

It remains to prove that the meta-reduction from \Cref{thm:restricted-min-max-regret-reduction}, which we denote by $g'$, is also an SSP reduction by defining the corresponding function $((f'_{I'})_{I' \in \I'})$ to complete the SSP reduction.

\begin{corollary}
    For all SSP-NP-complete problems $\Pi$, \textsc{Restricted Interval Min-Max Regret-Sat} $\leqSSP$ \textsc{Restricted Interval Min-Max Regret-$\Pi$}.
\end{corollary}
\begin{proof}
    The underlying instance of the reduction $(g, (f_I)_{I \in \I})$ from \textsc{Sat} to $\Pi$ remains the same in the reduction of \Cref{thm:restricted-min-max-regret-reduction}. Note that this reduction transforms instances $I$ of $\textsc{Sat}$ into instances $g(I)$ of $\Pi$.
    We now want to show that the reduction $(g', (f'_{I'})_{I' \in \I'})$ with some additional function $((f'_{I'})_{I' \in \I'})$ from \textsc{Restricted Interval Min-Max Regret-Sat} to \textsc{Restricted Interval Min-Max Regret-$\Pi$} indeed is an SSP reduction.
    For this, we set $f'_{I'} := f_I$.
    Recall that if $e \notin f(\U(\varphi))$, the costs defined by the reduction are $\underline h_e = \overline h_e = 0$.
    Thus, no restriction is imposed by these elements on any solution.
    Now, we consider the elements that are part of $f_I(\U(\varphi)) = f'_{I'}(\U(\varphi))$.
    (Note that the instance remains the same.)
    If $e \in f(\U(\varphi))$, the costs are set to $\underline h_e = \underline h_\ell$ and $\overline h_e = \overline h_\ell$.
    The original reduction $\textsc{Sat} \leqSSP \Pi$ has the solution preserving property.
    Hence the solution sets $\sol(I)$ and $\sol(g(I))$ correspond to each other.
    Furthermore, we have that $\sol'(I')$ are exactly those elements of $\sol(I)$ with regret at most $q$, and $\sol'(g'(I'))$ are exactly those elements of $\sol(g(I))$ with regret at most $q$.
    At last, the regret of a solution stays the same when transferring it from $\textsc{Sat}$ to $\Pi$.
    Consequently, we conclude that the reduction $(g', (f'_{I'})_{I' \in \I'})$ is an SSP reduction.
 \end{proof}
\section{Two-Stage Adjustable Problems}
\label{sec:two-stage}
In this section, we consider two-stage adjustable robust optimization problems with discrete budgeted uncertainty.
We show that for every single LOP problem, whose derived SSP problem is SSP-NP-complete, that the corresponding two-stage adjustable problem is $\Sigma^p_3$-complete.
We choose to make the following definition of a two-stage adjustable problem, where in the first stage an arbitrary set $S_1$ can be selected, but needs to be completed to a feasible set in the second stage (i.e. $S_1 \cup S_2 \in \F(I)$).
We remark that the formal definition mentions the cost function $d$ and the cost threshold $t$ of the LOP.
While these are not necessary for the definition (the function $d$ and the threshold $t$ get substituted with new ones), they are necessary for the proofs.

\begin{definition}[Two-Stage Adjustable Problem]
\label{def:two-stage}
    Let an LOP problem $\Pi = (\I, \U, \F, d, t)$ be given.
    The two-stage adjustable problem associated to $\Pi$ is denoted by \textsc{Two-Stage Adjustable-$\Pi$} and defined as follows:
    The input is an instance $I \in \I$ together with three cost functions: the first stage cost function $c_1: \U(I) \rightarrow \Z$ and the second stage cost functions $\underline c_2: \U(I) \rightarrow \Z$ and $\overline c_2: \U(I) \rightarrow \Z$, a threshold $t_{TS} \in \Z$ and an uncertainty parameter $\Gamma \in \Z$.
    Then, the uncertainty set is defined by all cost functions such that at most $\Gamma$ elements $u \in \U(I)$ have costs of $\overline c_2(u)$, while all other have $\underline c_2(u)$:
    $$
        C_\Gamma := \{c_2 \mid \forall u \in \U(I) : c_2(u) = \underline c_2(u) + \delta_u(\overline c_2(u) - \underline c_2(u)), \ \delta_u \in \bin, \sum_{u \in \U(I)} \delta_u \leq \Gamma \}.
    $$
    The question is whether
    $$
        \min_{S_1 \subseteq \U(I)} \max_{c_2 \in C_\Gamma} \min_{\substack{S_2 \subseteq \U(I) \setminus S_1\\ S_1 \cup S_2 \in \F(I)}} c_1(S_1) + c_2(S_2) \leq t_{TS}.
    $$
\end{definition}

We remark that this definition only applies to LOP problems, i.e.\ not to every problem in \cref{sec:ssp-reductions}.
This is simply due to the fact that it makes use of the concept of the feasible solutions $\F(I)$, and not every SSP problem has a set $\F(I)$.
However, in the next subsection, we introduce a slight variant of the two-stage problem which is defined for every SSP problem and also show $\Sigma^p_3$-completeness for that variant.

Two-stage adjustable problems can be also viewed as a two-player game of an $\exists$-player that wants to find a feasible solution against an adversary (the $\forall$-player).
The $\exists$-player is in control of both \emph{min} and thus is able to choose the partial solutions $S_1$ and $S_2$.
On the other hand, the adversary is in control of the \emph{max} and is thus able to choose the cost function $c_2$ from the set $C_\Gamma$.
Consequently, we can reformulate the question to 
$$
    \exists S_1 \subseteq \U(I) : \forall c_2 \in C_\Gamma : \exists S_2 \subseteq \U(I) \setminus S_1 : S_1 \cup S_2 \in \F(I) \ \textit{and} \ c_1(S_1) + c_2(S_2) \leq t_{TS}.
$$

Again, we aim to locate the complexity of two-stage adjustable problems based on LOP problems exactly.
By the game theoretical perspective, it is intuitive to see the containment in $\Sigma^p_3$ for all problems that are polynomial-time checkable.
Thus, we derive the following lemma and prove this intuition to be true.

\begin{lemma}\label{thm:two-stageInSigma3}
    If $\Pi = (\I, \F, \U, d, t)$ is an LOP problem such that the derived SSP problem is in SSP-\NP, then $\textsc{Two-Stage-}\Pi$ is in $\Sigma^p_3$.
\end{lemma}
\begin{proof}
    We provide a polynomial time algorithm that verifies a specific solution $y_1, y_2, y_3$ of polynomial size for instance $I$ such that
    $$
        I \in L \ \Leftrightarrow \ \exists y_1 \in \{0,1\}^{m_1} \ \forall y_2 \in \{0,1\}^{m_2} \ \exists y_3 \in \{0,1\}^{m_3} : V(I, y_1, y_2, y_3) = 1.
    $$
    With the $\exists$-quantified $y_1$, we encode the first partial solution to $S_1 \subseteq \U(I)$.
    Because $|\U(I)| \leq poly(|I|)$, the encoding size of $y_1$ is polynomially bounded in the input size of $\Pi$.
    Next, we encode all cost functions $c_2 \in C_\Gamma$ over the $\forall$-quantified $y_2$.
    For this, we encode which of the at most $\Gamma$ elements are chosen to have costs $\overline c_2$.
    This is at most polynomial in the input size.
    At last, we encode the second partial solution $S_2 \subseteq \U(I)$ with the help of the $\exists$-quantified $y_3$.
    Again, this is at most polynomial in the input size as before.

    Now, $V$ has to check whether $S_2 \subseteq \U(I) \setminus S_1$,  $S_1 \cup S_2 \in \F(I)$ \textit{and} $c_1(S_1) + c_2(S_2) \leq t_{TS}$.
    All of these checks can be done in polynomial time, where $S_1 \cup S_2 \in \F(I)$ \textit{and} $c_1(S_1) + c_2(S_2) \leq t_{TS}$ can be checked in polynomial time because $\Pi \in \NP$ and we have assumed that feasible solutions can be checked in polynomial time.
    It follows that $\textsc{Two-Stage Adjustable-}\Pi$ is in $\Sigma^p_3$.
\end{proof}

\subsection{Combinatorial Two-Stage Adjustable Problems}
To show the hardness of two-stage adjustable problems, we introduce a new problem which we call the \emph{combinatorial version} of two-stage adjustable problems.
In this problem, the cost function is substituted by a set of blockable elements.
This combinatorial version is slightly more specific than the cost version and allows for a more precise reduction in the SSP framework.
The hardness of the combinatorial version also implies the hardness of the cost version as we show later.
For this, we redefine the cost functions into sets of elements:
The set of first stage elements $U_1 \subseteq \U(I)$ and the set of second stage elements $U_2 = \U(I) \setminus U_1$.
The set of blockable elements $B \subseteq U_2$ is a subset of the second stage elements.

We define the combinatorial version on basis of an SSP problem, which is no restriction because LOP problems are special SSP problems, which contain more structure.
On the contrary, SSP reductions allow for a more succinct presentation of the hardness proof and its preliminaries.

\begin{definition}[Combinatorial Two-Stage Adjustable Problem]
\label{def:comb-two-stage}
    Let an SSP problem $\Pi = (\I, \U, \sol)$ be given.
    The combinatorial two-stage adjustable problem associated to $\Pi$ is denoted by \textsc{Comb. Two-Stage Adjustable-$\Pi$} and defined as follows:
    The input is an instance $I \in \I$ together with three sets $U_1 \subseteq \U(I)$, $U_2 = \U(I) \setminus U_1$ and $B \subseteq U_2$, and an uncertainty parameter $\Gamma \in \N_0$.
    The question is whether
    \begin{align*}
        & \exists S_1 \subseteq U_1 : \forall B' \subseteq B \ \text{with} \ |B'| \leq \Gamma : \exists S_2 \subseteq U_2 : S_1 \cup S_2 \in \sol(I) \ \text{and} \ S_2 \cap B' = \emptyset.
    \end{align*}
\end{definition}

As the basis for our meta-reduction, we use the canonical \textsc{Satisfiability} problem.
\cref{def:comb-two-stage} applied to $\Pi = \textsc{Satisfiability}$ yields the following:

\begin{definition}[\textsc{Combinatorial Two-Stage Adjustable-Satisfiability}]
The problem \textsc{Comb. Two-Stage Adjustable-Sat} is defined as follows.
The input is a CNF with clauses $C$ and literal set $L$, sets $U_1 \subseteq L$, $U_2 = L \setminus U_1$ and $B \subseteq U_2$, and an uncertainty parameter $\Gamma \in \N_0$.
Then the question is whether there is $S_1 \subseteq U_1$ such that for all $B' \subseteq B$ with $|B'| \leq \Gamma$, there is $S_2 \subseteq U_2$, such that $|(S_1 \cup S_2) \cap c_j| \geq 1$ for all $c_j \in C$, $|(S_1 \cup S_2) \cap \{x_i, \overline x_i\}| = 1$ for all pairs $x_i, \overline x_i \in L$, and $S_2 \cap B' = \emptyset$.
\end{definition}

\subsection{A Meta-Reduction for Combinatorial Two-Stage Adjustable Problems}
Now, we show the $\Sigma^p_3$-hardness of the canonical \textsc{Satisfiability} problem, \textsc{Comb. Two-Stage Adjustable-Satisfiability}.

\begin{lemma}
    \textsc{Combinatorial Two-Stage Adjustable-Satisfiability} is $\Sigma^p_3$-hard.
\end{lemma}
\begin{proof}
    The containment in the class $\Sigma^p_3$ is analogous to \cref{thm:two-stageInSigma3}.
    For the hardness, we reduce \textsc{$\exists\forall\exists$CNF-Sat}, which is $\Sigma^p_3$-hard as shown by Stockmeyer \cite{DBLP:journals/tcs/Stockmeyer76}, to \textsc{Comb. Two-Stage Adjustable-Sat}.
    Let $\exists X \forall Y \exists Z \varphi(X, Y, Z)$ be the \textsc{$\exists\forall\exists$CNF-Sat} instance.
    We transform this into an instance $(\exists X' \psi(X'), U_1, U_2, B, \Gamma)$ of the combinatorial two-stage adjustable satisfiability problem, with the set of first stage elements $U_1$, set of second stage elements $U_2$, set of blockable elements $B$, and uncertainty parameter $\Gamma$.

    Again, we use the idea of Goerigk, Lendl and Wulf \cite[Theorem 1]{DBLP:journals/dam/GoerigkLW24} for their reduction for \textsc{Robust Adjustable Sat} to model the $\forall$-quantifier with $\Gamma$-uncertainty. 
    The proof is very similar to the proof of \cref{lem:interdiction-sat-sigma-2-complete}. 
    The difference this time is that the two-stage problem can be understood as a game between Alice and Bob, 
    where Alice selects the literals from $U_1$, Bob selects some blocker $B' \subseteq B$ with $|B'| \leq \Gamma$, and Alice responds by selecting literals from $U_2 \setminus B'$.
    Another difference is that we start with a CNF formula.
    Analogously to \cref{lem:interdiction-sat-sigma-2-complete}, for $n := |Y|$, we introduce new variables $Y^t = \set{y^t_1, \dots, y^t_n}$ and $Y^f = \set{y^f_1, \dots, y^f_n}$. 
    We say that Bob plays honestly, if he chooses a blocker $B'$ with $|B' \cap \set{y^t_i, y^f_i}| = 1$ for all $i=1,\dots,n$. 
    Otherwise we say that Bob cheats.

    \begin{description}
        \item[Description of the instance] Given an instance $\exists X \forall Y \exists Z \varphi(X, Y, Z)$ of \textsc{$\exists \forall \exists$CNF-Sat}, we define an instance of the two-stage adjustable satisfiability problem the following way: 
        We introduce new variables $Y^t, Y^f$ as explained above. A new formula $\varphi'$ is created from $\varphi$ in terms of a literal substitution process. 
        For each $i \in \set{1,\dots,n}$, each occurrence of some literal $y_i$ is substituted with the positive literal $y^t_i$. 
        Each occurrence of some literal $\overline y_i$ is substituted with the positive literal $y^f_i$. 
        We introduce a new variable $s$ and obtain a new formula $\varphi''$ by adding $s$ to every clause of $\varphi'$, that is $\varphi'' \equiv \varphi' \lor s$.
        We introduce new variables $s_1, \dots, s_n$. we let $W := \set{s} \cup \set{s_1, \dots, s_n}$. The formula $\psi$ is then defined by
        \begin{align*}
            \psi(X, Y^t, Y^f, Z, W) = \ &\varphi''(X, Y^t, Y^f, Z, s) \\
            &\land \left( \bigwedge\limits_{i=1}^n (y^t_i \lor \overline s_i)\land (y^f_i \lor \overline s_i) \right) \land (\overline s \lor s_1 \lor s_s \lor \dots \lor s_n).
        \end{align*}
        Finally, we define $L$ to be the literal set of $\psi$, i.e. the set of positive and negative literals corresponding to the variables $X, Y^t, Y^f, Z, W$. We let $U_1 := X \cup \overline X$, and $U_2 := L \setminus U_1$, and $B := Y^t \cup Y^f$, and $\Gamma = n$. This completes the description of the two-stage SAT instance. 
        \item[Correctness] Since the proof is very similar to \cref{lem:interdiction-sat-sigma-2-complete}, we only sketch the high-level argument. 

        Observe that in an optimal game, Bob can not cheat. On the other hand, if Bob plays honestly, Alice can not take any of the positive literals $s, s_1, \dots, s_n$ into her second-stage solution $S_2$. In other words, the cheat-detection gadget can be used if and only if Bob cheats.
        
        First, assume that $\varphi(X, Y, Z)$ is a Yes-instance of \textsc{$\exists \forall \exists$CNF-Sat}.
        This means that $\exists \alpha_X \forall \alpha_Y \exists \alpha_Z \varphi(\alpha_X, \alpha_Y, \alpha_Z)$. 
        Here, $\alpha_X : X \to \bin$, $\alpha_Y : Y \to \bin$ and $\alpha_Z : Z \to \bin$ denote assignments. 
        We claim that Alice has a winning strategy for the two-stage SAT game. 
        In fact, her winning strategy is as follows: In the first step, on the set $U_1$, she chooses literals which correspond to $\alpha_X$. 
        Now there are two cases: If Bob cheats, Alice can win because of the cheat-detection gadget. 
        If Bob plays honestly, Alice can win because $\forall \alpha_Y \exists \alpha_Z \varphi(\alpha_X, \alpha_Y, \alpha_Z)$.

        Second, assume that $\varphi(X, Y, Z)$ is a No-instance of \textsc{$\exists \forall \exists$CNF-Sat}. Consider the first-stage choice $S_1 \subseteq U_1$. 
        Note that Alice is forced to take a set $S_1 \subseteq U_1$ such that $|S_1 \cap \set{x_i, \overline x_i}| = 1$ for all $i$, because otherwise it is impossible that $S_1 \cup S_2 \in \sol(\psi)$. 
        Hence the choice $S_1$ of Alice corresponds to some assignment $\alpha_X : X \to \bin$. 
        Then, because $\varphi$ is a No-Instance, this means that with respect to this assignment $\alpha_X$, we have $\exists \alpha_Y \forall \alpha_Z \neg \varphi(\alpha_X, \alpha_Y, \alpha_Z)$. This means that Bob can play honestly according to $\alpha_Y$ and have a winning strategy. In total, we have shown that $\exists \forall \exists$CNF-SAT reduces to the \textsc{Comb. Two-Stage Adjustable-Sat}.
        \item[Polynomial Time] The transformation is doable in polynomial time because for each variable a constant number of variables and clauses is introduced.
    \end{description}
\end{proof}

Now, we are able to develop the meta-reduction for two-stage problems based on SSP-NP-complete problems $\Pi$. Precisely, we show a meta-reduction from \textsc{Combinatorial Two-Stage Adjustable-Satisfiability} to \textsc{Combinatorial Two-Stage Adjustable-$\Pi$}.

\begin{theorem}\label{thm:two-stageSigma3Hard}
    For all SSP-NP-complete problems $\Pi$, the combinatorial two-stage adjustable variant $\textsc{Comb. Two-Stage Adjustable-}\Pi$ is $\Sigma^p_3$-complete.
\end{theorem}
\begin{proof}
    The containment in the class $\Sigma^p_3$ is analogous to \cref{thm:two-stageInSigma3}.
    For the hardness, we again use that there is an SSP reduction $(g, (f_I)_{I \in \I})$ from \textsc{Sat} to $\Pi$ because $\Pi$ is an SSP-NP-complete problem.
    We design a reduction $g'$ from \textsc{Comb. Two-Stage Adjustable-Sat} to $\textsc{Comb. Two-Stage Adjustable-}\Pi$ by extending the existing reduction $(g, (f_I)_{I \in \I})$ as depicted in \Cref{fig:two-stage-meta-reduction}.
    As in the reduction for interdiction problems, the underlying reduction from \textsc{Sat} to $\Pi$ remains and we construct the corresponding combinatorial two-stage instance.
    
    \begin{figure}[!ht]
        \centering
        \scalebox{1}{
            \begin{tikzpicture}
                \node[] (Sat) at (0,0) {\textsc{Sat}};
                \node[] (Pi) at (6,0) {$\Pi$};
                \node[] (TSSat) at (-1.5,-2) {\textsc{Comb. Two-Stage Adj.-Sat}};
                \node[] (TSPi) at (7.5,-2) {\textsc{Comb. Two-Stage Adj.-$\Pi$}};
                \draw[->] (Sat) to node[above] {$(g,(f_I)_{I \in \I})$} (Pi);
                \draw[->] (TSSat) to node[above] {$g'$} (TSPi);
                \draw[->] (0,-0.5) to node[left] {$I' = (I, U_1 , U_2, B, \Gamma)$} (0,-1.5);
                \draw[->] (6,-0.5) to (6,-1.5);
            \end{tikzpicture}
        }
        \caption{The fact that $\textsc{Sat}$ is SSP reducible to $\Pi$ induces a reduction from $\textsc{Comb. Two-Stage Adjustable-Sat}$ to $\textsc{Comb. Two-Stage Adjustable-}\Pi$.}
        \label{fig:two-stage-meta-reduction}
    \end{figure}

    Let $I' = (I, U_1, U_2, B, \Gamma)$ be the instance of $\textsc{Comb. Two-Stage Adjustable-Sat}$ with the corresponding \textsc{Sat} instance $I \in \I$.
    We denote the set of first stage elements by $U_1$, the set of second stage elements by $U_2$ and the set of blockable elements by $B$.
    Furthermore, we denote the uncertainty parameter by $\Gamma$.
    The reduction $g'$ is defined by $g'(I') = (g(I), U'_1, U'_2, B_\text{new}, \Gamma)$.
    The new sets are defined by
    $$
        U'_1 = f_{I}(U_1), \ U_2' = \U(g(I)) \setminus U_1', \ B_\text{new} = f_{I}(B).
    $$
    The uncertainty parameter $\Gamma$ remains the same.
    This completes the description of the reduction.
    We claim that we have a direct one-to-one correspondence between the solutions of \textsc{Comb. Two-Stage Adjustable-Sat} and \textsc{Comb. Two-Stage Adjustable-$\Pi$}.
    Indeed, this follows from the fact $B_\text{new} \subseteq f_I(\U(I))$ together with the SSP property of the SSP-reduction $\textsc{Sat} \leq \Pi$, which states that the solutions of $\textsc{Sat}$ and $\Pi$ correspond one-to-one to each other on the set $f_I(\U(I))$.
    Furthermore, the transformation is computable in polynomial time because $g$ and $f$ are polynomial-time computable.
    It follows that \textsc{Comb. Two-Stage Adjustable-Sat} $\leq$ \textsc{Comb. Two-Stage Adjustable-$\Pi$}.
    In particular, this implies that \textsc{Comb. Two-Stage Adjustable-$\Pi$} is $\Sigma^p_3$-complete.
\end{proof}

\subsection{Adapting the Meta-Reduction to the Cost Version}

As the last step of the meta-reduction, we have to show that the result on the combinatorial version of two-stage adjustable problems is also applicable to the cost version.
This is indeed the case and we show this by reintroducing the cost function $d$ and threshold $t$ of the original nominal LOP problem.

\begin{theorem}
    For all LOP problems $\Pi$ with the property that the SSP problem derived from them is contained in SSP-NPc, the two-stage adjustable variant \textsc{Two-Stage Adjustable-$\Pi$} is $\Sigma^p_3$-complete.
\end{theorem}
\begin{proof}
    Let $\Pi = (\I, \U, \F, d, t)$ be the LOP problem of consideration and \textsc{Two-Stage Adjustable-$\Pi$} be the corresponding two-stage adjustable version.  
    The containment of \textsc{Two-Stage Adjustable-$\Pi$} in the class $\Sigma^p_3$ follows from \cref{thm:two-stageInSigma3}.
    
    Let $\Pi' = (\I, \U, \sol)$ be the SSP problem derived from the LOP problem $\Pi$.
    Recall that by the definition of the derived SSP problem, we have that $\Pi$ and $\Pi'$ share the same input instances and universe, and for $I \in \I$, we have $\sol(I) = \set{F \in \F(I) : d^{(I)}(F) \leq t^{(I)}}$. 
    
    Let \textsc{Comb. Two-Stage Adjustable-$\Pi'$} be the corresponding combinatorial two-stage adjustable version.
    By assumption on $\Pi$ and by the previous section, \textsc{Comb. Two-Stage Adjustable-$\Pi'$} is $\Sigma^p_3$-complete.
    The goal is now to reduce \textsc{Comb. Two-Stage Adjustable-$\Pi'$} to \textsc{Two-Stage Adjustable-$\Pi$}.
    This reduction is defined the following way:
    Let an instance $(I, U_1, U_2, B, \Gamma)$ of \textsc{Comb. Two-Stage Adjustable-$\Pi'$} be given.
    Recall that associated to the underlying instance $I$ of $\Pi$ there is the cost function $d^{(I)}$ and threshold $t^{(I)}$.
    We let the underlying instance $I$ of $\Pi$ remain the same and define an instance of \textsc{Two-Stage Adjustable-$\Pi$} via the following tuple $(c_1, \underline c_2, \overline c_2, \Gamma', t_{TS})$:
    \begin{align*}
        &c_1(u) =            \begin{cases}
                                d^{(I)}(u), u \in U_1\\
                                t^{(I)}+1, u \in U_2
                            \end{cases}
        \\
        &\underline c_2(u) = \begin{cases}
                                t^{(I)}+1, u \in U_1\\
                                d^{(I)}(u), u \in U_2
                            \end{cases}
        \quad
        \overline c_2(u) =  \begin{cases}
                                t^{(I)}+1, u \in U_1\\
                                t^{(I)}+1, u \in B\\
                                d^{(I)}(u), u \in U_2 \setminus B
                            \end{cases}
    \end{align*}
    At last, we set $t_{TS} = t^{(I)}$ and $\Gamma' = \Gamma$ and let $C_\Gamma$ be the discrete budgeted uncertainty set defined by $\underline c, \overline c$.
    This completes the description of the  \textsc{Two-Stage Adjustable-$\Pi$} instance.
    Recall that this instance by definition is a Yes-instance if and only if the following inequality is true:

    \begin{equation}
        \min_{S_1 \subseteq \U(I)} \max_{c_2 \in C_\Gamma} \min_{\substack{S_2 \subseteq \U(I) \setminus S_1\\ S_1 \cup S_2 \in \F(I)}} c_1(S_1) + c_2(S_2) \leq t_{TS} \label{eq:proof-two-stage-cost-version}
    \end{equation}

    To prove the correctness of the reduction, observe that whenever an element has cost of $t^{(I)}+1$, the element cannot be in a solution because the threshold $t_{TS}$ is set to $t^{(I)}$.

    Now, assume that inequality~(\ref{eq:proof-two-stage-cost-version}) is true. Then the solution $S_1$ in the first stage has to meet the condition $S_1 \subseteq U_1$.
    For a similar reason, the second stage solution $S_2$ has to meet the condition $S_2 \subseteq \U(I) \setminus (U_1 \cup B') = U_2 \setminus B'$, where $B' \subseteq B$ denotes the set of elements whose costs were increased in the second stage. Note that $|B'| \leq \Gamma$.
    These two facts together imply $c_1(S_1) + c_2(S_2) = d^{(I)}(S_1 \cup S_2)$. 
    This in turn implies that $S_1 \cup S_2$ is not only contained in $\F(I)$ as described in inequality~(\ref{eq:proof-two-stage-cost-version}), but we even have the stronger condition $S_1 \cup S_2 \in \sol(I)$. 
    All the arguments above together show that $(I, U_1, U_2, B, \Gamma)$ is a Yes-instance of \textsc{Comb. Two-Stage Adjustable-$\Pi$}.

    For the other direction, assume that $(I, U_1, U_2, B, \Gamma)$ is a Yes-instance of \textsc{Comb. Two-Stage Adjustable-$\Pi'$}. We can argue in a very similar way to the above that the inequality~(\ref{eq:proof-two-stage-cost-version}) is true.
    At last, this transformation is computable in polynomial time, because the underlying instance is unchanged and only numbers encoded in polynomial size ($\Pi$ is in NP) are added.
    In total, we get that the two instances are equivalent, thus the reduction is correct and \textsc{Two-Stage Adjustable-$\Pi$} is $\Sigma^p_3$-complete.
\end{proof}

\subsection{The Meta-Reduction is an SSP reduction}
The meta-reduction from \Cref{thm:two-stageSigma3Hard} again relies heavily on the one-to-one correspondence between the set of first stage elements and second stage elements as well as the set of blockable elements of the problems \textsc{Comb. Two-Stage Adjustable-Sat} and \textsc{Comb. Two-Stage Adjustable-$\Pi$}.
Because all of the sets from above are subsets of the universe, we are able to prove that the meta-reduction is also an SSP reduction.
For this, we first show that the combinatorial two-stage adjustable variant of an SSP problem again is an SSP problem.

\begin{observation}\label{obs:two-stageIsSSP}
\label{lem:two-stage-as-SSP}
    The combinatorial two-stage variant \textsc{Comb. Two-Stage Adjustable-$\Pi$} of an SSP problem $\Pi$ is an SSP problem.
\end{observation}
\begin{proof}
    Let $\Pi = (\I, \U, \sol)$ be an SSP problem and denote \textsc{Comb. Two-Stage Adjustable-$\Pi$} $= (\I', \U', \sol')$.
    Then, set $\U = \U'$.
    Now, let $I \in \I$ be an instance of $\Pi$.
    Then, we can define the corresponding instances of \textsc{Comb. Two-Stage Adjustable-$\Pi$} by setting $I' = \{(I, U_1, U_2, B, \Gamma) \mid I \in \I, U_1 \subseteq \U, U_2 = \U \setminus U_1, B \subseteq U_2, \Gamma \in \Z\}$.
    We define the solutions $\sol'(I')$ to be all sets $S_1 \subseteq U_1 \subseteq \U'(I')$ such that for all $B' \subseteq B$ with $|B'| \leq \Gamma$ there is $S_2 \subseteq U_2$ such that $S_2 \cap B' = \emptyset$ and $S_1 \cup S_2 \in \sol(I)$.
    Thus, $\textsc{Two-Stage Adjustable-}\Pi = (\I', \U', \sol')$ is an SSP problem.
\end{proof}

With this additional observation, we are able to elegantly prove the meta-reduction from any SSP-NP-complete problem $\Pi$ by extending the existing SSP reduction $\textsc{Sat} \leqSSP \Pi$ to be an SSP reduction \textsc{Comb. Two-Stage Adjustable-Sat} $\leqSSP$ \textsc{Comb. Two-Stage Adjustable-$\Pi$}.

\begin{corollary}\label{thm:two-stageSigma3Hard:SSP}
    For all SSP-NP-complete problems $\Pi$, \textsc{Comb. Two-Stage Adjust\-able-Sat} $\leqSSP$ \textsc{Comb. Two-Stage Adjustable-$\Pi$}.
\end{corollary}
\begin{proof}
    By \Cref{lem:two-stage-as-SSP}, \textsc{Comb. Two-Stage Adjustable-Sat} and \textsc{Comb. Two-Stage Adjustable-$\Pi$} are SSP problems and we design an SSP reduction $(g', (f_{I'})_{I' \in \I'})$ from \textsc{Comb. Two-Stage Adjustable-Sat} to \textsc{Comb. Two-Stage Adjustable-$\Pi$} by extending the existing reduction $(g, (f_I)_{I \in \I})$ as depicted in \Cref{fig:two-stage-ssp-reduction}, where $g'$ is the reduction from the proof of \Cref{thm:two-stageSigma3Hard}.
    
    \begin{figure}[!ht]
        \centering
        \scalebox{1}{
            \begin{tikzpicture}
                \node[] (Sat) at (0,0) {\textsc{Sat}};
                \node[] (Pi) at (6,0) {$\Pi$};
                \node[] (TSSat) at (-1.5,-2) {\textsc{Comb. Two-Stage Adj.-Sat}};
                \node[] (TSPi) at (7.5,-2) {\textsc{Comb. Two-Stage Adj.-$\Pi$}};
                \draw[->] (Sat) to node[above] {$(g,(f_I)_{I \in \I})$} (Pi);
                \draw[->] (TSSat) to node[above] {$(g',(f'_{I'})_{I' \in \I'})$} (TSPi);
                \draw[->] (0,-0.5) to node[left] {$I' = (I, U_1 , U_2, B, \Gamma)$} (0,-1.5);
                \draw[->] (6,-0.5) to (6,-1.5);
            \end{tikzpicture}
        }
        \caption{The fact that \textsc{Sat} is SSP reducible to $\Pi$ induces an SSP reduction from \textsc{Comb. Two-Stage Adjustable-Sat} to \textsc{Comb. Two-Stage Adjustable-$\Pi$}.}
        \label{fig:two-stage-ssp-reduction}
    \end{figure}
    We recall the reduction from the proof of \Cref{thm:two-stageSigma3Hard}.
    As we have already proven the correctness of $g'$, we focus solely on the function $(f'_{I'})_{I' \in \I'}$.
    We define $f_I = f'_{I'}$ such that the following holds:
    $$
        U'_1 = f'_{I'}(U_1), \ U_2' = \U(g(I)) \setminus U_1', \ B' = f'_{I'}(B).
    $$
    The function $f'_{I'}$ is well-defined because the universe $\U$ of the \textsc{Sat} instance $I$ and the \textsc{Comb. Two-Stage Adjustable-Sat} instance $I'$ are the same.
    Therefore, $(g', (f'_{I'})_{I' \in \I'})$ is an SSP reduction.
\end{proof}

\section{SSP Reductions for Various Problems}\label{sec:ssp-reductions}
In this section, we present a multitude of SSP-NP-complete problems.
For this, we provide SSP definitions of all the corresponding problems, i.e. for each SSP problem we describe its set $\I$ of input instances, as well as the universe $\U(I)$ and solution set $\sol(I)$ associated to each instance $I \in \I$.
We show SSP reductions between all the problems, hence proving that all considered problems are contained in the class SSP-NPc.
In this section we distinguish between SSP problems which are derived from an LOP problem, and those who are not (see \cref{sec:ssp-reduction-intro} for an explanation).
If the SSP problem is derived from an LOP problem, we additionally provide the set $\F(I)$ of feasible solutions.
In all of these cases, the cost function $d^{(I)}$ and the threshold $t^{(I)}$ of the original LOP problem can be derived from the context.

All reductions presented in this section are already known and can be found in the literature.
For each existing reduction, we shortly describe the known reduction $g: \I \rightarrow \I'$ as well as the functions $(f_I)_{I \in \I}: \U(I) \rightarrow \U'(g(I))$.
However, for the sake of conciseness, we do not always show the correctness of the original reductions explicitly and focus only on the correct embedding into the SSP framework.
In \Cref{fig:reductions}, the tree of all presented reductions can be found, beginning at \textsc{Satisfiability}.
We heavily use the transitivity of SSP reductions (\cref{lem:SSP-transitive}).

\begin{samepage}
    \begin{mdframed}
    	\begin{description}
        \item[]\textsc{Satisfiability}\hfill\\
        \textbf{Instances:} Literal Set $L = \fromto{\ell_1}{\ell_n} \cup \fromto{\overline \ell_1}{\overline \ell_n}$, Clauses $C \subseteq \powerset{L}$.\\
        \textbf{Universe:} $L =: \U$.\\
        \textbf{Solution set:} The set of all sets $L' \subseteq \U$ such that for all $i \in \fromto{1}{n}$ we have $|L' \cap \set{\ell_i, \overline \ell_i}| = 1$, and such that $|L' \cap c_j| \geq 1$ for all $c_j \in C$, $j \in \fromto{1}{|C|}$.
    	\end{description}
    \end{mdframed}
\end{samepage}

\begin{figure}[!ht]
	\centering
	\scalebox{0.67}{
	\begin{tikzpicture}[scale=1]
		\node[text width=1cm,align=center](sat) at (0, 0.75) {\textsc{Sat}};
		\node[text width=1cm,align=center](0) at (0, 0) {\textsc{3Sat}};
		\node[](1) at (-8, -1) {\textsc{VC}};
		\node[](2) at (-2.75, -1) {\textsc{IS}};
		\node[](3) at (0, -1) {\textsc{Subset Sum}};
		\node[](4) at (6.5, -0.85) {}; \node[text width=1cm,align=center](4-1) at (7, -1.19) {\textsc{Steiner Tree}};
		\node[](5) at (3.25, -1) {\textsc{DHamPath}};
		\node[](6) at (5.5, -1) {\textsc{2DDP}};
		\node[](11) at (-12.5, -2) {\textsc{DS}};
		\node[](12) at (-11.5, -2) {\textsc{SC}};
		\node[](13) at (-10.5, -2) {\textsc{HS}};
		\node[](14) at (-9.5, -2) {\textsc{FVS}};
		\node[](15) at (-8.5, -2) {\textsc{FAS}};
		\node[](10) at (-7.6, -2) {\textsc{UFL}};
		\node[](16) at (-6.25, -2) {\textsc{p-Center}};
		\node[](17) at (-4.5, -2) {\textsc{p-Median}};
		\node[](21) at (-2.75, -2) {\textsc{Clique}};
		\node[](31) at (-1, -2) {\textsc{Knapsack}};
		\node[](32) at (1, -2) {\textsc{Partition}};
		\node[](51) at (3.25, -2) {\textsc{DHamCyc}};
		\node[](61) at (5.5, -2) {\textsc{kDDP}};
		\node[](321) at (1, -2.75) {\textsc{Scheduling}};
		\node[](511) at (3.25, -2.75) {\textsc{UHamCyc}};
        \node[](5111) at (3.25, -3.5) {\textsc{TSP}};
		\path[->] (sat) edge (0);
		\path[->] (0) edge (1);
		\path[->] (0) edge (2);
		\path[->] (0) edge (3);
		\path[->] (0) edge (4);
		\path[->] (0) edge (5);
		\path[->] (0) edge (6);
		\path[->] (1) edge (10);
		\path[->] (1) edge (11);
		\path[->] (1) edge (12);
		\path[->] (1) edge (13);
		\path[->] (1) edge (14);
		\path[->] (1) edge (15);
		\path[->] (1) edge (16);
		\path[->] (1) edge (17);
		\path[->] (2) edge (21);
		\path[->] (3) edge (31);
		\path[->] (3) edge (32);
		\path[->] (32) edge (321);
		\path[->] (5) edge (51);
		\path[->] (51) edge (511);
		\path[->] (511) edge (5111);
		\path[->] (6) edge (61);
	\end{tikzpicture}
	}
	\caption{The tree of SSP reductions for all considered problems.}
	\label{fig:reductions}
\end{figure}

With the tree of SSP reductions shown in \Cref{fig:reductions}, we derive the following theorem.
\begin{samepage}
    \begin{theorem}
        The following problems are SSP-NP-complete:
        \textsc{Satisfiability},
        \textsc{3-Satis\-fiability},
        \textsc{Vertex Cover}, \textsc{Dominating Set}, \textsc{Set Cover}, \textsc{Hitting Set}, \textsc{Feedback Vertex Set}, \textsc{Feedback Arc Set}, \textsc{Uncapacitated Facility Location}, \textsc{p-Center}, \textsc{p-Median},
        \textsc{Independent Set}, \textsc{Clique},
        \textsc{Subset Sum}, \textsc{Knapsack}, \textsc{Partition}, \textsc{Scheduling},
        \textsc{Directed Hamiltonian Path}, \textsc{Directed Hamiltonian Cycle}, \textsc{Undirected Hamiltonian Cycle}, \textsc{Traveling Salesman Problem},
        \textsc{Two Directed Vertex Disjoint Path}, \textsc{$k$-Vertex Directed Disjoint Path},
        \textsc{Steiner Tree}
    \end{theorem}
\end{samepage}

\begin{samepage}
    \begin{mdframed}
    	\begin{description}
        \item[]\textsc{3-Satisfiability}\hfill\\
        \textbf{Instances:} Literal Set $L = \fromto{\ell_1}{\ell_n} \cup \fromto{\overline \ell_1}{\overline \ell_n}$, Clauses $C \subseteq \powerset{L}$ s.t. $\forall c_j \in C : |c_j| = 3$.\\
        \textbf{Universe:} $L =: \U$.\\
        \textbf{Solution set:} The set of all sets $L' \subseteq \U$ such that for all $i \in \fromto{1}{n}$ we have $|L' \cap \set{\ell_i, \overline \ell_i}| = 1$, and such that $|L' \cap c_j| \geq 1$ for all $c_j \in C$.
    	\end{description}
    \end{mdframed}
\end{samepage}
We begin with Karp's reduction \cite{DBLP:conf/coco/Karp72} from \textsc{Satisfiability} to \textsc{3-Satisfiability}. We claim that this reduction is an SSP reduction.
Let $I = (L, C)$ be the \textsc{Sat} instance and $(L', C')$ be the corresponding \textsc{3Sat} instance.
The reduction maps each clause $c \in C$ of more than three literals to a set of clauses in $C'$ of length three by introducing helper variables $h_1, h_2, \ldots$ splitting the clause into smaller clauses.
Every clause $\{a, b, c, d, \ldots\} \in C$ with more than three literals is recursively split until there are no more clauses with more than three literals as follows:
$$
    \{a, b, c, d, \ldots\} \mapsto \{a, b, h_i\}, \{\overline h_i, c, d, \ldots\}.
$$
The number of splits is bounded by the length of the instance, thus it is computable in polynomial time.

The literals $L$ of the \textsc{Sat} instance remain in the \textsc{3Sat} instance and are one-to-one correspondent.
Thus, the corresponding solutions have the SSP property by defining the functions $(f_I)_{I \in \I}$ by $f_I(\ell) = \ell \in L'$ for $\ell \in L$.
Note that the set $f_I(L)$ is exactly the set of positive and negative literals corresponding to non-helper variables.
It is easily verified that each satisfying assignment of the \textsc{3Sat} instance restricted to the set $f_I(L)$ implies a satisfying assignment of the \textsc{Sat} instance (i.e. forgetting about the helper variables). This is property (P2) as explained in \cref{sec:technicalOverview}.
Likewise, each satisfying assignment of the \textsc{Sat} instance can be completed to a satisfying assignment of the \textsc{3Sat} instance by setting the helper variables appropriately. This is property (P1).
Hence, the solutions of \textsc{Sat} and \textsc{3Sat} correspond one-to-one to each other on the set $f_I(L)$, in the precise sense that $\set{f_I(S) : S \in \sol_\text{Sat}} = \set{S'  \cap f_I(L) : S' \in \sol_\text{3Sat}}$.
Thus this reduction is indeed an SSP reduction.

\begin{samepage}
    \begin{mdframed}
    	\begin{description}
        \item[]\textsc{Vertex Cover}\hfill\\
        \textbf{Instances:} Graph $G = (V, E)$, number $k \in \N$.\\
        \textbf{Universe:} Vertex set $V =: \U$.\\
        \textbf{Feasible solution set:} The set of all vertex covers.\\
        \textbf{Solution set:} The set of all vertex covers of size at most $k$.
    	\end{description}
    \end{mdframed}
\end{samepage}
The reduction of Garey and Johnson \cite{DBLP:books/fm/GareyJ79} from \textsc{3Sat} to \textsc{Vertex Cover} is an SSP reduction.
In order to show this, we reformulate the reduction and adapt it to the SSP framework.
Let $I = (L,C)$ be the \textsc{3Sat} instance with literals $L$ and clauses $C$.
We define the corresponding instance $g(L,C)$ of \textsc{Vertex Cover} as the following tuple $(G',k') = ((V', E'), k')$.
Each literal $\ell \in L$ is mapped to a literal vertex $v_\ell \in V'$, where $v_\ell$ and its negation $v_{\overline \ell}$ are connected by an edge $\{v_\ell, v_{\overline \ell}\} \in E'$.
Furthermore, we introduce a 3-clique for each clause $c \in C$.
Each of the three vertices $v^{c}_{\ell_{i_1}}, v^{c}_{\ell_{i_2}}, v^{c}_{\ell_{i_3}}$ represents a literal in the clause, and is then connected to the corresponding literal vertex, i.e. $\{v_\ell, v^c_\ell\} \in E'$ for $\ell \in c$.
Finally, we define the parameter $k'$ by $k' := |L|/2 + 2|C|$.
An example instance can be found in \Cref{fig:reduction:3sat-vertex-cover}, where the set of literal vertices is denoted by $W$.

\tikzstyle{vertex}=[draw,circle,fill=black, minimum size=4pt,inner sep=0pt]
\tikzstyle{edge} = [draw,-]
\begin{figure}[thpb]
\centering
\resizebox{0.67\textwidth}{!}{
\begin{tikzpicture}[scale=1,auto]

\node[vertex] (x1) at (0,0) {}; \node[above] at (x1) {$v_{\ell_1}$};
\node[vertex] (notx1) at (2,0) {}; \node[above] at (notx1) {$v_{\overline \ell_1}$};
\draw[edge] (x1) to (notx1);

\node[vertex] (x2) at (4,0) {}; \node[above] at (x2) {$v_{\ell_2}$};
\node[vertex] (notx2) at (6,0) {}; \node[above] at (notx2) {$v_{\overline \ell_2}$};
\draw[edge] (x2) to (notx2);

\node[vertex] (x3) at (8,0) {}; \node[above] at (x3) {$v_{\ell_3}$};
\node[vertex] (notx3) at (10,0) {}; \node[above] at (notx3) {$v_{\overline \ell_3}$};
\draw[edge] (x3) to (notx3);

\node[vertex] (c1) at (4,-2.25) {}; \node[below] at (c1) {$v^{c_1}_{\overline \ell_1}$};
\node[vertex] (c2) at (5,-1.25) {}; \node[above left] at (c2) {$v^{c_1}_{\overline \ell_2}$};
\node[vertex] (c3) at (6,-2.25) {}; \node[below] at (c3) {$v^{c_1}_{\ell_3}$};
\node[] at (5,-1.92) {$c_1$};
\draw[edge] (c1) to (c2) to (c3) to (c1);
\draw[edge] (notx1) to (c1);
\draw[edge] (notx2) to (c2);
\draw[edge] (x3) to (c3);

\node at ($(x1)+(-1,0)$) {$W$};
\draw[dashed,rounded corners] ($(x1)+(-.5,+.7)$) rectangle ($(notx3) + (.5,-.4)$);

\end{tikzpicture}
}
\caption{Classic reduction of \textsc{3Sat} to \textsc{Vertex Cover} for $\varphi = (\overline \ell_1 \lor \overline \ell_2 \lor \ell_3)$.}
\label{fig:reduction:3sat-vertex-cover}
\end{figure}

The universe elements of \textsc{3Sat} are injectively mapped to the literal vertices in $W \subseteq V'$, where $f_{I}(\ell) = v_\ell$.
All valid solutions (if there are any) include exactly one of $v_\ell$ or $v_{\overline \ell} \in W$ and two additional vertices from the 3-clique corresponding to each clause.
To cover a 3-clique at least two vertices of that 3-clique have to be in the vertex cover.
If a clause is not satisfied, then no neighboring literal vertex is in the solution. In this case all three vertices of the 3-clique simulating the clause have to be taken into the solution
(otherwise the edges connecting the literal vertices with the 3-clique are not covered).
But then this vertex cover must already have size more than $k'$.
In total, we have that every vertex cover of size at most $k'$ restricted to the set $W$ corresponds to a solution of \textsc{3Sat}.
On the other hand, every solution of \textsc{3Sat} can be transferred over to the set $W$ and be completed in at least one way to a vertex cover of size at most $k'$, i.e. the following equation holds true
\begin{align*}
    \{f_{I}(S) : S \subseteq L \ \text{s.t.} \ S \in \sol_{\textsc{3Sat}}\} 
    &= \{S' \cap f_{I}(L) : S' \in \sol_{VC}\}.
\end{align*}
Thus, the SSP reduction is indeed correct.

\begin{samepage}
    \begin{mdframed}
    	\begin{description}
        \item[]\textsc{Dominating Set}\hfill\\
        \textbf{Instances:} Graph $G = (V, E)$, number $k \in \N$.\\
        \textbf{Universe:} Vertex set $V =: \U$.\\
        \textbf{Feasible solution set:} The set of all dominating sets.\\
        \textbf{Solution set:} The set of all dominating sets of size at most $k$.
    	\end{description}
    \end{mdframed}
\end{samepage}
For a reduction from \textsc{Vertex Cover} to \textsc{Dominating Set}, we use a (modified) folklore reduction as depicted in \Cref{fig:reduction:vertex-cover-dominating-set}.
Let $I = ((V, E), k)$ be the \textsc{Vertex Cover} instance and $((V', E'), k')$ be the \textsc{Dominating Set} instance.
For every two vertices $v, w \in V$ connected by an edge $\{v, w\}$, we introduce $|V|+1$ additional vertices $vw_i$ for $i \in \fromto{1}{|V|+1}$ and connect them to $v$ and to $w$.
All isolated vertices $v \in V$ are mapped to itself, that is $v \in V'$.
Furthermore, we introduce a star around vertex $v_{iso}$ connected to vertices $v^i_{iso}$ for $i \in \fromto{1}{|V|+1}$.
Then, we connect $v_{iso}$ to all the isolated vertices from $V'$.
The parameter $k'$ is set to $k' = k+1$.
With this construction, the vertex cover is directly translatable to a dominating set in the \textsc{Dominating Set} instance by leaving it as is and by taking $v_{iso}$ into the dominating set.
Note that $v_{iso}$ dominates itself, all $v^i_{iso}$ for $i \in \fromto{1}{|V|+1}$, and all originally isolated vertices.
The other way around, we claim that every dominating set of size at most $k' = k+1$ in the new graph has the property that restricted to the set $W$ it encodes a vertex cover of size at most $k$ of the old graph.
Indeed, observe that one needs at least $|V|+1$ vertices to dominate all vertices $v, w$ and $vw_i$ $i \in \fromto{1}{|V|+1}$ of one edge, the same holds for the star around $v_{iso}$.
Thus for all dominating sets and for all \enquote{original} edges $\set{v,w}$, one needs to include either $v$ or $w$.
Note that $v_{iso}$ is always part of the dominating set because it is the center of large star and all of the originally isolated vertices are dominated.
Consequently, we have a one-to-one correspondence between the vertex cover and the dominating set and the solutions are preserved accordingly.
This one-to-one correspondence with the mapping $f_I(v) = v \in V'$ for all $v \in V$ directly implies that this is an SSP reduction.

\tikzstyle{vertex}=[draw,circle,fill=black, minimum size=4pt,inner sep=0pt]
\tikzstyle{edge} = [draw,-]
\begin{figure}[thpb]
\centering
\begin{tikzpicture}[scale=1,auto]

\node[vertex] (v) at (0,0) {}; \node[above] at (v) {$v$};
\node[vertex] (w) at (2,0) {}; \node[above] at (w) {$w$};
\draw[edge] (v) to (w);

\node[vertex] (vw1) at (1,-0.67) {}; \node[below] at (vw1) {$vw_1$};
\draw[edge] (v) to (vw1);
\draw[edge] (vw1) to (w);
\node[vertex] (vwV+1) at (1,-2) {}; \node[below] at (vwV+1) {$vw_{|V|+1}$};
\draw[edge] (v) to (vwV+1);
\draw[edge] (vwV+1) to (w);
\node[] () at (1,-1.25) {$\vdots$};

\node at ($(v)+(-1,0)$) {$W$};
\draw[dashed,rounded corners] ($(v)+(-.5,+.6)$) rectangle ($(w) + (.5,-.4)$);

\end{tikzpicture}
\caption{The modified reduction of \textsc{Vertex Cover} to \textsc{Dominating Set}.}
\label{fig:reduction:vertex-cover-dominating-set}
\end{figure}

\begin{samepage}
    \begin{mdframed}
    	\begin{description}
        \item[]\textsc{Set Cover}\hfill\\
        \textbf{Instances:} Sets $S_i \subseteq \fromto{1}{m}$ for $i \in \fromto{1}{n}$, number $k \in \N$.\\
        \textbf{Universe:} $\{S_1 \dots, S_n\} =: \U$.\\
        \textbf{Feasible solution set:} The set of all $S \subseteq \{S_1, \dots, S_n\}$ s.t. $\bigcup_{s \in S} s = \fromto{1}{m}$.\\
        \textbf{Solution set:} Set of all feasible solutions with $|S| \leq k$.
    	\end{description}
    \end{mdframed}
\end{samepage}
The reduction from \textsc{Vertex Cover} to \textsc{Set Cover} by Karp \cite{DBLP:conf/coco/Karp72} is an SSP reduction.
\textsc{Vertex Cover} and \textsc{Set Cover} are basically the same problem, which means that the syntax of the input is the same, however the semantics behind the encoding are different.
Thus, the reduction of Karp implies a direct one-to-one correspondence not only between the universe elements but also between the edges and the sets.
Let $I = ((V, E), k)$ be the \textsc{Vertex Cover} instance.
In the reduction, each vertex $v \in V$ is mapped to the set $S_v$ and each edge $e \in E$ is mapped to the set $\fromto{1}{|E|} = \fromto{1}{m}$ according to their index.
Each set $S_v$ includes its the indices of the incident edges.
Thus, if a vertex $v \in V$ is taken into the vertex cover all incident edges are covered which is equivalent to including $S_v$ into the set cover such that all elements of $S_v$ are covered, which are exactly the indices of all incident edges to $v$.
Consequently, the one-to-one correspondence is defined by $f_I(v) = S_v$ for $v \in V$ as desired, which also implies that this is an SSP reduction.

\begin{samepage}
    \begin{mdframed}
    	\begin{description}   
        \item[]\textsc{Hitting Set}\hfill\\
        \textbf{Instances:} Sets $S_j \subseteq \fromto{1}{n}$ for $j \in \fromto{1}{m}$, number $k \in \N$.\\
        \textbf{Universe:} $\fromto{1}{n} =: \U$.\\
        \textbf{Feasible solution set:} The set of all $H \subseteq \fromto{1}{n}$ such that $H \cap S_j \neq \emptyset$ for all $j \in \fromto{1}{m}$.\\
        \textbf{Solution set:} Set of all feasible solutions with $|H| \leq k$.
    	\end{description}
    \end{mdframed}
\end{samepage}
Karp's reduction \cite{DBLP:conf/coco/Karp72} from \textsc{Vertex Cover} to \textsc{Hitting Set} is an SSP reduction.
Similar to \textsc{Set Cover}, \textsc{Hitting Set} is basically the same problem as \textsc{Vertex Cover}.
We only have to reinterpret the semantic of the input encoding as follows.
Let $I = ((V, E), k)$ be the \textsc{Vertex Cover} instance.
Each vertex $v \in V$ is mapped into the set $\fromto{1}{|V|}$ by its index $id(v)$.
Then, each edge $e = \{v, w\} \in E$ is mapped to the set $S_e = \{id(v), id(w)\}$ (thus $m = |E|$).
It follows that each vertex is exactly one-to-one correspondent to its index and the solutions are preserved because every vertex cover covers all edges which is equivalent that the corresponding hitting set induced by this vertex cover hits all $S_e$.
Consequently, we can define $f_I(v) = id(v)$ and we have an SSP reduction.

\begin{samepage}
    \begin{mdframed}
    	\begin{description}
        \item[]\textsc{Feedback Vertex Set}\hfill\\
        \textbf{Instances:} Directed Graph $G = (V, A)$, number $k \in \N$.\\
        \textbf{Universe:} Vertex set $V =: \U$.\\
        \textbf{Feasible solution set:} The set of all vertex sets $V' \subseteq V$ such that after deleting $V'$ from $G$, the resulting graph is cycle-free (i.e. a forest).\\
        \textbf{Solution set:} The set of all feasible solutions $V'$ of size at most $k$.
    	\end{description}
    \end{mdframed}
\end{samepage}
The reduction by Karp \cite{DBLP:conf/coco/Karp72} from \textsc{Vertex Cover} to \textsc{Feedback Vertex Set} is an SSP reduction.
Let $I = (G,k) = ((V, E), k)$ be the \textsc{Vertex Cover} instance and $(G', k') = ((V', A'), k')$ the \textsc{Feedback Vertex Set} instance.
The transformation maps every vertex $v \in V$ to itself ($v \in V'$) and every edge $\{v, w\} \in E$ is mapped to two arcs $(v,w), (w,v) \in A'$ orienting the edge in both directions.
We further set $k = k'$.
We define the injective embedding function $f_I$ by the identity on the vertices, i.e.\ every vertex in $V$ is mapped onto its corresponding twin in $V'$.
Note that the solutions are also directly one-to-one transformable and thus preserved.
To see this, assume to have a vertex cover for the graph $G$, then the same set removes all cycles from the directed graph $G'$, because a vertex cover is incident to all edges $E$ in $G$ and thus to all arcs $A'$ in $G'$ resulting in an independent set, which is obviously cycle-free.
On the other hand, a solution to the \textsc{Feedback Vertex Set} instance $G'$ needs to remove all cycles.
Each cycle is induced by two vertices connected by both arcs $(v,w), (w,v) \in A'$.
Thus, at least one vertex of $v$ and $w$ has to be deleted such that both arcs $(v,w), (w,v) \in A'$ are also deleted and do not form a cycle.
This, however, is obviously equivalent to a vertex cover in $G$.

\begin{samepage}
    \begin{mdframed}
    	\begin{description}
        \item[]\textsc{Feedback Arc Set}\hfill\\
        \textbf{Instances:} Directed Graph $G = (V, A)$, number $k \in \N$.\\
        \textbf{Universe:} Arc set $A =: \U$.\\
        \textbf{Feasible solution set:} The set of all arc sets $A' \subseteq A$ such that after deleting $A'$ from $G$, the resulting graph is cycle-free (i.e. a forest).\\
        \textbf{Solution set:} The set of all feasible solutions $A'$ of size at most $k$.
    	\end{description}
    \end{mdframed}
\end{samepage}
A modification of the reduction by Karp \cite{DBLP:conf/coco/Karp72} from \textsc{Vertex Cover} to \textsc{Feedback Arc Set} is an SSP reduction.
Let $I = (G,k) = ((V, E),k)$ be the \textsc{Vertex Cover} instance and $(G',k') = ((V', A'),k')$ the \textsc{Feedback Arc Set} instance.
This reduction is more complicated in the SSP framework in comparison to the reductions, we have seen before.
Due to the fact that the universe is changed from the vertex set to the arc set, we have to be more careful in analyzing the individual mappings.
First of all, we transform the vertices $v \in V$ to two vertices $v_0, v_1 \in V'$.
We define the injective embedding function $f_I$ by mapping each vertex $v \in V$ to the arc $(v_0, v_1) \in A'$, which is also the corresponding element in all solutions, that is $f_I(v) = (v_0, v_1)$.
At last, we transform each edge $\{v, w\} \in E$ to $|V + 1|$ once subdivided arcs from $v_1$ to $w_0$ and to $|V + 1|$ once subdivided arcs from $w_1$ to  $v_0$. Finally, we leave the parameter $k = k'$ unchanged. This completes the description of the reduction.
We denote the vertices added by the subdivision $v^i_1$ for the arcs between $v_1$ to $w_0$ and $w^i_1$ between $w_1$ to  $v_0$ for $i \in \fromto{1}{|V| + 1}.$
Overall, one vertex pair with an edge induces more than $|V|$ cycles $(v_0, v_1), (v_1, v^i_1) (v^i_1, w_0), (w_0, w_1), (w_1, w^i_1), (w^i_1, v_0)$ of length six.
By deleting the arc $(v_0, v_1)$, which corresponds to vertex $v$ in $G$ all of these induced cycles are disconnected. This implies that every vertex cover of $G$ is translated to a feedback arc set by the function $f_I$.
On the other hand, a feedback arc set must contain for every original edge $\set{v, w}$ either the arc $(w_0, w_1)$ or the arc $(_0, v_1)$, because otherwise a cycle remains. This shows that every feedback arc set of size at most $k$, when restricted to the set $f_I(V)$ encodes a vertex cover of $G$ of size at most $k$.
Hence we have an SSP reduction.

\begin{samepage}
    \begin{mdframed}
    	\begin{description} 
        \item[]\textsc{Uncapacitated Facility Location}\hfill\\
        \textbf{Instances:} Set of potential facilities $F = \fromto{1}{n}$, set of clients $C = \fromto{1}{m}$, fixed cost of opening facility function $f: F \rightarrow \Z$, service cost function $c: F \times C \rightarrow \Z$, cost threshold $k \in \Z$\\
        \textbf{Universe:} Facility set $F =: \U$.\\
        \textbf{Solution set:} The set of sets $F' \subseteq F$ s.t. $\sum_{i \in F'} f(i) + \sum_{j \in C} \min_{i \in F'} c(i, j) \leq k$.
    	\end{description}
    \end{mdframed}
\end{samepage}
Note that we define this problem explicitly as SSP and not as LOP because in the standard interpretation, the objective function is not linear.
The reduction by Cornuéjols, Nemhauser, and Wolsey \cite{cornuejols1983uncapicitated} from \textsc{Vertex Cover} to \textsc{Uncapacitated Facility Location} is an SSP reduction.
Let $I=((V,E),k)$ be the \textsc{Vertex Cover} instance and $(F, C, f, c)$ be the \textsc{Uncapacitated Facility Location} instance.
We let $F := V$ and $C := E$. The injective embedding function $f_I$ is given by $f_I(v) = v \in F$ for $v \in V$.
Further, we define $c(v, e) = 0$ if $v \in e$ and $c(v, e) = |V|+1$ otherwise.
At last, we set $f(v) = 1$ for all $v \in F$ and leave the parameter $k$ unchanged.
The one-to-one correspondence between the solutions can be explained by analyzing the correctness of the reduction.
On the one hand, a vertex cover $S$ is a solution to the facility location problem, because at most $k$ many facilities are opened and all clients $e \in C$ are served, which corresponds exactly that all edges $e \in E$ are covered by $v \in S$.
On the other hand, if there is a facility set $F'$ with cost $k$, then it has to include at most $k$ facilities and additionally serve all clients $e \in C$, i.e. the corresponding vertex set covers all edges $e \in E$ (because of the high costs of $c(v, e) = |V|+1$ for $v \notin e$).
Thus, this is an SSP reduction.

\begin{samepage}
    \begin{mdframed}
    	\begin{description} 
        \item[]\textsc{p-Center}\hfill\\
        \textbf{Instances:} Set of potential facilities $F = \fromto{1}{n}$, set of clients $C = \fromto{1}{m}$, service cost function $c: F \times C \rightarrow \Z$, facility threshold $p \in \N$, cost threshold $k \in \Z$\\
        \textbf{Universe:} Facility set $F =: \U$.\\
        \textbf{Solution set:} The set of sets $F' \subseteq F$ s.t. $|F'| \leq p$ and $\max_{j \in C} \min_{i \in F'} c(i, j) \leq k$.
    	\end{description}
    \end{mdframed}
\end{samepage}
Like in the previous problem, we cannot interpret this problem as an LOP problem since the objective is not linear.
A modified version of the reduction by Cornuéjols, Nemhauser, and Wolsey \cite{cornuejols1983uncapicitated} from \textsc{Vertex Cover} to \textsc{Uncapacitated Facility Location} is an SSP reduction.
Let $I = ((V,E), k)$ be the \textsc{Vertex Cover} instance and $(F, C, c, p, k')$ be the \textsc{p-Center} instance.
We map each $v \in V$ to $v \in F$ and each $e \in E$ to $e \in C$.
Further, we define $c(v, e) = 0$ if $v \in e$ and $c(v, e) = |V|+1$ otherwise.
At last, we set $p$ equal to the size $k$ of the vertex cover and $k' = 0$.
Note that this implies that in a solution the objective has to be $0$.
We now argue analogous to the reduction to \textsc{Uncapacitated Facility Location}.
The embedding function $f_I$ is given by the one-to-one correspondence between the universe elements $v \in V$ and $v \in F$.
The rest of the argument is analogous to the above.

\begin{samepage}
    \begin{mdframed}
    	\begin{description} 
        \item[]\textsc{p-Median}\hfill\\
        \textbf{Instances:} Set of potential facilities $F = \fromto{1}{n}$, set of clients $C = \fromto{1}{m}$, service cost function $c: F \times C \rightarrow \Z$, facility threshold $p \in \N$, cost threshold $k \in \Z$\\
        \textbf{Universe:} Facility set $F =: \U$.\\
        \textbf{Solution set:} The set of sets $F' \subseteq F$ s.t. $|F'| \leq p$ and $\sum_{j \in C} \min_{i \in F'} c(i, j) \leq k$.
    	\end{description}
    \end{mdframed}
\end{samepage}
Like in the previous problem, we cannot interpret this problem as an LOP problem since the objective is not linear.
A modified version of the reduction by Cornuéjols, Nemhauser, and Wolsey \cite{cornuejols1983uncapicitated} from \textsc{Vertex Cover} to \textsc{Uncapacitated Facility Location} is an SSP reduction.
It is the same as for \textsc{p-Center}.
Let $I = ((V,E), k)$ be the \textsc{Vertex Cover} instance and $(F, C, c, p, k')$ be the \textsc{p-Median} instance.
We map each $v \in V$ to $v \in F$ and each $e \in E$ to $e \in C$.
Further, we define $c(v, e) = 0$ if $v \in e$ and $c(v, e) = |V|+1$ otherwise.
At last, we set $p = k$ equals to the size of the vertex cover and $k' = 0$.
The one-to-one correspondence between the universe elements $v \in V$ and $v \in F$ defines the embedding function $f_I$.
The rest of the argument is analogous to the above.

\begin{samepage}
    \begin{mdframed}
    	\begin{description} 
        \item[]\textsc{Independent Set}\hfill\\
        \textbf{Instances:} Graph $G = (V,E)$, number $k \in \N$.\\
        \textbf{Universe:} Vertex set $V =: \U$.\\
        \textbf{Feasible solution set:} The set of all independent sets.\\
        \textbf{Solution set:} The set of all independent sets of size at least $k$.
    	\end{description}
    \end{mdframed}
\end{samepage}
For a reduction from \textsc{3Sat} to \textsc{Independent Set}, we use a folklore reduction, which is based on the reduction from \textsc{3Sat} to \textsc{Vertex Cover} by Garey and Johnson \cite{DBLP:books/fm/GareyJ79}.
Let $I = (L,C)$ be the \textsc{3Sat} instance.
We define a  corresponding \textsc{Vertex Cover} instance $((V',E'), k')$.
Every literal $\ell \in L$ is transformed to a vertex $v_\ell \in V'$ and every pair $(\ell, \overline \ell) \in L \times L$ is again transformed to an edge $\{v_\ell, v_{\overline \ell}\} \in E'$ between the corresponding literal vertices.
Every clause $c \in C$ is again transformed to a 3-clique, where each vertex $v^{c}_{\ell_{i_1}}, v^{c}_{\ell_{i_2}}, v^{c}_{\ell_{i_3}}$ represents a literal in the clause.
In contrast to the \textsc{Vertex Cover} reduction, the clause vertices are connected to the opposite literal vertex. 
For example in \Cref{fig:reduction:3sat-independent-set}, we have that the 3SAT clause is given by $\overline \ell_1 \lor \overline \ell_2 \lor \ell_3$.
Hence, the clique is connected to the corresponding literal vertices $\ell_1, \ell_2$ and $\overline \ell_3$.
Finally, we define the parameter $k'$ by $k' := |L|/2 + |C|$.

\tikzstyle{vertex}=[draw,circle,fill=black, minimum size=4pt,inner sep=0pt]
\tikzstyle{edge} = [draw,-]
\begin{figure}[thpb]
\centering
\resizebox{0.67\textwidth}{!}{
\begin{tikzpicture}[scale=1,auto]

\node[vertex] (x1) at (0,0) {}; \node[above] at (x1) {$v_{\ell_1}$};
\node[vertex] (notx1) at (2,0) {}; \node[above] at (notx1) {$v_{\overline \ell_1}$};
\draw[edge] (x1) to (notx1);

\node[vertex] (x2) at (4,0) {}; \node[above] at (x2) {$v_{\ell_2}$};
\node[vertex] (notx2) at (6,0) {}; \node[above] at (notx2) {$v_{\overline \ell_2}$};
\draw[edge] (x2) to (notx2);

\node[vertex] (x3) at (8,0) {}; \node[above] at (x3) {$v_{\ell_3}$};
\node[vertex] (notx3) at (10,0) {}; \node[above] at (notx3) {$v_{\overline \ell_3}$};
\draw[edge] (x3) to (notx3);

\node[vertex] (c1) at (4,-2.25) {}; \node[below] at (c1) {$v^{c_1}_{\overline \ell_1}$};
\node[vertex] (c2) at (5,-1.25) {}; \node[above right] at (c2) {$v^{c_1}_{\overline \ell_2}$};
\node[vertex] (c3) at (6,-2.25) {}; \node[below] at (c3) {$v^{c_1}_{\ell_3}$};
\node[] at (5,-1.92) {$c_1$};
\draw[edge] (c1) to (c2) to (c3) to (c1);
\draw[edge] (x1) to (c1);
\draw[edge] (x2) to (c2);
\draw[edge] (notx3) to (c3);

\node at ($(x1)+(-1,0)$) {$W$};
\draw[dashed,rounded corners] ($(x1)+(-.5,+.7)$) rectangle ($(notx3) + (.5,-.4)$);

\end{tikzpicture}
}
\caption{Classical reduction of \textsc{3Sat} to \textsc{Independent Set} for $\varphi = (\overline \ell_1 \lor \overline \ell_2 \lor \ell_3)$.}
\label{fig:reduction:3sat-independent-set}
\end{figure}

Again, the universe elements of \textsc{3Sat} are injectively mapped to the literal vertices in $W \subseteq V'$, where $f_{I}(\ell) = v_\ell \in V'$.
All solutions included exactly one of $v_\ell$ or $v_{\overline \ell} \in W$ corresponding to $\ell$ and $\overline \ell$ in the \textsc{3Sat} solution as well as one additional vertex for each clause.
Note that every independent set has size at most $k'$, since only one vertex of every 3-clique and every 2-clique can be taken into the solution. 
Whenever a clause $c \in C$ is not satisfied, all three vertices in the clause $c$ are blocked from the independent set from the opposite literals that are in the solution.

In total, we have that every independent set of size at least $k'$ restricted to the set $W$ corresponds to a solution of \textsc{3Sat}.
On the other hand, every solution of \textsc{3Sat} can be transferred over to the set $W$ and be completed in at least one way to an independent set of size at least $k'$, i.e. the following equation holds true
\begin{align*}
    \{f_{I}(S) : S \subseteq L \ \text{s.t.} \ S \in \sol_{\textsc{3Sat}}\} = \{S' \cap f_{I}(L) : S' \in \sol_{IS}\}.
\end{align*}
Thus, the SSP reduction is correct.

\begin{samepage}
    \begin{mdframed}
    	\begin{description}   
        \item[]\textsc{Clique}\hfill\\
        \textbf{Instances:} Graph $G = (V, E)$, number $k \in \N$.\\
        \textbf{Universe:} Vertex set $V =: \U$.\\
        \textbf{Feasible solution set:} The set of all cliques.\\
        \textbf{Solution set:} The set of all cliques of size at least $k$.
    	\end{description}
    \end{mdframed}
\end{samepage}
There is a reduction by Garey and Johnson \cite{DBLP:books/fm/GareyJ79} from \textsc{Independent Set} to \textsc{Clique}, which is an SSP reduction.
Let $I = (G, k) = ((V, E), k)$ be the \textsc{Independent Set} instance and $(G', k') = ((V', E'), k')$ the \textsc{Clique} instance.
The reduction simply maps every vertex $v \in V$ to itself in $V'$.
Furthermore every edge $\{v, w\} \in E$ mapped to a non-edge $\{v, w\} \notin E'$ and every non-edge $\{v, w\} \notin E$ is mapped to an edge $\{v, w\} \in E'$.
Thus, every independent set $S \subseteq V$ in $G$ is transformed in to a clique of the same vertices $S \subseteq V = V'$ in $G'$.
By this transformation, the vertices are directly one-to-one correspondent with $f_I(v) = v$.
Thus, this reduction is an SSP reduction.

\begin{samepage}
    \begin{mdframed}
    	\begin{description}   
        \item[]\textsc{Subset Sum}\hfill\\
        \textbf{Instances:} Numbers $\fromto{a_1}{a_n} \subseteq \N$, and target value $M \in \N$.\\
        \textbf{Universe:} $\fromto{a_1}{a_n} =: \U$.\\
        \textbf{Solution set:} The set of all sets $S \subseteq \U$ with $\sum_{a_i \in S}a_i = M$.
    	\end{description}
    \end{mdframed}
\end{samepage}
The reduction by Sipser \cite{DBLP:books/daglib/0086373} from \textsc{3Sat} to \textsc{Subset Sum} is an SSP reduction.
Let $I = (L, C)$ be the \textsc{3Sat} instance.
We define a \textsc{Subset Sum} instance $(\fromto{a_1}{a_n}, M)$.
We create a table as depicted in \Cref{fig:reduction:3sat-subset-sum} to transform each literal pair $(\ell_i, \overline \ell_i)$ (or variable $x_i$) into a number whose binary representation has length $|L|/2 + |C|$.

\tikzstyle{vertex}=[draw,circle,fill=black, minimum size=4pt,inner sep=0pt]
\tikzstyle{edge} = [draw,-]
\begin{figure}[thpb]
\centering
\begin{tabular}{c|c|c|c|c|}
 & $x_1$ & $x_2$ & $x_3$ & $c_1 = \overline \ell_1 \lor \overline \ell_2 \lor \ell_3$ \\
\hline
$s_1$ & 1 & 0 & 0 & 0 \\
\hline
$s_2$ & 1 & 0 & 0 & 1 \\
\hline
$s_3$ & 0 & 1 & 0 & 0 \\
\hline
$s_4$ & 0 & 1 & 0 & 1 \\
\hline
$s_5$ & 0 & 0 & 1 & 1 \\
\hline
$s_6$ & 0 & 0 & 1 & 0 \\
\hline
$s_7$ & 0 & 0 & 0 & 1\\
\hline
$s_8$ & 0 & 0 & 0 & 2\\
\hline
$\Sigma$ & 1 & 1 & 1 & 4\\
\end{tabular}
\caption{Classical reduction of \textsc{3Sat} to \textsc{Subset Sum} for $\varphi = (\overline \ell_1 \lor \overline \ell_2 \lor \ell_3)$.}
\label{fig:reduction:3sat-subset-sum}
\end{figure}

We fix an ordering on the variables and clauses to define the table.
Each variable and each clause has a unique column $i$ defining the $i$-th digit of each number.
For each literal $\ell$, we add a number that has a 1 exactly at the position of the corresponding variable and additional 1s at the positions of clauses that contain the literal.
The target sum $M$ is to be defined as $1$ in each variable column and $4$ in each clause column.
This means that exactly one of the numbers corresponding to a literal pair can be added to the solution.
Furthermore to satisfy a clause $c \in C$, the sum of the column corresponding to $c$ has to be exactly $4$.
Thus, we add two numbers for each clause $c \in C$, one that contains a $1$ and one that contains a $2$ in the column of clause $c$.
Consequently, whenever a clause $c$ is satisfied, i.e. the sum of the columns of $c$ is greater than $1$, the sum can be expanded to exactly $4$.
This reduction can be transformed into binary (and any other) encoding as well by introducing leading zeros such that no carryover occurs.

Note that the described \textsc{Subset Sum} instance in total contains the following numbers: Two numbers $a_i, \overline a_i$ for every literal pair $(\ell_i, \overline \ell_i)$ plus some additional helper numbers. We define the injective embedding function $f_I$ by $f_I(\ell_i) = a_i$ and $f_I(\overline \ell_i) = \overline a_i$.
Note that with respect to this embedding $f_I$ we have the SSP property, i.e. every subset of numbers with total sum $M$ restricted to the set $f_I(L)$ encodes a correct solution of 3SAT.
Therefore, this reduction is an SSP reduction.

\begin{samepage}
    \begin{mdframed}
    	\begin{description}   
        \item[]\textsc{Knapsack}\hfill\\
        \textbf{Instances:} Objects with prices and weights $\fromto{(p_1, w_1)}{(p_n, w_n)} \subseteq \N^2$, and $W, P \in \N$.\\
        \textbf{Universe:} $\fromto{(p_1, w_1)}{(p_n, w_n)} =: \U$.\\
        \textbf{Feasible solution set:} The set of all $S \subseteq \U$ with $\sum_{(p_i, w_i) \in S}w_i \leq W$.\\
        \textbf{Solution set:} The set of feasible $S$ with $\sum_{(p_i, w_i) \in S} p_i \geq P$.
    	\end{description}
    \end{mdframed}
\end{samepage}
The reduction from \textsc{Subset Sum} to \textsc{Knapsack} is an easy-to-see folklore result.
The \textsc{Subset Sum} instance $I = (\{a_1, \ldots, a_n\}, M)$ can be transformed to a \textsc{Knapsack} instance $(\{(a_1, a_1), \ldots, (a_n, a_n)\}, W, P)$ of objects of the same price and weight.
Furthermore, the target value $M$ is mapped to the weight threshold $W = M$ and price threshold $P = M$.
Thus, $\sum_{(a_i, a_i) \in S} a_i \geq M$ and $\sum_{(a_i, a_i) \in S}a_i \leq M$ is equivalent to $\sum_{(a_i, a_i) \in S}a_i = M$.
The one-to-one correspondence between the number $a_i$ and the object $(a_i, a_i)$, i.e. $f_I(a_i) = (a_i, a_i)$, such that this reduction is an SSP reduction.

\begin{samepage}
    \begin{mdframed}
    	\begin{description}
        \item[]\textsc{Partition}\hfill\\
        \textbf{Instances:} Numbers $\fromto{a_1}{a_n} \subseteq \N$.\\
        \textbf{Universe:} $\fromto{a_1}{a_n} =: \U$.\\
        \textbf{Solution set:} The set of all sets $S \subseteq \U$ with $a_n \in S$ and $\sum_{a_i \in S}a_i = \sum_{a_j \notin S}a_j$.
    	\end{description}
    \end{mdframed}
\end{samepage}
Note that we demand w.l.o.g. the last element to be in the solution.
With this, we avoid symmetry of solutions, i.e. if $S$ is a solution then $\U \setminus S$ is also a solution, which is not compatible with the SSP framework in the following reduction. 
The problems \textsc{Subset Sum} and \textsc{Partition} are almost equivalent such that the reduction between them is easy-to-see.
We use basically the same reduction as Karp's \cite{DBLP:conf/coco/Karp72} from \textsc{Knapsack} to \textsc{Partition}.
For this let $I = (\{a_1, \ldots, a_n\}, M)$ be the \textsc{Subset Sum} instance.
We map each number $a_i$ to itself in the \textsc{Partition} instance and add additional numbers $M+1$ and $\sum_{i} a_i + 1 - M$, whereby we set $a_n = \sum_{i} a_i + 1 - M$.
Thus, the first $n-2$ $a_i$ in \textsc{Partition} are one-to-one correspondent with the $a_i$ from \textsc{Subset Sum}, i.e. $f_I(a_i) = a_i$, such that this reduction is an SSP reduction.

\begin{samepage}
    \begin{mdframed}
    	\begin{description}
        \item[]\textsc{Two Machine Scheduling}\hfill\\
        \textbf{Instances:} Jobs with processing time $\fromto{t_1}{t_n} \subseteq \N$, threshold $T \in \N$.\\
        \textbf{Universe:} The set of jobs $\fromto{t_1}{t_n} =: \U$.\\
        \textbf{Solution set:} The set of all $J_1 \subseteq \U$ such that $t_n \in J_1$ and $\sum_{t_i \in J_1}t_i \leq T$ and $\sum_{t_j \in J_2}t_j \leq T$ with $J_2 = \U \setminus J_1$, i.e. both machines finish in time $T$.
    	\end{description}
    \end{mdframed}
\end{samepage}
Again, we demand w.l.o.g. the last element to be in the solution for the first machine as in \textsc{Partition}.
With this, we break the symmetry of solutions, which is not compatible with the SSP framework in the following reduction.
The reduction from \textsc{Partition} to \textsc{Two-Machine-Scheduling} is a folklore reduction, which exploits the equivalence of the problems and is therefore easy-to-see.
For this let $I = \{a_1, \ldots, a_n\}$ be the \textsc{Partition} instance.
We transform each number $a_i$ in the \textsc{Partition} instance to a job with processing time $a_i$ in the \textsc{Two-Machine-Scheduling} instance and set the threshold $T = \frac{1}{2} \sum_i a_i$.

Because for both sets $J_1$ and $J_2$ holds $\sum_{a_i \in J_1}a_i \leq T$ and $\sum_{a_j \in J_2}a_j \leq T$ as well as that $\sum_{a_i \in J_1}a_i + \sum_{a_j \in J_2}a_j = 2T$.
We can transform the constraints above to the equivalent constraints $\sum_{t_i \in J_1}t_i = T$ and $\sum_{t_j \in J_2}t_j = T$.
Therefore, we can interpret the two sets of the partition as the two machines in \textsc{Two-Machine-Scheduling} and have a direct one-to-one correspondence between the solutions with $f_I(a_i) = a_i$.
Thus, this reduction is an SSP reduction.

\begin{samepage}
    \begin{mdframed}
    	\begin{description}
        \item[]\textsc{Directed Hamiltonian Path}\hfill\\
        \textbf{Instances:} Directed Graph $G = (V, A)$, Vertices $s, t \in V$.\\
        \textbf{Universe:} Arc set $A =: \U$.\\
        \textbf{Solution set:} The set of all sets $C \subseteq A$ forming a Hamiltonian path going from $s$ to $t$.
    	\end{description}
    \end{mdframed}
\end{samepage}
The reduction from \textsc{3Sat} to \textsc{Directed Hamiltonian Path} from Arora and Barak \cite{DBLP:books/daglib/0023084} is an SSP reduction.
Let $I = (L, C)$ be the \textsc{3Sat} instance that we transform to the \textsc{Directed Hamiltonian Cycle} instance $G = (V, A)$.
An example of the transformation can be found in \Cref{fig:reduction:3sat-directed-hamiltonian-path}.
First, we introduce two additional vertices $s, t \in V$.
For each literal pair $(\ell_i, \overline \ell_i)$ (or variable $x_i$), we introduce a path with $4|C|$ vertices, where we denote the vertices along the path with $v^1_i, \ldots, v^{4|C|}_i$.
The path is directed in both ways such that $v^1_i$ is reachable from $v^{4|C|}_i$ and vice versa.
The direction from $v^1_i$ to $v^{4|C|}_i$ encodes that $\ell$ is taken into the solution and the other direction encodes that $\overline \ell$ is taken into the solution.
Additionally, we add arcs $(s, v^1_1)$ and $(s, v^{4|C|}_1)$, $(v^1_{|L|/2}, t)$ and $(v^{4|C|}_{|L|/2}, t)$ as well as $(v^1_i, v^{4|C|}_{i+1})$ and $(v^{4|C|}_{i}, v^1_{i+1})$ for all $i \in \fromto{1}{{|L|/2}-1}$.
At last, we need to simulate the clauses.
For this, we add a vertex for each clause $c_j \in C$ and connect them to the variable paths by introducing two arcs for each literal $\ell_i$ in the clause $c_j$.
If $\ell_i$ is the non-negated literal of variable $x_i$, then we add the arcs $(x^{4j-1}, c_j)$ and $(c_j, x^{4j-2})$.
Otherwise, we add the arcs $(x^{4j-2}, c_j)$ and $(c_j, x^{4j-1})$.
Thus, one can satisfy the clause $c_j$, i.e. traveling over the vertex $c_j$, if and only if by traveling in the correct direction, i.e. whenever a literal in the clause is taken into the solution.
Overall, a Hamiltonian path from $s$ to $t$ includes all vertices that is all clause vertices, i.e. all clauses are satisfied, and all vertices defined by literals.
Consequently, each variable is assigned a value by the direction of the taken path.

\tikzstyle{vertex}=[draw,circle,fill=black, minimum size=4pt,inner sep=0pt]
\tikzstyle{edge} = [draw,-]
\tikzstyle{arc} = [draw,->]
\tikzstyle{doublearc} = [draw,<->]
\begin{figure}[thpb]
\centering
\begin{tikzpicture}[scale=1,auto]

\node[vertex] (s) at (0,5) {}; \node[above] at (s) {$s$};
\node[vertex] (t) at (0,1) {}; \node[below] at (t) {$t$};

\node[] (x1s) at (-2.5,4) {}; \node[right] at (x1s) {$x_1$};
\node[vertex] (x11) at ($(x1s) + (1,0)$) {};
\node[vertex] (x12) at ($(x1s) + (2,0)$) {};
\node[vertex] (x13) at ($(x1s) + (3,0)$) {};
\node[vertex] (x14) at ($(x1s) + (4,0)$) {};
\draw[doublearc] (x11) to (x12);
\draw[doublearc] (x12) to (x13);
\draw[doublearc] (x13) to (x14);

\node[] (x2s) at (-2.5,3) {}; \node[right] at (x2s) {$x_2$};
\node[vertex] (x21) at ($(x2s) + (1,0)$) {};
\node[vertex] (x22) at ($(x2s) + (2,0)$) {};
\node[vertex] (x23) at ($(x2s) + (3,0)$) {};
\node[vertex] (x24) at ($(x2s) + (4,0)$) {};
\draw[doublearc] (x21) to (x22);
\draw[doublearc] (x22) to (x23);
\draw[doublearc] (x23) to (x24);

\node[] (x3s) at (-2.5,2) {}; \node[right] at (x3s) {$x_3$};
\node[vertex] (x31) at ($(x3s) + (1,0)$) {};
\node[vertex] (x32) at ($(x3s) + (2,0)$) {};
\node[vertex] (x33) at ($(x3s) + (3,0)$) {};
\node[vertex] (x34) at ($(x3s) + (4,0)$) {};
\draw[doublearc] (x31) to (x32);
\draw[doublearc] (x32) to (x33);
\draw[doublearc] (x33) to (x34);

\draw[arc] (s) to (x11);
\draw[arc] (s) to (x14);
\draw[arc] (x11) to (x21);
\draw[arc] (x11) to (x24);
\draw[arc] (x14) to (x21);
\draw[arc] (x14) to (x24);
\draw[arc] (x21) to (x31);
\draw[arc] (x21) to (x34);
\draw[arc] (x24) to (x31);
\draw[arc] (x24) to (x34);
\draw[arc] (x31) to (t);
\draw[arc] (x34) to (t);

\node[vertex] (c1) at (4,3) {}; \node[right] at (c1) {$c_1 = \overline \ell_1 \lor \overline \ell_2 \lor \ell_3$};
\draw[arc, bend left] (x13) to (c1);
\draw[arc, bend right] (c1) to (x12);
\draw[arc, bend left] (x23) to (c1);
\draw[arc, bend right] (c1) to (x22);
\draw[arc, bend right] (x32) to (c1);
\draw[arc, bend left] (c1) to (x33);

\end{tikzpicture}
\caption{Classical reduction of \textsc{3Sat} to \textsc{Directed Hamiltonian Path} for $\varphi = (\overline \ell_1 \wedge \overline \ell_2 \wedge \ell_3)$.}
\label{fig:reduction:3sat-directed-hamiltonian-path}
\end{figure}

For this reduction, it is not directly obvious, how we find a one-to-one correspondence, because a whole path corresponds to one literal.
However, we can use exactly one arc of that path to act as representative.
We define the function $f_I$ by $f_I(\ell_i) = (x^1_i, x^2_i)$ and $f_I(\overline \ell_i) = (x^2_i, x^1_i)$.
Thus, we have a one-to-one correspondence between the literals and the arcs.
Furthermore by correctness of the reduction, we have a one-to-one correspondence between the solutions to $(L, C)$ and $G = (V, A)$ by Hamiltonian path using the either one of the arcs $(x^1_i, x^2_i)$ and $(x^2_i, x^1_i)$ for each $i \in \fromto{1}{|L|/2}$ and including each clause vertex $c_j$ for $j \in \fromto{1}{|C|}$.

\begin{samepage}
    \begin{mdframed}
    	\begin{description}
        \item[]\textsc{Directed Hamiltonian Cycle}\hfill\\
        \textbf{Instances:} Directed Graph $G = (V, A)$.\\
        \textbf{Universe:} Arc set $A =: \U$.\\
        \textbf{Solution set:} The set of all sets $C \subseteq A$ forming a Hamiltonian cycle.
    	\end{description}
    \end{mdframed}
\end{samepage}
We extend the reduction from \textsc{3Sat} to \textsc{Directed Hamiltonian Cycle} from Arora and Barak \cite{DBLP:books/daglib/0023084} by simply adding an arc from $t$ to $s$.
Obviously, all possible cycles have to go through $s$ and $t$.
This has no influence on the rest of the reduction, especially on the solutions and the one-to-one correspondence of the literals and the arcs.
Consequently, the reduction is still an SSP reduction.

\begin{samepage}
    \begin{mdframed}
    	\begin{description}  
        \item[]\textsc{Undirected Hamiltonian Cycle}\hfill\\
        \textbf{Instances:} Graph $G = (V, E)$.\\
        \textbf{Universe:} Edge set $E =: \U$.\\
        \textbf{Solution set:} The set of all sets $C \subseteq E$ forming a Hamiltonian cycle.
    	\end{description}
    \end{mdframed}
\end{samepage}
Karp's reduction \cite{DBLP:conf/coco/Karp72} from \textsc{Directed Hamiltonian Cycle} to \textsc{Undirected Hamiltonian Cycle} is an SSP reduction.
Let $I = (V, A)$ be the \textsc{Directed Hamiltonian Cycle} and $(V', E')$ be the \textsc{Undirected Hamiltonian Cycle} instance.
The reduction replaces each vertex $v$ with three vertices $v'_{in}, v', v'_{out}$ and adds edges $\{v'_{in}, v'\}, \{v', v'_{out}\}$ to connect the three vertices to a path.
All arcs $(v, w) \in A$ are replaced by one edge $\{v'_{out}, w'_{in}\}$ essentially preserving the one-to-one correspondence between the elements to the corresponding unique edge, i.e. $f_I(v,w) = \{v'_{out}, w'_{in}\}$.
The solutions are also preserved because no additional solutions are added and all original solutions are preserved (every Hamiltonian cycle has to run through each $v'_{in}, v', v'_{out}$ exactly once in the specified order).

\begin{samepage}
    \begin{mdframed}
    	\begin{description}
        \item[]\textsc{Traveling Salesman Problem}\hfill\\
        \textbf{Instances:} Complete Graph $G = (V, E)$, weight function $w: E \rightarrow \Z $, number $k \in \N$.\\
        \textbf{Universe:} Edge set $E =: \U$.\\
        \textbf{Feasible solution set:} The set of all TSP tours $T\subseteq E$.\\
        \textbf{Solution set:} The set of feasible $T$ with $w(T) \leq k$.
    	\end{description}
    \end{mdframed}
\end{samepage}
There is an easy-to-see folklore reduction from \textsc{Undirected Hamiltonian Cycle} to \textsc{Traveling Salesman Problem}, which is an SSP reduction.
Let $I = (V, E)$ be the \textsc{Undirected Hamiltonian Cycle} instance and $(V', E', w', k')$ the \textsc{Traveling Salesman Problem} instance.
Every vertex $v \in V$ is mapped to itself $v \in V'$.
Furthermore, we map each edge $e \in E$ to itself in $E'$ and add additional edges to form a complete graph.
The weight function $w': E' \rightarrow \Z$ is defined for all $e' \in E'$ as
$$
    w(e') = \begin{cases}
        0, \quad \text{if} \ e' \in E\\
        1, \quad \text{if} \ e' \notin E
    \end{cases}
$$
At last, we set $k' = 0$ resulting that only the edges from $E$ are usable.
Thus, we preserve the one-to-one correspondence between the edges with $f_I(e) = e$.
Consequently, this is an SSP reduction.

\begin{samepage}
    \begin{mdframed}
    	\begin{description}
        \item[]\textsc{Directed Two Vertex Disjoint Path}\hfill\\
        \textbf{Instances:} Directed graph $G = (V, A)$, $s_i, t_i \in V$ for $i \in \{1, 2\}$.\\
        \textbf{Universe:} Arc set $A =: \U$.\\
        \textbf{Solution set:} The set of all sets set $A' \subseteq A$ such that $A' = A(P_1) \cup A(P_2)$, where $P_1$ and $P_2$ are some vertex-disjoint paths s.t. $P_i$ goes from $s_i$ to $t_i$ for $i \in \set{1,2}$. 
    	\end{description}
    \end{mdframed}
\end{samepage}
The reduction by Fortune, Hopcroft and Wyllie \cite{DBLP:journals/tcs/FortuneHW80} is an SSP reduction.
The reduction makes extensive use of a switch gadget, which is depicted in \Cref{fig:reduction:3sat-directed-two-disjoint-path:switch}.
The gadget has four input arcs, $B,C,W$ and $Y$, and four output arcs, $A,D,X$ and $Z$.
The idea of this switch gadget is that if you have two disjoint paths going through the gadget entering at $B$ and $C$, then the path entering at $B$ must leave at $D$ and the one entering at $C$ must leave at $A$, and additionally either a path from $W$ to $X$ exists or a path from $Y$ to $Z$ exists.
We can then use the first of the two paths to run first through the switches and then to the rest of the construction and the second path to run through the switches as in visualized in \Cref{fig:reduction:3sat-directed-two-disjoint-path}.
In doing so, the second path running only through the switches controls that the first path running through the construction is only able to satisfy the clauses according to the assignment of the variables.

\tikzstyle{vertex}=[draw,circle,fill=black, minimum size=4pt,inner sep=0pt]
\tikzstyle{edge} = [draw,-]
\tikzstyle{arc} = [draw,->]
\tikzstyle{doublearc} = [draw,<->]
\begin{figure}[thpb]
\centering
\begin{tikzpicture}[scale=0.5,auto]

\node[vertex] (c) at (0,0) {};
\node[vertex] (a) at (0,-6) {};
\draw[arc] ($(c) + (0,1)$) to node[left] {$C$} (c);
\draw[arc] (a) to node[left] {$A$} ($(a) - (0,1)$);

\node[vertex] (01) at ($(c) + (-1,-1)$) {};
\node[vertex] (02) at ($(c) + (-1,-2)$) {};
\node[vertex] (03) at ($(c) + (-1,-3)$) {};
\node[vertex] (04) at ($(c) + (-1,-4)$) {};
\node[vertex] (05) at ($(c) + (-1,-5)$) {};
\draw[arc] (c) to (01);
\draw[arc] (01) to (02);
\draw[arc] (02) to (03);
\draw[arc] (03) to (04);
\draw[arc] (04) to (05);
\draw[arc] (05) to (a);

\node[vertex] (11) at ($(c) + (1,-1)$) {};
\node[vertex] (12) at ($(c) + (1,-2)$) {};
\node[vertex] (13) at ($(c) + (1,-3)$) {};
\node[vertex] (14) at ($(c) + (1,-4)$) {};
\node[vertex] (15) at ($(c) + (1,-5)$) {};
\draw[arc] (c) to (11);
\draw[arc] (11) to (12);
\draw[arc] (12) to (13);
\draw[arc] (13) to (14);
\draw[arc] (14) to (15);
\draw[arc] (15) to (a);

\node[vertex] (08) at (-4,-1) {};
\node[vertex] (09) at (-3,-1) {};
\node[vertex] (00) at (-2,-1) {};
\draw[arc] ($(08) - (1,0)$) to node[above] {$W$} (08);
\draw[arc] (08) to (09);
\draw[arc] (09) to (00);
\draw[arc] (00) to (01);
\draw[arc, out=135, in=270, looseness=0.5] (15) to (09);

\node[vertex] (18) at (4,-1) {};
\node[vertex] (19) at (3,-1) {};
\node[vertex] (10) at (2,-1) {};
\draw[arc] ($(18) + (1,0)$) to node[above] {$Y$} (18);
\draw[arc] (18) to (19);
\draw[arc] (19) to (10);
\draw[arc] (10) to (11);
\draw[arc, out=45, in=270, looseness=0.5] (05) to (19);

\node[vertex] (b) at (2,-6) {};
\draw[arc] (b) to (04);
\draw[arc] (b) to (14);
\draw[arc] ($(b) - (0,1)$) to node[right] {$B$} (b);

\node[vertex] (d) at (2,0) {};
\draw[arc] (00) to (d);
\draw[arc] (10) to (d);
\draw[arc] (d) to node[right] {$D$} ($(d) + (0,1)$);

\node[vertex] (011) at (-4,-2) {};
\draw[arc] (02) to (011);
\draw[arc] (011) to node[below] {$X$} ($(011) - (1,0)$);

\node[vertex] (111) at (4,-2) {};
\draw[arc] (12) to (111);
\draw[arc] (111) to node[below] {$Z$} ($(111) + (1,0)$);

\end{tikzpicture}
\caption{The switch gadget.}
\label{fig:reduction:3sat-directed-two-disjoint-path:switch}
\end{figure}

Let $I = (L, C)$ be the \textsc{3Sat} instance and $(V, A, s_1, t_1, s_2, t_2)$ be the \textsc{Directed Two Disjoint Path} instance.
First, we introduce four vertices $s_1, t_1, s_2, t_2$ representing the start and ends of the two disjoint paths.
For every literal $\ell \in  L$, we create a path $\ell^1, \ldots, \ell^{4|C|}$ of $4|C|$ vertices.
For every literal pair $\ell_i, \overline \ell_i$ (or variable $x_i$), we connect the paths by introducing two additional vertices $x^s_i$ and $x^t_i$ with arcs $(x^s_i, \ell^1_i), (x^s_i, \overline \ell^1_i)$ and $(\ell^{4|C|}_i, x^t_i), (\overline \ell^{4|C|}_i, x^t_i)$.
We connect the literal gadgets by adding the arcs $(x^t_i, x^t_{i+1})$ for all $i \in \fromto{1}{|L|/2-1}$.
For each clause $c_j \in C$, we add two vertices $c^1_j$ and $c^2_j$ and connect them by three arcs.
We connect these clause vertices with the arcs $(c^2_{|C|}, t_1)$ as well as $(c^2_j, c^1_{j+1})$ for $j \in \fromto{1}{|C|-1}$.
At last, we connect the literal paths with the clause path with an arc $(x^t_{|L|/2}, c^1_1)$.

Now, we have everything to introduce the switches into the construction.
We stack the switches one after another by merging the $C$ and $D$ input arcs and the $A$ and $B$ input arcs, respectively.
Then, we connect $s_2$ to the input arc $C$ of the last switch in the stack and $t_2$ to the output arc $A$ of the first switch of the complete switch stack.
We do this analogously for $s_1$, which we connect to the input arc $B$ of the first switch of the stack and the rest of the construction with the output arc $D$ of the last switch of the stack.
Thus both path run through the switch stack as described above.
At last, we use the switches to check whether the \textsc{3Sat} assignment is correct.
For this, we use the schematic description of a switch as depicted in \Cref{fig:reduction:3sat-directed-two-disjoint-path:switch:schema}.

\tikzstyle{vertex}=[draw,circle,fill=black, minimum size=4pt,inner sep=0pt]
\tikzstyle{edge} = [draw,-]
\tikzstyle{arc} = [draw,->]
\tikzstyle{doublearc} = [draw,<->]
\begin{figure}[thpb]
\centering
\begin{tikzpicture}[scale=0.5,auto]

\node[vertex, label=above:{$W$}] (w) at (0,0) {};
\node[vertex, label=below:{$X$}] (x) at (0,-2) {};
\node[vertex, label=above:{$Y$}] (y) at (2,0) {};
\node[vertex, label=below:{$Z$}] (z) at (2,-2) {};

\draw[arc] (w) to (x);
\draw[arc] (y) to (z);
\draw[edge] (0,-1) to (2,-1);

\end{tikzpicture}
\caption{The schematic switch gadget.}
\label{fig:reduction:3sat-directed-two-disjoint-path:switch:schema}
\end{figure}

That is, the arc $(\overline \ell^{4j-2}_i, \overline \ell^{4j-1}_i)$ is connected to the arc and $(c^1_j, c^2_j)$ if and only if the corresponding literal $\ell_i \in L$ is in clause $c_j \in C$.
More precisely, the arc $(\overline \ell^{4j-2}_i, \overline \ell^{4j-1}_i)$ is substituted by using the input arc $W$ from $\overline \ell^{4j-2}_i$ and output arc $X$ to $\overline \ell^{4j-1}_i$ and for the clause vertices $c^1_j$ is incident to input arc $Y$ and output arc $Z$ is incident to $c^2_j$.
Because only one path either from $W$ to $X$ or from $Y$ to $Z$ is usable, the path has to go from $s_1$ through the switch stack, then over the literal paths of the literals that are in the solution and at last over the clause vertices to $t_1$.
If for a clause there is no literal satisfying it, the path in the switch is blocked.
The full construction can be found in \Cref{fig:reduction:3sat-directed-two-disjoint-path}.

\tikzstyle{vertex}=[draw,circle,fill=black, minimum size=4pt,inner sep=0pt]
\tikzstyle{edge} = [draw,-]
\tikzstyle{arc} = [draw,->]
\tikzstyle{doublearc} = [draw,<->]
\begin{figure}[thpb]
\centering
\begin{tikzpicture}[scale=0.5,auto]

\node[vertex] (x1s) at (1,0) {};
\node[vertex] (x101) at ($(x1s) + (1,1)$) {};
\node[vertex] (x102) at ($(x1s) + (2,1)$) {};
\node[vertex] (x103) at ($(x1s) + (3,1)$) {};
\node[vertex] (x104) at ($(x1s) + (4,1)$) {}; \node[above left] at (x101) {$\overline \ell_1$};
\node[vertex] (x111) at ($(x1s) + (1,-1)$) {};
\node[vertex] (x112) at ($(x1s) + (2,-1)$) {};
\node[vertex] (x113) at ($(x1s) + (3,-1)$) {};
\node[vertex] (x114) at ($(x1s) + (4,-1)$) {}; \node[below left] at (x111) {$\ell_1$};
\node[vertex] (x1t) at ($(x1s) + (5,0)$) {};
\draw[arc] (x1s) to (x101);
\draw[arc] (x101) to (x102);
\draw[arc] (x102) to (x103);
\draw[arc] (x103) to (x104);
\draw[arc] (x104) to (x1t);
\draw[arc] (x1s) to (x111);
\draw[arc] (x111) to (x112);
\draw[arc] (x112) to (x113);
\draw[arc] (x113) to (x114);
\draw[arc] (x114) to (x1t);

\node[vertex] (x2s) at (7,0) {};
\node[vertex] (x201) at ($(x2s) + (1,1)$) {};
\node[vertex] (x202) at ($(x2s) + (2,1)$) {};
\node[vertex] (x203) at ($(x2s) + (3,1)$) {};
\node[vertex] (x204) at ($(x2s) + (4,1)$) {}; \node[above left] at (x201) {$\overline \ell_2$};
\node[vertex] (x211) at ($(x2s) + (1,-1)$) {};
\node[vertex] (x212) at ($(x2s) + (2,-1)$) {};
\node[vertex] (x213) at ($(x2s) + (3,-1)$) {};
\node[vertex] (x214) at ($(x2s) + (4,-1)$) {}; \node[below left] at (x211) {$\ell_2$};
\node[vertex] (x2t) at ($(x2s) + (5,0)$) {};
\draw[arc] (x2s) to (x201);
\draw[arc] (x201) to (x202);
\draw[arc] (x202) to (x203);
\draw[arc] (x203) to (x204);
\draw[arc] (x204) to (x2t);
\draw[arc] (x2s) to (x211);
\draw[arc] (x211) to (x212);
\draw[arc] (x212) to (x213);
\draw[arc] (x213) to (x214);
\draw[arc] (x214) to (x2t);

\node[vertex] (x3s) at (13,0) {};
\node[vertex] (x301) at ($(x3s) + (1,1)$) {};
\node[vertex] (x302) at ($(x3s) + (2,1)$) {};
\node[vertex] (x303) at ($(x3s) + (3,1)$) {};
\node[vertex] (x304) at ($(x3s) + (4,1)$) {}; \node[above left] at (x301) {$\overline \ell_3$};
\node[vertex] (x311) at ($(x3s) + (1,-1)$) {};
\node[vertex] (x312) at ($(x3s) + (2,-1)$) {};
\node[vertex] (x313) at ($(x3s) + (3,-1)$) {};
\node[vertex] (x314) at ($(x3s) + (4,-1)$) {}; \node[below left] at (x311) {$\ell_3$};
\node[vertex] (x3t) at ($(x3s) + (5,0)$) {};
\draw[arc] (x3s) to (x301);
\draw[arc] (x301) to (x302);
\draw[arc] (x302) to (x303);
\draw[arc] (x303) to (x304);
\draw[arc] (x304) to (x3t);
\draw[arc] (x3s) to (x311);
\draw[arc] (x311) to (x312);
\draw[arc] (x312) to (x313);
\draw[arc] (x313) to (x314);
\draw[arc] (x314) to (x3t);

\node[vertex] (s1) at (-3,0) {}; \node[above left] at (s1) {$s_1$};
\node[vertex] (b0) at (-2,0) {}; \node[above] at (b0) {$B_1$};
\node[vertex] (d1) at (0,0) {}; \node[below] at (d1) {$D_{\ell\textit{ast}}$};
\node[vertex] (t11) at (19,0) {}; \node[above right] at (t11) {};
\node[vertex] (t1) at (0,-4) {}; \node[above left] at (t1) {$t_1$};
\node[vertex] (s11) at (19,-4) {}; \node[above right] at (s11) {};
\draw[arc] (s1) to (b0);
\draw[arc] (d1) to (x1s);
\draw[arc] (x1t) to (x2s);
\draw[arc] (x2t) to (x3s);
\draw[arc] (x3t) to (t11);
\draw[arc, out=0, in=0, looseness=1] (t11) to (s11);

\node[vertex] (s2) at (6.5,4) {}; \node[above] at (s2) {$s_2$};
\node[vertex] (t2) at (12.5,4) {}; \node[above] at (t2) {$t_2$};
\node[vertex] (c1) at (8,4) {}; \node[below] at (c1) {$C_{\ell\textit{ast}}$};
\node[vertex] (a0) at (11,4) {}; \node[below] at (a0) {$A_1$};
\draw[arc] (s2) to (c1);
\draw[arc] (a0) to (t2);

\node[vertex] (c11) at (8,-4) {}; \node[below] at (c11) {$c^1_1$};
\node[below] (c1) at (9.5,-5) {$c_1$};
\node[vertex] (c12) at (11,-4) {}; \node[below] at (c12) {$c^2_1$};
\draw[arc] (s11) to (c12);
\draw[arc] (c11) to (t1);

\draw[arc] (c12) to (c11);
\draw[arc, in=45, out=135, looseness=0.7] (c12) to (c11);
\draw[arc, in=315, out=225, looseness=0.7] (c12) to (c11);

\draw[edge] (9.5,-4) to ($(x112)+(0.4,0)$);
\draw[edge] (9.5,-3.5) to ($(x212)+(0.4,0)$);
\draw[edge] (9.5,-4.5) to ($(x302)+(0.4,0)$);

\end{tikzpicture}
\caption{Classical reduction of \textsc{3Sat} to \textsc{Directed Two Disjoint Path} for $\varphi = (\overline \ell_1 \wedge \overline \ell_2 \wedge \ell_3)$.}
\label{fig:reduction:3sat-directed-two-disjoint-path}
\end{figure}

There is a one-to-one correspondence between the literals and the arcs of the path from $s_1$ to $t_1$.
We can define $f_I(\ell_i) = (x^s_i, \ell^1_i)$ because the path over $(x^s_i, \ell^1_i)$ is taken if and only if $\ell_i$ is in the \textsc{3Sat} solution.

\begin{samepage}
    \begin{mdframed}
    	\begin{description}
        \item[]\textsc{Directed} $k$-\textsc{Vertex Disjoint Path}\hfill\\
        \textbf{Instances:} Directed graph $G = (V, A)$, $s_i, t_i \in V$ for $i \in \fromto{1}{k}$.\\
        \textbf{Universe:} Arc set $A =: \U$.\\
        \textbf{Solution set:} The sets of all sets $A' \subseteq A$ such that $A' = \bigcup^k_{i = 1} A(P_i)$, where all $P_i$ are pairwise vertex-disjoint paths from $s_i$ to $t_i$ for $1 \leq i \leq k$.
    	\end{description}
    \end{mdframed}
\end{samepage}
The following reduction from \textsc{Directed Two Vertex Disjoint Path} to \textsc{Directed} $k$-\textsc{Vertex Disjoint Path} is an easy-to-see SSP reduction.
We introduce $k-2$ additional vertex pairs $s_i, t_i$ for $i \in \fromto{3}{k}$, which we connect by adding arcs $(s_i, t_i)$ for all $i \in \fromto{3}{k}$.
Thus, the original \textsc{Directed Two Vertex Disjoint Path} reduction still works for itself, while we added the necessary additional paths.
The SSP properties of the \textsc{Directed Two Vertex Disjoint Path} reduction are obviously not compromised.

\begin{samepage}
    \begin{mdframed}
    	\begin{description}
        \item[]\textsc{Steiner Tree}\hfill\\
        \textbf{Instances:} Undirected graph $G = (S \cup T, E)$, set of Steiner vertices $S$, set of terminal vertices $T$, edge weights $c: E \rightarrow \N$, number $k \in \N$.\\
        \textbf{Universe:} Edge set $E =: \U$.\\
        \textbf{Feasible solution set:} The set of all sets $E' \subseteq E$ such that $E'$ is a tree connecting all terminal vertices from $T$.\\
        \textbf{Solution set:} The set of feasible solutions $E'$ with $\sum_{e' \in E'} c(e') \leq k$.
    	\end{description}
    \end{mdframed}
\end{samepage}
There is a folklore reduction from \textsc{3Sat} to \textsc{Steiner Tree}, which is an SSP reduction, and which is depicted in \Cref{fig:reduction:3sat-steiner-tree}.
First, there are designated terminal vertices $s$ and $t$.
For every literal $\ell \in L$, there is a Steiner vertex $\ell$.
Additionally for every literal pair $(\ell_i, \overline \ell_i)$, $1 \leq i \leq |L|/2-1$, we add a Steiner vertex $v_i$. We define $v_0 := s$.
Then all of the above vertices are connected into a \enquote{diamond chain}, where we begin with $s$ connected to both $\ell_1$ and $\overline \ell_1$.
Both Vertices $\ell_1$ and $\overline \ell_1$ are connected to $v_1$.
This vertex $v_1$ is then connected to vertices $\ell_2$ and $\overline \ell_2$ and so on.
At last, $\ell_{|L|}$ and $\overline \ell_{|L|}$ are connected to $t$.

Furthermore, for every clause $c_j \in C$, we add a corresponding terminal vertex $c_j$.
The vertex $c_j$ is then connected its corresponding literals $\ell \in c_j$ via a path of Steiner vertices of length $|L| + 1$.
The costs of every edge is set to $1$ and the threshold is set to $k = |L| + |C| \cdot (|L| + 1)$.

For the correctness, observe that every solution of Steiner tree includes a path from $s$ to $t$ over the literal vertices because all paths over a clause vertex are longer than $|L|$.
This path encodes the set of literals included in a corresponding \textsc{3Sat} solution, where a positive literal $\ell_i$ for $1 \leq i \leq |L|/2$ is included in the \textsc{3Sat} solution if and only if the edge $\{v_{i-1}, \ell_i\}$ is in the Steiner tree solution.
The analogous statement holds for negative literals.
We therefore define the embedding function $f_I$ in the above fashion, i.e. for all $\ell \in L$ we have $f_I(\ell) = \set{v_{i-1}, \ell}$.

Next, for every clause $c$, the path from literal $\ell$ to terminal vertex $c$ for one $\ell \in c$ is included in the solution as well.
Thus, $|C|$ paths of length $|L|+1$ are included.
If a clause $c_j$ is not satisfied, then a path of length of at least $|L|+2$ is needed to include the terminal vertex $c_j$, which violates the threshold.
Thus, the reduction is correct.

The SSP property holds, because every correct \textsc{3Sat} solution can be translated with the function $f_I$ and be completed to a Steiner tree with at most $k$ edges.
On the other hand, every Steiner tree with at most $k$ edges restricted to the set $f_I(L)$ encodes a \textsc{3Sat} solution.

\tikzstyle{vertex}=[draw,circle,fill=black, minimum size=4pt,inner sep=0pt]
\tikzstyle{terminal}=[draw,rectangle,fill=black, minimum size=4pt,inner sep=0pt]
\tikzstyle{edge} = [draw,-]
\begin{figure}[thpb]
\centering
\begin{tikzpicture}[scale=1,auto]

\node[terminal] (s) at (0,0) {}; \node[left] at (s) {$s$};
\node[vertex] (x12) at ($(s) + (2,0)$) {};
\node[vertex] (x23) at ($(s) + (4,0)$) {};
\node[terminal] (t) at ($(s) + (6,0)$) {}; \node[right] at (t) {$t$};

\node[vertex] (x1) at ($(s) + (1,0.5)$) {}; \node[above] at (x1) {$\ell_1$};
\node[vertex] (notx1) at ($(s) + (1,-0.5)$) {}; \node[above] at (notx1) {$\overline \ell_1$};

\node[vertex] (x2) at ($(x12) + (1,0.5)$) {}; \node[above] at (x2) {$\ell_2$};
\node[vertex] (notx2) at ($(x12) + (1,-0.5)$) {}; \node[above] at (notx2) {$\overline \ell_2$};

\node[vertex] (x3) at ($(x23) + (1,0.5)$) {}; \node[above] at (x3) {$\ell_3$};
\node[vertex] (notx3) at ($(x23) + (1,-0.5)$) {}; \node[above] at (notx3) {$\overline \ell_3$};

\draw[edge] (s) to (x1);
\draw[edge] (s) to (notx1);
\draw[edge] (x1) to (x12);
\draw[edge] (notx1) to (x12);
\draw[edge] (x12) to (x2);
\draw[edge] (x12) to (notx2);
\draw[edge] (x2) to (x23);
\draw[edge] (notx2) to (x23);
\draw[edge] (x23) to (x3);
\draw[edge] (x23) to (notx3);
\draw[edge] (x3) to (t);
\draw[edge] (notx3) to (t);

\node[terminal] (c1) at (3,-2.5) {}; \node[below] at (c1) {$c_1 = \overline \ell_1 \lor \overline \ell_2 \lor \ell_3$};

\node[vertex] (11) at ($(notx1) + (0.67,-0.5)$) {};
\node (12) at ($(11) + (0,-0.4)$) {$\vdots$};
\node[vertex] (13) at ($(11) + (0,-1)$) {};

\node[vertex] (21) at ($(notx2) + (0,-0.5)$) {};
\node (22) at ($(21) + (0,-0.4)$) {$\vdots$};
\node[vertex] (23) at ($(21) + (0,-1)$) {};

\node[vertex] (31) at ($(notx3) + (-0.67,-0.5)$) {};
\node (32) at ($(31) + (0,-0.4)$) {$\vdots$};
\node[vertex] (33) at ($(31) + (0,-1)$) {};

\draw[edge] (notx1) to (11);
\draw[edge] (11) to ($(12) + (0,0.18)$);
\draw[edge] (12) to (13);
\draw[edge] (13) to (c1);

\draw[edge] (notx2) to (21);
\draw[edge] (21) to ($(22) + (0,0.18)$);
\draw[edge] (22) to (23);
\draw[edge] (23) to (c1);

\draw[edge] (x3) to (31);
\draw[edge] (31) to ($(32) + (0,0.18)$);
\draw[edge] (32) to (33);
\draw[edge] (33) to (c1);

\node at ($(s)+(-1,0)$) {$W$};
\draw[dashed,rounded corners] ($(s)+(-0.5,+1)$) rectangle ($(t) + (0.5,-0.67)$);

\end{tikzpicture}
\caption{Classical reduction of \textsc{3Sat} to \textsc{Steiner Tree}.}
\label{fig:reduction:3sat-steiner-tree}
\end{figure}

\section{Conclusion}
We have shown in this paper that for a large amount of NP-complete problems, their min-max variant (min-max-min variant, respectively) is automatically $\Sigma^p_2$-complete ($\Sigma^p_3$-complete, respectively).
We have showcased this behavior in the areas of network interdiction, min-max regret robust optimization and two-stage adjustable robust optimization.
Our findings constitute a leap in the understanding of the basic behavior of such problems.

However, many questions still remain unanswered.
First, we would like to understand if our theorem can also be adapted to work with other popular areas of robust or multi-level optimization, for example to recoverable robust optimization, to Stackelberg games, to attacker-defender games, or to the area of computational social choice.
Secondly, we restricted our attention in this article only to such problems which can be expressed as finding a certain subset with a nice property.
It would be insightful to understand, if similar meta-theorems can be made about natural variants of problems which look for a nice partition, a nice assignment function, a nice permutation, etc.
For space reasons, we restricted our focus in this paper mainly on multi-stage problems with two or three stages.
It seems natural to extend our arguments to multi-stage problems with an arbitrary number of stages.

Since all three main results of our paper are proven in an essentially analogous way, it seems intriguing to consider a potential meta-meta-theorem. Which properties does a min-max \emph{modification scheme} (such as the interdiction-modification, the regret-modification, the two-stage modification) need to possess, such that a similar meta-theorem applies?

Finally, our framework only applies to nominal problems which are NP-complete in the first place.
However, researchers in the area often are interested in robust variants of nominal problems in P.
Our framework can not say anything about these problems -- for NP-complete problems a vast amount of existing completeness reductions between them exist, which are upgraded by our meta-theorem.
Is there some notion of reduction between problems in P, which supports our framework (i.e., for problems admitting such reductions, their robust variant is automatically NP-hard)?

\newpage

\bibliography{bib_general,bib_interdiction,bib_regret,bib_two-stage,bib_recoverable_robust,bib_reductions}

\end{document}